%% file: 2-coh-internal-models.tex
\documentclass{article}

\input{packages-utility}
\input{packages-math}
\input{packages-typesetting}
\input{packages-referencing}

\input{macros-math}

\input{macros-typesetting}
\input{macros-math-paper}

\input{setup-math}
\input{setup-typesetting}
\input{setup-tikzcd-paper}

\title{2-Coherent Internal Models \titlebreak of Homotopical Type Theory}
\authordata{Joshua Chen}{https://orcid.org/0000-0001-5041-0794}{School of Computer Science, University of Nottingham}

\begin{document}

\maketitle

\input{abstract}
\input{intro}
\input{wild-category-theory}
\input{pullbacks}
\input{wild-cwfs}
\input{2-coh-ctx-ext}
\input{ctx-comprehension}
\input{discussion}




\pagebreak
\phantomsection
\printbibliography[heading=bibintoc]

\end{document}

%% file: packages-utility.tex
\usepackage{etoolbox}
\usepackage{refcount}
\usepackage{titlecaps}
\usepackage{xstring}
\usepackage[x11names]{xcolor}

%% file: packages-math.tex
\usepackage{amsmath, mathtools, stmaryrd}

\usepackage{bbm}
\usepackage{extarrows}
\usepackage{scalerel}
\usepackage{upgreek}
\usepackage{wasysym}

\usepackage{amsthm}

\newtheoremstyle{default} 
{.75\baselineskip} 
{.75\baselineskip} 
{} 
{} 
{} 
{} 
{.5em minus 2pt} 
{\textsb{\thmname{#1}}\textbf{\thmnumber{ #2.}}\thmnote{\hspace{.5em minus 2pt}\emph{#3}.}} 

\usepackage{tikz-cd}
\usetikzlibrary{arrows.meta}
\usetikzlibrary{decorations.markings}
\usetikzlibrary{nfold}

%% file: packages-typesetting.tex
\usepackage{graphicx}
\usepackage[htt]{hyphenat} 
\usepackage{microtype}

\usepackage[T1]{fontenc}
\usepackage[p,osf]{cochineal}

\usepackage[cloname=fs,fontsize=11]{fontsize}
\generateclofile{fs}{11}

\newcommand*{\crimsonsb}{\fontfamily{Crimson-TLF}\fontseries{sb}\selectfont} 
\DeclareTextFontCommand{\textsb}{\crimsonsb}

\usepackage[scale=0.88,zerostyle=d]{newtxtt}

\usepackage[libertine,varbb]{newtxmath} 
\usepackage[cal=boondoxo,frak=euler,bb=dsserif]{mathalpha} 
\DeclareMathAlphabet{\mathsf}{T1}{LinuxBiolinumT-OsF}{m}{n} 

\usepackage{adforn}
\usepackage{pifont}

\usepackage{geometry}
\usepackage[indent=1em, skip=0pt]{parskip}

\usepackage[inline]{enumitem}
\setlist{
  topsep=\itemsep,
  parsep=0pt,
  leftmargin=2em,
}
\setlist[enumerate]{
  labelsep=1ex,
}
\setlist[itemize]{
  label={\scriptsize\ding{98}},
  labelsep=.75em,
}

\usepackage{titling}
\pretitle{\vspace{-4\baselineskip}\begin{center}\crimsonsb\Largerr}
\posttitle{\end{center}\vspace{0.25\baselineskip}}
\predate{}
\postdate{}

\renewenvironment{abstract}
  {\begin{quotation}\noindent\textsb{\small Abstract\hspace{.75ex}\textbullet\hspace{.75ex}}\small}
  {\end{quotation}}

\makeatletter
\def\titlebreak{\\}
\makeatother

\usepackage{titlesec}

\titleformat{\section}[block]
  {\centering}{\bfseries\large\thesection.}
  {1ex}{\crimsonsb\larger}

\titleformat{\subsection}
  {}{\bfseries\normalsize\thesubsection.}
  {1ex}{\crimsonsb\normalsizer}

\titlespacing*{\section}{0pt}{*5}{*2}
\titlespacing*{\subsection}{0pt}{*3}{*1}

\titleformat{\paragraph}[runin]
  {}{\bfseries\theparagraph}
  {1ex}{\crimsonsb}[.]

\titlespacing*{\paragraph}{0ex}{\baselineskip}{.5em}

\usepackage{fancyhdr}

\fancypagestyle{plain}{\fancyhead{}\fancyfoot{}}

\usepackage{etoc}

\usepackage{float}
\usepackage{subcaption}



%% file: packages-referencing.tex
\usepackage[hyphens]{url}
\usepackage[pdfusetitle,breaklinks,colorlinks,allcolors=Blue3]{hyperref}

\usepackage[
  citestyle=alphabetic,
  bibstyle=alphabetic,
  maxnames=6,
  maxalphanames=4
  ]{biblatex}

\providecommand*{\mkibid}[1]{#1}
\newbool{cbx:loccit}

\DeclareBibliographyOption{ibidpage}[true]{%
    \ifstrequal{#1}{true}
    {\ExecuteBibliographyOptions{loccittracker=constrict}}
    {\ExecuteBibliographyOptions{loccittracker=false}}}

\DeclareCiteCommand{\cite}[\mkbibbrackets]
    {\usebibmacro{prenote}}
    {\usebibmacro{citeindex}%
        \usebibmacro{cite}}
    {\multicitedelim}
    {\usebibmacro{cite:postnote}}

\renewbibmacro*{cite}{%
    \global\boolfalse{cbx:loccit}%
    \ifthenelse{\ifciteibid\AND\NOT\iffirstonpage}
      {\usebibmacro{cite:ibid}}
      {\printtext[bibhyperref]{%
        \printfield{prefixnumber}%
        \printfield{labelalpha}%
        \printfield{extraalpha}}}}

\newbibmacro*{cite:ibid}{%
    \printtext[bibhyperref]{\bibstring[\mkibid]{ibidem}}%
    \ifloccit
      {\global\booltrue{cbx:loccit}}
      {}}

\newbibmacro*{cite:postnote}{%
    \ifbool{cbx:loccit}
      {}
      {\usebibmacro{postnote}}}

\usepackage[capitalize]{cleveref}
\crefformat{footnote}{#2\footnotemark[#1]#3} 

%% file: macros-math.tex
\newcommand{\textmathcal}[1]{\(\mathcal{#1}\)}
\newcommand{\textmathfrak}[1]{\(\mathfrak{#1}\)}

\newcommand{\eqntextspace}{\hspace{1.5ex}}
\newcommand{\eqntext}[1]{\eqntextspace\text{#1}\eqntextspace}
\newcommand{\eqntextall}{\eqntext{for all}}

\newcommand{\eqnequalspace}{\hspace{1ex}}
\newcommand{\eqnequal}[1]{\eqnequalspace\xlongequal{#1}\eqnequalspace}
\newcommand{\eqnequiv}{\eqnequalspace\equiv\eqnequalspace}

\newcommand{\cC}{\mathcal{C}}

\newcommand{\cE}{\mathcal{E}}

\newcommand{\cU}{\mathcal{U}}

\newcommand{\fI}{\mathfrak{I}}
\newcommand{\fP}{\mathfrak{P}}
\newcommand{\fQ}{\mathfrak{Q}}

\newcommand{\fS}{\mathfrak{S}}
\newcommand{\fX}{\mathfrak{X}}
\newcommand{\fc}{\mathfrak{c}}

\newcommand{\fp}{\mathfrak{p}}
\newcommand{\fq}{\mathfrak{q}}

\newcommand{\Greek}[1]{\mathrm{#1}} 

\newcommand{\Beta}{\Greek{B}}
\newcommand{\Epsilon}{\Greek{E}}

\newcommand{\Zeta}{\Greek{Z}}

\DeclareFontFamily{U}{mathb}{}
\DeclareFontShape{U}{mathb}{m}{n}{
  <-5.5> mathb5 <5.5-6.5> mathb6
  <6.5-7.5> mathb7 <7.5-8.5> mathb8 <8.5-9.5> mathb9 <9.5-11> mathb10
  <11-> mathb12}{}
\DeclareFontSubstitution{U}{mathb}{m}{n}
\DeclareRobustCommand{\blackdiamond}{\text{\usefont{U}{mathb}{m}{n}\symbol{"0C}}}

\makeatletter
\newcommand{\fdsy@scale}{1.0}
\newcommand\fdsy@mweight@normal{Book}
\newcommand\fdsy@mweight@small{Book}
\newcommand\fdsy@bweight@normal{Medium}
\newcommand\fdsy@bweight@small{Medium}
\DeclareFontFamily{U}{FdSymbolA}{}
\DeclareFontShape{U}{FdSymbolA}{m}{n}{
    <-7.1> s * [\fdsy@scale] FdSymbolA-\fdsy@mweight@small
    <7.1-> s * [\fdsy@scale] FdSymbolA-\fdsy@mweight@normal}{}
\DeclareFontShape{U}{FdSymbolA}{b}{n}{
    <-7.1> s * [\fdsy@scale] FdSymbolA-\fdsy@bweight@small
    <7.1-> s * [\fdsy@scale] FdSymbolA-\fdsy@bweight@normal}{}
\makeatother

\DeclareSymbolFont{fdsymbols}{U}{FdSymbolA}{m}{n}
\SetSymbolFont{fdsymbols}{bold}{U}{FdSymbolA}{b}{n}

\DeclareMathSymbol{\oequal}{\mathbin}{fdsymbols}{110}

\newcommand{\ZF}{ZF}

\newcommand{\blank}{\_}
\newcommand{\midblank}{\textnormal{--}}

\newcommand{\equivalent}{\simeq}
\newcommand{\isomorphic}{\cong}

\newcommand{\defeq}{\mathbin{\vcentcolon\equiv}}
\newcommand{\eqndefeq}{\eqnequalspace\defeq\eqnequalspace}

\newcommand{\primed}[1]{{#1}^\prime}
\newcommand{\pprimed}[1]{{#1}^{\prime\hspace{-1pt}\prime}}

\newcommand{\comp}{\mathbin{\circ}}

\newcommand{\transpose}[1]{{#1}^\mathrm{T}}

\RequirePackage{adjustbox}

\def\tytfont{\mathrm}
\def\tytsffont{\mathsf} 
\def\tytcalfont{\mathcal} 

\newcommand{\lawbeta}{\upbeta}
\newcommand{\laweta}{\upeta}

\newcommand{\oftype}{\mathrel{:}}

\newcommand{\MLTT}{MLTT}
\newcommand{\HoTT}{HoTT}
\newcommand{\twoLTT}{\textlf{2}LTT{}}

\newcommand{\UniverseType}{\tytcalfont{U}}
\NewDocumentCommand{\El}{s}{\tytfont{El}\IfBooleanTF{#1}{\,}{}}
\newcommand{\UniverseSuc}[1]{{#1}^{+}}

\newcommand{\SetType}{\tytfont{Set}}

\newcommand{\UnitType}{\mathbb{1}}

\newcommand{\EmptyType}{\mathbb{0}}
\NewDocumentCommand{\EmptyTypeelim}{s}{\EmptyType\text{-}\tytfont{elim}\IfBooleanTF{#1}{\,}{}}

\NewDocumentCommand{\Pitype}{sO{}sO{}}{
  \IfBooleanTF{#1}%
    {\raisebox{-0.12\baselineskip}{\Large\(\Pi\)}%
      \ifblank{#2}
        {}
        {\,{\raisebox{-1pt}{\footnotesize\(#2\)}} 
          \ifblank{#4}{}{\IfBooleanTF{#3}{\hspace{1pt},\,{#4}}{\;{#4}}}}}%
    {\Pi\ifblank{#2}{}{\,{#2}\ifblank{#4}{}{\,{#4}}}}
}

\NewDocumentCommand{\Sigmatype}{sO{}sO{}}{
  \IfBooleanTF{#1}%
    {\raisebox{-0.12\baselineskip}{\Large\(\Sigma\)}%
      \ifblank{#2}
        {}
        {\,{\raisebox{-1pt}{\footnotesize\(#2\)}}%
          \ifblank{#4}{}{\IfBooleanTF{#3}{\hspace{1pt},}{}{\,{#4}}}}}%
    {\Sigma%
      \ifblank{#2}
        {}
        {\,{#2}
          \ifblank{#4}{}{\IfBooleanTF{#3}{,}{}{\,{#4}}}}}%
}

\NewDocumentCommand{\fst}{s}{\tytfont{fst}\IfBooleanTF{#1}{\,}{}}
\NewDocumentCommand{\snd}{s}{\tytfont{snd}\IfBooleanTF{#1}{\,}{}}

\NewDocumentCommand{\equivto}{O{}}{\xrightarrow[#1]{\sim}}
\newcommand{\isequiv}{\tytfont{is\text{-}equiv}}

\newcommand{\refl}{\tytfont{refl}}
\newcommand{\pathinv}[1]{{#1}^{-1}} 
\newcommand{\pathcomp}{\mathbin{\cdot}} 
\NewDocumentCommand{\transpover}{O{}m}
  {{}_{\,\downarrow}{\ifblank{#1}{_{\,{#2}}}{^{#1}_{\,{#2}}}}}
\NewDocumentCommand{\ap}{O{}O{}}
  {\tytfont{ap}\ifblank{#1}{}{\,{#1}\ifblank{#2}{}{\,{#2}}}}

\newcommand{\fiber}[1]{{#1}^{-1}}

\NewDocumentCommand{\Prop}{O{}}{\tytfont{Prop}_{\ifblank{#1}{}{#1}}}

\newcommand{\Sigmatypeeq}{\tytfont{pair^{=}}}

\def\ctfont{\mathrm}

\newcommand{\Ob}[1]{{#1}_0}
\NewDocumentCommand{\Hom}{sO{}mm}{
  \ifblank{#2}{\ctfont{hom}}{#2}%
  \IfBooleanTF{#1}%
    {{#3}\,{#4}}%
    {({#3},\hspace{1.5pt}{#4})}
}

\NewDocumentCommand{\id}{O{}}
  {\ctfont{id}\ifblank{#1}{}{_{#1}}}

\newcommand{\opcat}[1]{{#1}^\mathrm{op}}
\newcommand{\arrowcat}[1]{{#1}^{\shortrightarrow}}

\newcommand{\Set}{\ctfont{Set}}

\newcommand{\whisker}{\mathbin{\ast}}

\newcommand{\associator}{\alpha}
\newcommand{\invassociator}{\pathinv{\associator}}
\newcommand{\lunitor}{\lambda}
\newcommand{\runitor}{\rho}

\newcommand{\Con}{\tytsffont{Con}}

\NewDocumentCommand{\Sub}{O{}O{}}
  {\tytsffont{Sub}\ifblank{#1}{}{\,{#1}\ifblank{#2}{}{\,{#2}}}}

\newcommand{\subcomp}{}

\newcommand{\emptycon}{\blackdiamond}

\NewDocumentCommand{\Ty}{s}
  {\tytsffont{Ty}\IfBooleanTF{#1}{\,}{}}

\NewDocumentCommand{\Tm}{O{}s}
  {\tytsffont{Tm}\ifblank{#1}{}{_{\,{#1}}}\IfBooleanTF{#2}{\,}{}}

\newcommand{\subst}[1]{\tytsffont{[}{#1}\tytsffont{]}}
\newcommand{\substT}[1]{{\subst{#1}_\tytsffont{T}}}
\newcommand{\substt}[1]{{\subst{#1}_\tytsffont{t}}}

\newcommand{\cwfproj}{\tytsffont{p}}
\newcommand{\cwfvar}{\tytsffont{q}}
\newcommand{\ctxext}{\mathord{.}\hspace{1pt}}
\newcommand{\subext}{\hspace{1.5pt}\mathord{,}\hspace{2pt}}

\newcommand{\substTid}{{\substT{\tytsffont{id}}}}
\newcommand{\substtid}{{\substt{\tytsffont{id}}}}
\newcommand{\substTcomp}{{\substT\subcomp}}
\newcommand{\substtcomp}{\substt\subcomp}
\newcommand{\cwfprojbeta}{\cwfproj\upbeta}
\newcommand{\cwfvarbeta}{\cwfvar\upbeta}
\newcommand{\subexteta}{\subext\upeta}
\newcommand{\subextcomp}{\subext\subcomp}


%% file: macros-typesetting.tex
\newcommand{\emphb}[1]{\textsb{#1}}

%% file: macros-math-paper.tex
\newcommand{\C}{\cC}

\newcommand{\catcomp}{\mathbin{\diamond}}
\newcommand{\UniverseCat}{\cU}

\newcommand{\Sect}{\tytfont{Sect}}
\newcommand{\Retr}{\tytfont{Retr}}

\newcommand{\wildequiv}[3]{{#2} \equivalent_{#1} {#3}}
\newcommand{\wildiso}[3]{{#2} \isomorphic_{#1} {#3}}
\newcommand{\wildisotoequiv}[1]{\ctfont{isotoeqv}_{#1}}
\newcommand{\idd}{\ctfont{idd}}
\newcommand{\idtowildequiv}[1]{\ctfont{idtoeqv}_{#1}}
\newcommand{\wildequivtoid}[1]{\ctfont{eqvtoid}_{#1}}

\newcommand{\trianglecoh}[2]{\triangle_{{#1},\,{#2}}}
\newcommand{\pentagoncoh}[4]{\pentagon_{{#1},\,{#2},\,{#3},\,{#4}}}

\newcommand{\Cospan}{\mathrm{Cospan}}
\newcommand{\CommSq}{\mathrm{CommSq}}
\newcommand{\sqcomp}{\mathbin{\raisebox{0.08\baselineskip}{\(\scriptscriptstyle\boxempty\)}}}

\newcommand{\squarehpaste}[2]{{#1}\mid{#2}}
\newcommand{\squarevpaste}[2]{\frac{#1}{#2}}
\newcommand{\squarevpasteP}{\squarevpaste\midblank\fP}
\newcommand{\Psqcomp}[1]{\fP \sqcomp_{#1} \blank}
\newcommand{\ispb}{\tytfont{is\text{-}pullback}}
\newcommand{\CommSqV}{\CommSq}

\newcommand{\Pullback}{\tytfont{Pullback}}

\newcommand{\idpb}[1]{\fI_{#1}}

\newcommand{\thecwf}{\C}

\newcommand{\eqsubsubstT}[1]{{{}\substT{^= {#1}}}}
\newcommand{\eqsubsubstt}[1]{{{}\substt{^= {#1}}}}

\newcommand{\eqtypesubstT}[2]{{#1}\substT{#2}}
\newcommand{\eqtermsubstt}[2]{{#1}\substt{#2}}

\newcommand{\UniverseCwf}{\UniverseCat}
\newcommand{\SetCwf}{\SetType}

\newcommand{\subeq}{\tytfont{sub}^{=}}
\newcommand{\subeqaux}{\subeq_0}

\newcommand{\subetaequality}[1]{\laweta^{\tytfont{sub}}_{#1}}

\newcommand{\conjsubstTcomp}[1]{\pathinv\substTcomp \pathcomp {#1} \pathcomp \substTcomp}

\newcommand{\liftT}[1]{{}^{\,.\,{#1}}}
\newcommand{\subliftT}[2]{{#1}\liftT{#2}}
\newcommand{\substTpb}[2]{\fP_{#1,\,{#2}}}
\newcommand{\substTpbtwocell}{\pathinv\cwfprojbeta}
\newcommand{\thegapmap}{\mu}
\newcommand{\gapmap}[2]{\thegapmap_{#1, #2}}

\newcommand{\thecleaving}{\mathfrak{cl}}
\newcommand{\cleaving}[4]{\thecleaving_{{#1}, {#2}}({#3}, {#4})}
\newcommand{\cleavinglift}[2]{\mathcal{l}_{{#1}, {#2}}}
\newcommand{\cleavingcomm}[2]{\fp_{{#1}, {#2}}}

\NewDocumentCommand{\cwfPitype}{O{}O{}}{\hat{\Uppi}\ifblank{#1}{}{\,{#1}\ifblank{#2}{}{\,{#2}}}}
\NewDocumentCommand{\cwflam}{s}{\uplambda\IfBooleanTF{#1}{\,}{}}
\NewDocumentCommand{\cwfapp}{s}{\tytfont{app}\IfBooleanTF{#1}{\,}{}}

\NewDocumentCommand{\cwfSigmatype}{O{}O{}}{\hat{\Upsigma}\ifblank{#1}{}{\,{#1}\ifblank{#2}{}{\,{#2}}}}

\newcommand{\cwfUtype}{\mathbb{U}}
\NewDocumentCommand{\cwfEl}{s}{\tytfont{El}\IfBooleanTF{#1}{\,}{}}

\NewDocumentCommand{\cwfEmptyUtype}{O{}}{\cwfUtype^\EmptyType\ifblank{#1}{}{_{#1}}}
\NewDocumentCommand{\cwfEmptyEl}{O{}s}{\cwfEl^\EmptyType\ifblank{#1}{}{_{#1}}\IfBooleanTF{#2}{\,}{}}
\NewDocumentCommand{\cwfUnitUtype}{O{}}{\cwfUtype^\UnitType\ifblank{#1}{}{_{#1}}}
\NewDocumentCommand{\cwfUnitEl}{O{}s}{\cwfEl^\UnitType\ifblank{#1}{}{_{#1}}\IfBooleanTF{#2}{\,}{}}

\NewDocumentCommand{\Tel}{s}
  {\tytsffont{Tel}\IfBooleanTF{#1}{\,}{}}

\NewDocumentCommand{\emptytel}{O{}}{\scaleobj{.8}{\bullet}_{\ifblank{#1}{}{#1}}}

%% file: setup-math.tex
\theoremstyle{default}
\newtheorem{definition}{Definition}[subsection]
\newtheorem{theorem}[definition]{Theorem}
\newtheorem{lemma}[definition]{Lemma}
\newtheorem{proposition}[definition]{Proposition}
\newtheorem{corollary}[definition]{Corollary}
\newtheorem{example}[definition]{Example}
\newtheorem{examples}[definition]{Examples}
\newtheorem{nonexample}[definition]{Non-example}

\newtheorem{remark}[definition]{Remark}
\newtheorem{conjecture}[definition]{Conjecture}

\AddToHook{env/theorem/begin}{\crefalias{definition}{theorem}}
\AddToHook{env/lemma/begin}{\crefalias{definition}{lemma}}
\AddToHook{env/proposition/begin}{\crefalias{definition}{proposition}}
\AddToHook{env/corollary/begin}{\crefalias{definition}{corollary}}
\AddToHook{env/example/begin}{\crefalias{definition}{example}}
\AddToHook{env/examples/begin}{\crefalias{definition}{examples}}
\AddToHook{env/nonexample/begin}{\crefalias{definition}{nonexample}}
\AddToHook{env/remark/begin}{\crefalias{definition}{remark}}
\AddToHook{env/conjecture/begin}{\crefalias{definition}{conjecture}}

%% file: setup-typesetting.tex
\newcommand{\orcidsymbol}{\raisebox{-.5pt}{\includegraphics[height=9pt, keepaspectratio]{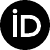}}}
\makeatletter
\newcommand{\@authdata}[3]{{#1} \href{#2}{\orcidsymbol}\\{\small{#3}}}
\newcommand{\authordata}[3]{
  \author{\@authdata{#1}{#2}{#3}}
  \def\theauthor{#1}
  \date{}}
\makeatother

\newfloat{diag}{ht}{diag}[definition]
\floatname{diag}{Diagram}
\DeclareCaptionLabelFormat{diag}{\textsb{#1} \textbf{#2}}
\DeclareCaptionSubType{diag}

\crefname{diag}{diagram}{diagrams}

%% file: setup-tikzcd-paper.tex
\tikzcdset{
  path/.style={
    equal,
    decoration={
      markings,
      mark=at position 0.58 with {\arrow[scale=1,line width=0.63pt]{>}}
    },
    postaction={decorate}
  }
}
\tikzcdset{paths/.style={arrows=path, labels={inner sep=1.3ex}}}

\tikzcdset{
  pathnearend/.style={
    equal,
    decoration={
      markings,
      mark=at position 0.85 with {\arrow[scale=1,line width=0.63pt]{>}}
    },
    postaction={decorate}
  }
}

\tikzset{commutative diagrams/.cd,
commutesstyle/.style={start anchor=center,end anchor=center,draw=none}
}
\newcommand\markcomm[2][\circ]{%
  \arrow[commutesstyle]{#2}[description]{#1}}

%% file: abstract.tex
\begin{abstract}
  The program of \emph{internal type theory} seeks to develop the categorical model theory of dependent type theory using the language of dependent type theory itself.
  In the present work we study internal homotopical type theory by relaxing the notion of a category with families (cwf) to that of a \emph{wild}, or \emph{precoherent higher} cwf, and determine coherence conditions that suffice to recover properties expected of models of dependent type theory.
  The result is a definition of a \emph{split 2-coherent wild cwf}, which admits as instances both the syntax and the 
  ``standard model'' given by a universe type.
  This will allow us to give a straightforward internalization of the notion of a 2-coherent reflection of homotopical type theory in itself---namely as a 2-coherent wild cwf morphism from the syntax to the standard model.
  Our theory also easily specializes to give definitions of ``low-dimensional'' higher cwfs, and conjecturally includes the container higher model as a further instance.
\end{abstract}

%% file: intro.tex
\section{Introduction}\label{sec:intro}

\subsection{Internal type theory}

Given a sufficiently expressive logical system \textmathfrak{L}, it is interesting and productive to ask
\begin{center}\emph{To what extent does \textmathfrak{L} internalize itself?}\end{center}
In more detail, one seeks to develop a suitable notion of \emph{interpreting structure}, aka \emph{model}, of \textmathfrak{L}, and to study the theory of 
such models, entirely within the language and logic of \textmathfrak{L} itself.%
\footnote{Immediate disclaimer: in this paper the term ``model'' always refers to this logical notion, and not to homotopy theoretic \emph{model structures}.
There are, of course, models of homotopy type theory in model structures.}
The techniques and perspectives granted by 
the study of these \emph{inner models} have historically been used to surprising effect, e.g.\ to prove \ZF-relative consistency of the axiom of choice and the continuum hypothesis~\cite{goedel:38:consistency-ac-gch}, or to show independence of the Whitehead problem---a statement in homological algebra about short exact sequences of abelian groups---from the traditionally accepted set theoretic foundations of mathematics~\cite{shelah:74:whitehead}.

The program of \emph{internal type theory}, first articulated by Dybjer~\cite{dybjer:96:internal-tt}, seeks to develop the same paradigm in the setting of dependent type theory%
\footnote{For conciseness, henceforth simply ``type theory''.}
by studying the categorical model theory of intensional Martin-L\"{o}f type theory (\MLTT{}) using \MLTT{}. 
Among the type theory community, this is also known as \emph{internal model theory} of type theory.

The foundational task of internal type theory is then to give a type theoretic definition of the notion of a categorical ``model'' of type theory, ensuring that it 
captures all the examples we care about.
In early work, Buisse and Dybjer~\cite{buisse-dybjer:08:towards-formalizing} formalize the type of \emph{1-categories with families} (\emph{1-cwfs})~\cite{dybjer:96:internal-tt,hofmann:97:syntax-semantics} in \MLTT{}, and discuss approaches to constructing term models and the initial cwf.
Later, Ahrens, Lumsdaine and Voevodsky~\cite{alv:18:categorical-structures} consider a number of other kinds of model in the same setting,%
\footnote{On occasion, assuming univalence.}
and show the equivalence of 1-cwfs with
(split) type categories~\cites[Definition 2.2.1]{vandenberg-garner:12:models-id-types}{pitts:95:cat-logic,pitts:01:cat-logic}
(aka categories with attributes~\cites[\S6]{moggi:91:categorical-account}{hofmann:97:syntax-semantics})
and representable natural transformations
(aka natural models~\cite{awodey:18:natural-models}),
via relative universes.

Two structures that we would particularly like to exhibit as internal models of a type theory are its \emph{syntax} and, if present, its \emph{universe type}.
Altenkirch and Kaposi~\cite{altenkirch-kaposi:16:tt-in-tt} show that if \MLTT{} with uniqueness of identity proofs (UIP) is extended with \emph{quotient inductive-inductive types} (\emph{QIITs})~\cite{alt+:18:qiits,kka:19:qiits} then the strongly typed term-cwf sketched by Buisse and Dybjer~\cite[\S6.1]{buisse-dybjer:08:towards-formalizing} can be internally constructed, thus yielding an internal syntactic model.
Even better, from this inductive construction it immediately follows that internal 1-cwfs satisfy the \emph{initiality principle}, which posits that syntax should give rise to an initial object in an appropriate (higher) category of models. 
Assuming UIP also allows any universe type in \MLTT{} to be equipped with a canonical 1-cwf structure, giving rise to an internal model known in the literature as the \emph{standard model}~\cites[\S4]{altenkirch-kaposi:16:tt-in-tt}[Example 3]{kraus:21:internal}.

Now, for any internal definition of model 
of \MLTT{} which includes as instances 
both the syntax \textmathcal{S} and a universe \textmathcal{U}, we may ask if the type of model morphisms 
from \textmathcal{S} to \textmathcal{U} is definable and inhabited.
Elements of such a type may be viewed as \emph{self-interpretations}, or \emph{reflections} of \MLTT{} in itself.%
\footnote{The self-interpretation---as opposed to just the internal \emph{representation}---of type theory is a fascinating topic of study; Rendel, Ostermann and Hofer~\cite[\S2]{roh:09:typed-self-representation} give an excellent account of the subtle distinctions involved, albeit in the context of typed lambda calculi.}
By the preceding discussion, in \MLTT{}+\(\UniverseType\)+UIP+QIITs the type of 1-cwf morphisms from the syntax to the standard universe model is definable, and moreover inhabited by eliminating the syntax into the universe.
We may colloquially summarize all of this by saying that any model%
\footnote{Taking 1-cwfs, or any of the equivalent notions mentioned previously.}
of \MLTT{}+\(\UniverseType\)+UIP+QIITs has an inner syntactic model \textmathcal{S} of its \MLTT{} fragment, as well as an internalized reflection function \(\mathcal{S} \to \UniverseType\).%

\subsection{Internal homotopical type theory}

Can we tell a similar story about inner models of \emph{homotopical} \MLTT{}, i.e.\ without assuming UIP?
Unfortunately, the account given in the previous section depends heavily on the uniqueness of identity proofs---both for the semantics of QIITs, but also for the interpretation of the universe as an internal model.
In particular, it is well known that types with their higher identities possess \emph{\(\infty\)-groupoidal} structure in the absence of UIP~\cite{lumsdaine:10:weak-omega-categories,vandenberg-garner:11:weak-omega-groupoids,kapulkin-lumsdaine:21:simplicial-model}, which means that the canonical standard model on the universe \(\UniverseType\) is no longer a 1-cwf since its substitutions \(\UniverseType(A, B) \defeq A \to B\) are not generally sets.

We must therefore return to the foundational task of internal type theory, and ask what a good notion of internal model of homotopical type theory might look like.
This question is considered by Kraus~\cite{kraus:21:internal} in the setting of a two-level type theory (\twoLTT{})~\cite{acks:23:2ltt} that extends homotopical \MLTT{} with an additional layer of \emph{strict}, aka \emph{non-fibrant}, types, including an equality type former satisfying UIP.
With this extra structure it is now possible to define \(\infty\)-categories as semi-Segal types~\cite{capriotti-kraus:18:semisegal} having idempotent equivalences~\cite[\S{}III]{kraus:21:internal}; to answer the question of models by defining \emph{\(\infty\)-categories with families}; and to show that the syntax QIIT, the universe standard model, and any slice of an \(\infty\)-cwf are all \(\infty\)-cwfs.


However, models of \twoLTT{} are somewhat strong extensions of models of \MLTT{}, and it might therefore be argued that \(\infty\)-cwfs should not be considered to be inner models of homotopical \MLTT{} in itself.%
\footnote{Of course, even the account we gave of inner models of \MLTT{}+UIP in the previous section does not, strictly speaking, provide models of \MLTT{}+UIP in itself.
One source of tension in inner model theory is to make the gap between the theory of the outer and inner models (the ``host'' theory and the ``object'' theory) as small as possible.
How large this gap is allowed to be is typically a matter of some subjective judgment.}
Indeed, Kraus conjectures~\cite[\S{}VI]{kraus:21:internal} that plain homotopy type theory (\HoTT{})~\cite{hott-book} does not internalize itself, but instead considers the definition of \(\infty\)-cwfs to be a step towards showing that \twoLTT{} internalizes \HoTT{} and also itself.

Another potential approach is to work in the axiomatic variation of \emph{simplicial homotopy type theory}~\cite{riehl-shulman:17:shott,riehl-shulman:23:shott} developed by Gratzer, Weinberger and Buchholtz~\cite{gwb:24:directed-univalence-shott}, in which it is possible to give a straightforward definition of (\(\infty\),\,1)-categories as Segal or Rezk types.
However, in this theory representable presheaves are defined using additional modalities~\cite{gwb:25:yoneda-embedding-stt}, again making it a rather strong extension of plain \HoTT{}.


\subsection{Contributions}

Thus, in the present work we 
stick with the original question 
and fix homotopical \MLTT{} to be the theory of our outer models (i.e.\ the theory in which we perform our constructions).
Since defining a type of \(\infty\)-categories and, a fortiori, \(\infty\)-categorical models in this theory remains an open problem, we take the notion of a \emph{wild}, or \emph{precoherent higher} category with families as our starting point for ``internal model of homotopical \MLTT{}'' (\Cref{sec:wild-cwfs}).
We then determine coherence conditions such that sufficiently coherent internal cwfs satisfy many good properties expected of any model of type theory (\Cref{sec:2-coh-ctx-ext}).
In particular, we show that requiring \emph{2-coherence} suffices to equip any wild cwf \(\C\) with the expected cloven fibrational structure, which furthermore satisfies a coherent ``splitting'' property when the category of contexts of \(\C\) is either set-level or univalent (\Cref{sec:ctx-comprehension}).
This generalizes the well known result that 1-cwfs are equivalent to full split comprehension categories.
In order to do this, we first have to develop some wild category theory (\Cref{sec:wild-categories}), as well as a theory of pullbacks in wild categories (\Cref{sec:pullbacks}).

The benefit of our approach is that we are able to capture, in a single internally definable notion, both set-level models such as the syntax as well as untruncated higher models such as the universe.
We are then able to prove results simultaneously for all such models by stating them as generally as possible in terms of coherence, without resorting to truncation assumptions or restricting to some homotopy \(n\)-level.
Our theory is still easily applicable to models with ``low-dimensional'' higher homotopy: in particular, the simple generalization of the theory of 1-cwfs to allow the type presheaf to be valued in 1-types is (almost trivially) an instance of our theory.

\subsection{Assumptions and conventions}

\paragraph{Logical setting} Our definitions and constructions are in the setting of 
\MLTT{} with \(\Pitype\), \(\Sigmatype\) and intensional identity types, without assuming uniqueness of identity proofs.
We assume function extensionality as well as \(\laweta\) for \(\Pitype\)-types throughout.
We do not globally assume univalence, but instead take it to be a property that may or may not hold for a given universe type.

\paragraph{Univalent 1-category theory} Of the reader, we assume familiarity with the basics of homotopy type theory as well as the theory of 1-categories in univalent foundations, standard references for which are \cite{hott-book} and \cite{aks:15:univalent-cats}.

\paragraph{Transport} We frequently work with explicit transports over complicated path concatenations, and use the notation
\[ a \transpover[P] e \]
to denote the transport of an element \(a \oftype P(x)\) along an equality \(e \oftype x = y\).
We will also use, without explicit comment, the equation
\[ a \transpover[P] e \transpover[P] {\primed{e}} = a \transpover[P] {e \pathcomp \primed{e}} \,,\]
which holds by \cite[Lemma 2.3.9]{hott-book}.

%% file: wild-category-theory.tex
\section[Wild categories]{Wild Categories}\label{sec:wild-categories}

To internalize categorical models of homotopical type theory we first have to define some type of suitable category-like structure that captures both the 1-categorical syntax as well as the \(\infty\)-categorical universe.
Until the problem of defining general \(\infty\)-categorical structures using only features of homotopical \MLTT{} is solved,
we consider the approximation given by \emph{wild}, or \emph{precoherent higher}, categories.
Variations on this notion appear a number of times in the literature and formalization libraries~\cites{capriotti-kraus:18:semisegal}[\S3.1]{hart-hou:24:coslice-colimits}{agda-unimath-wild-category-theory}{coq-hott-wild-cat}, typically as an intermediate stage toward defining more coherent structures.
We use a particularly simple definition: they are almost precategories~\cite[Definition 3.1]{aks:15:univalent-cats}, except that morphisms form \emph{types} instead of sets.

\begin{definition}[Wild categories]\label{def:wild-categories}
  A \emph{wild} or \emph{precoherent higher category} \(\C\) consists of:
  \begin{itemize}
  \item
    A type \(\Ob\C\) of \emphb{objects}.
  \item
    For all \(x, y \oftype \Ob\C\), a \emph{type} \(\Hom[\C]{x}{y}\) of \emphb{morphisms} from \(x\) to \(y\).
  \item
    An \emphb{composition} operation \(\catcomp\) of compatible morphisms, together with an \emphb{associator} \(\associator\) witnessing associativity of composition
    \[ (h \catcomp g) \catcomp f \xlongequal{\associator_{f,g,h}} h \catcomp g \catcomp f \]
    for all morphisms \(f\), \(g\), \(h\).
  \item
    \emphb{Identity morphisms} \(\id_x\) for all \(x \oftype \Ob\C\), together with \emphb{unitors} \(\lunitor\) and \(\runitor\) witnessing the left and right identity equations
    \begin{gather*}
      \id_y \catcomp f \xlongequal{\lunitor_f} f \\
      f \catcomp \id_x \xlongequal{\runitor_f} f
    \end{gather*}
    for all \(f \oftype \Hom[\C]xy\).%
  \end{itemize}
\end{definition}

Crucially, no further coherence laws on hom-types are required in the definition of a wild category.
We frequently leave the indices on the associator, unitors and identity morphisms implicit.

\begin{examples}[Wild categories]\label{eg:wild-categories}
  We are primarily interested in two particular kinds of wild category, which are in fact completely coherent.
  These are distinguished by the behavior of their hom-types: the ``maximally truncated'' and the ``nontrivially fully coherent''.

  The first kind consists of the wild categories whose morphisms form sets: any precategory---and hence, set-level or univalent 1-category---is immediately a wild category.

  The second kind consists of the type theoretic (sub-) universes.
  Any universe type \(\UniverseType\) gives rise to a wild category, also denoted \(\UniverseCat\), with objects \(\Ob\UniverseCat \defeq \UniverseType\), and whose hom-types \(\Hom[\UniverseCat]AB\) are the function types \(A \to B\).
  Composition is given by function composition and identity morphisms by identity functions.
  The associativity and unit laws hold definitionally, i.e.\ \(\associator\), \(\lunitor\) and \(\runitor\) are families of trivial equations.
  This definition applies equally well to any reflective subuniverse~\cite[Definition 7.7.1]{hott-book},
  \footnote{This should also straightforwardly yield a notion of ``wild reflective subcategory''; we do not develop this in this paper.}
  and in this way the \(m\)-modal types for any modality \(m\) on a universe \(\UniverseType\) \cites[\S7.7]{hott-book}[\S1]{rss:20:modalities} form a wild category.
  In particular, we have wild categories whose objects are the \(n\)-types in \(\UniverseType\), which might be seen as prototypical examples of wild \(\text{(\(n\),\,1)}\)-categories.
\end{examples}

We draw diagrams in wild categories after the familiar notation, with one important difference.
Since morphisms in a wild category can form arbitrary types, the commutativity of a given face of a diagram is no longer property, but \emph{data}.
Hence we denote commuting faces with the notation
\[\begin{tikzcd}
  x \rar["f",out=50,in=130,""{name=U,below}]
    \rar["g"{swap, yshift=-1pt},out=-55,in=-135,end anchor={[xshift=4pt]},""{name=D,above}]
  & y \arrow["\text{\footnotesize\(\,\gamma\)}",path,from=U,to=D]
\end{tikzcd}\]
to make explicit the commutativity witness \(\gamma \oftype f = g\).
As this notation suggests, equalities between morphisms in a wild category are considered to be \emph{2-cells}, and higher equalities, \emph{higher cells}.
We stress that these cells are not explicitly axiomatized as part of the generalized algebraic theory of a wild category, but instead arise out of the ambient homotopical type theory.

\begin{remark}[On the semantics of wild categories]
  The fact that a wild category is a proper generalization of a precategory therefore relies crucially on the homotopical nature of the identity type.
  This simple observation suggests that we distinguish the \emph{categorical} and \emph{typal} directions of a wild category: the 1-cells belong to the categorical direction, and the higher cells to the typal.
  This recalls the point of view---implicit in Rezk~\cite{rezk:01:model} and explicated by Joyal and Tierney~\cite{joyal-taylor:07:quasicats-segalspaces}---that Segal spaces (i.e.\ bisimplicial sets \(\opcat\Delta \times \opcat\Delta \to \Set\) satisfying the Segal condition) have categorical and \emph{spatial} directions.
  Indeed, Capriotti and Kraus~\cite[Theorem 4.9]{capriotti-kraus:18:semisegal} show (a version of) this to hold.%
  \footnote{By interpreting their type theoretic proof in the simplicial model of HoTT~\cite{kapulkin-lumsdaine:21:simplicial-model}.}

  Alternatively, following Lumsdaine~\cite{lumsdaine:10:weak-omega-categories} and van den Berg and Garner~\cite{vandenberg-garner:11:weak-omega-groupoids} in viewing types as Batanin-Leinster weak \(\omega\)-groupoids, it should also be possible to make precise the point of view that wild categories are categories ``wildly enriched'' in weak \(\omega\)-groupoids.%
  \footnote{That is, weakly enriched, but without requiring any coherences on \(n\)-cells for \(n \ge 2\).}
\end{remark}

\subsection{Common concepts}

Expectedly, a large number of elementary concepts from univalent 1-category theory~\cite{aks:15:univalent-cats} and bicategory theory~\cite{ahrens+:21:bicat-univalent-foundations} straightforwardly transfer into, and are subsumed by, the wild categorical setting.
We record them here, as well as in \Cref{subsec:2-coh,subsec:equivalence-univalence}, for completeness and future reference.

\begin{definition}[Terminal objects]\label{def:terminal-obj}
  An object \(y\) in a wild category \(\C\) is \emph{terminal} if \(\Hom[\C]xy\) is contractible for all objects \(x \oftype \Ob\C\).
\end{definition}

\begin{definition}[Sections and retractions]\label{def:section-retraction}
  The type of \emph{sections} of a morphism \(f \oftype \Hom[\C]xy\) in a wild category \(\C\) is
  \[ \Sect(f) \defeq \Sigmatype*[(s \oftype \Hom[\C]yx)]*[f \catcomp s = \id]. \]
  Similarly, the type of \emph{retractions} of \(f\) is
  \[ \Retr(f) \defeq \Sigmatype*[(r \oftype \Hom[\C]yx)]*[r \catcomp f = \id]. \]
\end{definition}

In contrast to 1-categories, sections and retractions in wild categories are not solely determined by their morphism component in general, but also by the identification of the section-retraction composite with the identity.

\begin{definition}[Whiskering]\label{def:whiskering}
  Given an equality \(\gamma \oftype g = \primed{g}\) of morphisms \(g, \primed{g} \oftype \Hom[\C]xy\), for any morphism \(f \oftype \Hom[\C]wx\) the \emph{right whiskering} \((\gamma \whisker f)\) of \(\gamma\) with \(f\) is the canonical equality
  \[\ap[(\blank \catcomp f)][\gamma] \oftype g \catcomp f = \primed{g} \catcomp f,\]
  and for any \(h \oftype \Hom[\C]yz\) the \emph{left whiskering} \((h \whisker \gamma)\) of \(\gamma\) with \(h\) is the equality
  \[ \ap[(h \catcomp \blank)][\gamma] \oftype h \catcomp g = h \catcomp \primed{g}.\]
\end{definition}

\begin{proposition}[Properties of whiskering]\label{prop:whiskering-properties}
  By induction, the following equations hold for right whiskering:
  \begin{gather*}
    \refl \whisker f \equiv \refl, \\
    \gamma \whisker \id = \runitor \pathcomp \gamma \pathcomp \pathinv\runitor, \\
    \pathinv{(\gamma \whisker f)} = \pathinv{\gamma} \whisker f, \\
    (\gamma \pathcomp \delta) \whisker f = (\gamma \whisker f) \pathcomp (\delta \whisker f),
  \end{gather*}
  as well as the analogous equations for left whiskering.
  We also have the following associativity laws expressing ``naturality'' of \(\associator\),
  \begin{gather*}
    g \whisker (f \whisker \gamma) = \invassociator \pathcomp ((g \catcomp f) \whisker \gamma) \pathcomp \associator, \\
    (\gamma \whisker g) \whisker f = \associator \pathcomp (\gamma \whisker (g \catcomp f)) \pathcomp \invassociator, \\
    (g \whisker \gamma) \whisker f = \associator \pathcomp (g \whisker (\gamma \whisker f)) \pathcomp \invassociator,
  \end{gather*}
  and the interchange law
  \[ (g \whisker \gamma) \pathcomp (\delta \whisker \primed{f}) = (\delta \whisker f) \pathcomp (\primed{g} \whisker \gamma) \]
  for all morphisms \(f\), \(\primed{f}\), \(g\), \(\primed{g}\) and equalities \(\gamma\), \(\delta\) as in
  \[\begin{tikzcd}
    x \rar["f",out=50,in=130,""{name=U,below}]
      \rar["\primed{f}"{swap, yshift=-1pt},out=-55,in=-135,end anchor={[xshift=4pt]},""{name=C,above}]
      \arrow["\text{\footnotesize\(\,\gamma\)}",path,from=U,to=C]
    & y \rar["g",out=50,in=130,""{name=V,below}]
        \rar["\primed{g}"{swap, xshift=-4pt},out=-50,in=-130,start anchor={[xshift=-2pt]},end anchor={[xshift=1pt]},""{name=D,above}]
        \arrow["\text{\footnotesize\(\,\delta\)}",path,from=V,to=D]
    & z
  \end{tikzcd}.\]
\end{proposition}

\subsection{2-coherence}\label{subsec:2-coh}

The following two coherence conditions are familiar from bicategory theory.

\begin{definition}[Triangle coherators]\label{def:triangle-coherators-unitors}
  A wild category \(\C\) has \emph{triangle coherators} if for all morphisms
  \[ x \xrightarrow{f} y \xrightarrow{g} z \]
  in \(\C\), there is an equality
  \[ \trianglecoh{f}g \oftype \associator \pathcomp (g \whisker \lunitor) = \runitor \whisker f \]
  filling the triangle
  \[\begin{tikzcd}[paths, column sep=0ex]
    (g \catcomp \id) \catcomp f
      \ar[rr, "\associator"]
      \ar[dr, "\runitor \whisker f" swap]
    && g \catcomp \id \catcomp f
      \ar[dl, "g \whisker \lunitor"]
      \\
    & g \catcomp f &
  \end{tikzcd}\]
\end{definition}

\begin{definition}[Pentagonators]\label{def:pentagon-coherators-associators}
  \(\C\) has \emph{pentagon coherators for associators}, or \emph{(\(\associator\)-) pentagonators}, if for all composable chains
  \[ v \xrightarrow{f} w \xrightarrow{g} x \xrightarrow{h} y \xrightarrow{k} z \]
  of morphisms in \(\C\), there is an equality
  \[ \pentagoncoh{f}ghk \oftype
    (\associator \whisker f) \pathcomp \associator \pathcomp (k \whisker \associator)
    = \associator \pathcomp \associator \]
  filling the usual pentagon
  \[\begin{tikzcd}[
    paths, 
    column sep=-5ex,
    nodes={
      row 1/.style={row sep=4ex},
      row 2/.style={row sep=5ex}}
    ]
    & ((k \catcomp h) \catcomp g) \catcomp f
      \ar[dl, "\associator \whisker f"{swap, yshift=-1.75ex, xshift=-0.75ex},
        start anchor={[xshift=-1ex]}, end anchor={[xshift=1ex]}]
      \ar[dr, "\associator"{yshift=-1.25ex, xshift=0.5ex}, ""{name=U, below},
        start anchor={[xshift=1ex]}, end anchor={[xshift=-1ex]},] &
    \\
    (k \catcomp h \catcomp g) \catcomp f
      \dar["\associator"{swap, yshift=1ex, xshift=-0.5ex},
        start anchor={[xshift=3ex]}, end anchor={[xshift=4.5ex]}]
    &&
    (k \catcomp h) \catcomp g \catcomp f
      \dar["\associator"{yshift=1.25ex, xshift=0.5ex},
        start anchor={[xshift=-2.75ex]}, end anchor={[xshift=-4.25ex]}]
    \\
    |[xshift=1.5ex]| k \catcomp (h \catcomp g) \catcomp f
      \ar[rr, "k \whisker \associator" swap]
    &&
    |[xshift=-1.5ex]| k \catcomp h \catcomp g \catcomp f
  \end{tikzcd}\]
\end{definition}

\begin{definition}[2-coherent wild categories]\label{def:2-coherent-wild-cat}
  We call a wild category \emph{2-coherent} if it has triangle and pentagon coherators.
\end{definition}

In the literature, Hart and Hou~\cite{hart-hou:24:coslice-colimits} call 2-coherent wild categories \emph{bicategories}; we do not use this terminology as 2-coherent wild categories always have invertible 2-cells (and higher cells). 
2-coherent wild categories are also essentially the \emph{wild 2-precategories} of Capriotti and Kraus~\cite{capriotti-kraus:18:semisegal}; the difference is that their definition also includes the other two triangle coherators of \Cref{prop:more-coherences-2-coh-wild-cats} in the data of the type.
We also have the following link to the univalent bicategory theory of Ahrens et al.~\cite{ahrens+:21:bicat-univalent-foundations}:

\begin{proposition}\label{prop:2-coh-wild-cat-prebicat}
  Any 2-coherent wild category is also a \emph{prebicategory} in the sense of Ahrens et al.~\cite[Definition 2.1]{ahrens+:21:bicat-univalent-foundations}, by taking the type of 2-cells from \(f\) to \(g\) to be the equality type \(f = g\), and using \Cref{prop:whiskering-properties}.
\end{proposition}

\begin{examples}[2-coherent higher categories]\label{eg:2-coherent-wild-cat}
  Any precategory trivially satisfies all higher equalities between equalities of morphisms.
  The universe wild categories \(\UniverseCat\) have definitionally unital and associative composition of morphisms, and thus have trivial triangle and pentagon coherators.
\end{examples}

Many coherences
involving \(\lunitor\), \(\runitor\) and \(\associator\) that hold in all bicategories
also hold in 2-coherent wild categories---namely, those that do not rely on uniqueness of equality of 2-cells.
In particular,

\begin{proposition}\label{prop:more-coherences-2-coh-wild-cats}
  In a 2-coherent wild category \(\C\), there are witnesses (not necessarily unique) that
  \begin{itemize}
    \item \(\lunitor_{\id_x} = \runitor_{\id_x}\) for all \(x \oftype \Ob\C\), and
    \item the diagrams of equalities
      \[
      \begin{tikzcd}[paths, column sep=0pt]
        (\id \catcomp g) \catcomp f \ar[rr,"\associator"] \ar[dr,"\lunitor \whisker f" swap]
        && \id \catcomp g \catcomp f \ar[dl,"\lunitor"] \\
        & g \catcomp f &
      \end{tikzcd} \eqntext{and}
      \begin{tikzcd}[paths, column sep=0pt]
        (g \catcomp f) \catcomp \id \ar[rr,"\associator"] \ar[dr,"\runitor" swap]
        &&  g \catcomp f \catcomp \id \ar[dl,"g \whisker \runitor"] \\
        & g \catcomp f &
      \end{tikzcd}
      \]
      commute for all \(f \oftype \Hom[\C]xy\) and \(g \oftype \Hom[\C]yz\).
  \end{itemize}
  We refer to \cite[Propositions 2.2.4 and 2.2.6]{johnson-yau:21:2-dim-cats} for proofs of these facts.
\end{proposition}

\begin{nonexample}[A wild category lacking 2-coherators]
  In \cite[Lemma 8]{kraus:21:internal-extended}, Kraus uses the circle higher inductive type to construct a wild category \(\cE\) with an object \(x \oftype\Ob\cE\) that refutes \(\lunitor_{\id_x} = \runitor_{\id_x}\).
  By \Cref{prop:more-coherences-2-coh-wild-cats}, \(\cE\) must therefore also fail to have triangle or pentagon coherators.
  In fact, one can see directly by \cite[Lemma 6.4.2]{hott-book} that the existence of triangle coherators for \(\cE\) is a \HoTT{} taboo.
\end{nonexample}

We will also use the following coherence.

\begin{proposition}\label{prop:a-coh-for-lunitor}
  Suppose that \(\gamma \oftype g = \id \catcomp f\) in a 2-coherent wild category.
  Then
  \[\begin{tikzcd}[paths]
    \id \catcomp g \dar["\lunitor" swap] \rar["\id \whisker \gamma"]
    & \id \catcomp \id \catcomp f \dar["\id \whisker \lunitor"]
    \\
    g \rar["\gamma" swap]
    & \id \catcomp f
  \end{tikzcd}\]
  commutes, i.e.\
  \(\lunitor_g \pathcomp \gamma =_{(\id\,\catcomp\,g\,=\,\id\,\catcomp\,f)} (\id \whisker \gamma) \pathcomp (\id \whisker \lunitor_f)\).
\end{proposition}
\begin{proof}
  We equivalently prove that \(\id \whisker \gamma = \lunitor_g \pathcomp \gamma \pathcomp \pathinv{(\id \whisker \lunitor_f)}\).
  Since \(\id \whisker \gamma = \lunitor_g \pathcomp \gamma \pathcomp \pathinv\lunitor_{\id \catcomp f}\) by properties of whiskering (\Cref{prop:whiskering-properties}), the result follows if \(\id \whisker \lunitor_f = \lunitor_{\id \catcomp f}\).
  \Cref{diag:idl-coh} shows how to construct such an equality: its interior is divided into two triangles and a bigon, which commute by the triangle coherator \(\trianglecoh{f}\id\) and the equalities in \Cref{prop:more-coherences-2-coh-wild-cats}.
  \input{diag-idl-coh}
\end{proof}

\subsection{Equivalence and univalence}\label{subsec:equivalence-univalence}

\begin{definition}[Wild equivalences]
  A morphism \(f \oftype \Hom[\C]xy\) in a wild category \(\C\) is a \emph{wild equivalence} if it has both a section and a retraction (i.e.\ is biinvertible).
  The type of wild equivalences from \(x\) to \(y\) in \(\C\) is denoted \(\wildequiv{\C}xy\).
  Its elements are also called \emph{\(\C\)-equivalences} to avoid confusion with type theoretic equivalences.
\end{definition}

The following observation appears as a basic result in the classical theory of \(\infty\)-categories~\cites[Proposition 1.2.4.1]{lurie:09:htt}[\S1.4]{harpaz:15:quasi-unital}.
To our knowledge, its type theoretic incarnation is due to Capriotti and Kraus, who also observe the immediate corollary that being a wild equivalence is a proposition~\cite[Lemma 5.3]{capriotti-kraus:18:semisegal}.

\begin{proposition}[The universal property of wild equivalences]\label{prop:wild-equiv-univ-prop}
  A morphism \(f \oftype \Hom[\C]xy\) in a wild category \(\C\) is a \(\C\)-equivalence if and only if it is \emph{neutral}---i.e.\ if and only if for any \(w, z \oftype \Ob\C\), the maps
  \begin{align*}
    f \catcomp \blank & \oftype \Hom[\C]wx \to \Hom[\C]wy \\
    \text{and} \quad \blank \catcomp f & \oftype \Hom[\C]yz \to \Hom[\C]xz
  \end{align*}
  are equivalences of hom-types.
\end{proposition}
\begin{proof}
  Let \(s\) and \(r\) be, respectively, a section and retraction of \(f\) in \(\C\).
  Then \((s \catcomp \blank)\) and \((r \catcomp \blank)\) are, respectively, sections and retractions of \((f \catcomp \blank)\), while \((\blank \catcomp r)\) and \((\blank \catcomp s)\) are, respectively, sections and retractions of \((\blank \catcomp f)\).

  Conversely, the inverse equivalences of \((f \catcomp \blank) \oftype \Hom[\C]yx \equivto \Hom[\C]yy\) and \((\blank \catcomp f) \oftype \Hom[\C]yx \equivto \Hom[\C]xx\) map identity morphisms to a section and retraction of \(f\), respectively.
\end{proof}

\begin{definition}[Wild isomorphisms]
  We also consider the type of wild \emph{isomorphisms} \(\wildiso{\C}xy\) in a wild category \(\C\), i.e.\ the type of morphisms \(f \oftype \Hom[\C]xy\) having a two-sided inverse \(g \oftype \Hom[\C]yx\).
\end{definition}

Any two-sided inverse is both a section and a retraction, so there is a canonical map
\[ \wildisotoequiv\C \oftype (\wildiso{\C}xy) \to (\wildequiv{\C}xy) \]
for any objects \(x, y \oftype \Ob\C\).
If \(\C\) is a precategory then \(\wildisotoequiv\C\) is an equivalence: its inverse sends a section-retraction pair \((s,r)\) to the two-sided inverse \((r \catcomp f \catcomp s)\) of \(f\).

\begin{definition}[Dependent identity morphisms]\label{def:idd}
  If \(x, y \oftype \Ob\C\) are objects of a wild category such that \(e \oftype x = y\), there is a morphism%
  \footnote{\(\idd\) for ``\emph{id}entity morphism \emph{d}ependent over an equality''.}%
  \[ \idd(e) \defeq \id_{x} \transpover[{\Hom[\C]x\blank}]e \oftype \Hom[\C]xy. \]
  By induction on \(e\), \(\idd(e)\) is an isomorphism, with inverse \(\idd(\pathinv{e})\).
  We thus get a map
  \[ \idd \oftype x = y \to \wildiso{\C}xy. \]
\end{definition}

In precategories, \(\idd\) is essentially \(\ctfont{idtoiso}\)~\cites[Definition 9.1.4]{hott-book}[Lemma 3.4]{aks:15:univalent-cats}.

\begin{definition}[Wild univalence]\label{def:idtowildequiv}
  A wild category \(\C\) is \emph{univalent} if, for all \(x, y \oftype \Ob\C\), the map
  \[ \idtowildequiv\C \defeq \wildisotoequiv\C \comp \idd \oftype x = y \to \wildequiv{\C}xy \]
  is an equivalence.
  Its inverse equivalence is denoted \(\wildequivtoid\C\).
\end{definition}

Wild univalence subsumes both 1-categorical univalence (when \(\C\) is a precategory) and the univalence axiom (when \(\C \equiv \UniverseCat\) is a universe).

%% file: diag-idl-coh.tex
\begin{diag}
  \caption{}\label{diag:idl-coh}
  \[\begin{tikzcd}[paths, column sep=1.5em]
    \id \catcomp \id \catcomp f
    \ar[rr, "\id \whisker \lunitor_f", in=135, out=45, end anchor=north]
    \ar[rr, "\lunitor_{\id \catcomp f}" swap, in=-135, out=-45, end anchor=south]
    & (\id \catcomp \id) \catcomp f
    \lar["\associator"']
    \rar["\runitor_{\id} \whisker f", bend left=20, start anchor=north, end anchor=north west]
    \rar["\lunitor_{\id} \whisker f", swap, bend right=20, start anchor=south, end anchor=south west]
    & [2em] \id \catcomp f
  \end{tikzcd}\]
\end{diag}

%% file: pullbacks.tex
\section[Pullbacks in 2-coherent wild categories]{Pullbacks in 2-Coherent Wild Categories}\label{sec:pullbacks}

In this section we develop the theory of pullbacks in 2-coherent wild categories.
These jointly generalize comma objects in (2,\,1)-categories (including 1-categorical pullbacks) as well as type theoretic homotopy pullbacks~\cite{akl:15:hlimits}.

\subsection{Commuting squares}

\begin{definition}[Cospans]\label{def:cospans}
  A \emph{cospan} in a wild category \(\C\) is an element of
  \[ \Cospan(\C) \defeq \Sigmatype*[(A, B, C \oftype \Ob\C)][(f \oftype \Hom[\C]AC)\,(g \oftype \Hom[\C]BC)]. \]
\end{definition}

We denote cospans graphically as \(\begin{tikzcd}[cramped,sep=1.2em]A \rar["f\,"] & C & \lar["\hspace{1ex}g",swap] B\end{tikzcd}\),
or simply by \((f,g)\) when the vertices are understood.

\begin{definition}[Commuting squares]\label{def:commuting-squares}
  A \emph{commuting square} in a wild category \(\C\) consists of
  \begin{itemize}
    \item a \emphb{cospan} \(\begin{tikzcd}[cramped,sep=1.2em]A \rar["f\,"] & C & \lar["\hspace{1ex}g",swap] B\end{tikzcd}\),
    \item a \emphb{source} object \(X \oftype \Ob\C\),
    \item two \emphb{legs} \(m_A \oftype \Hom[\C]XA\) and \(m_B \oftype \Hom[\C]XB\),
    \item and an \emphb{equality 2-cell} \(\gamma \oftype f \catcomp m_A = g \catcomp m_B\)
  \end{itemize}
  as in the following diagram in \(\C\),
  \[\begin{tikzcd}
    X \dar["m_A" swap] \rar["m_B"]
    & B \dar["g"] \\
    A \rar["f",swap] \ar[ur, "\gamma" swap, path] & C
  \end{tikzcd}.\]
\end{definition}

First, we consider the type of commuting squares on a cospan with fixed source.

\begin{definition}[Commuting squares, with fixed source and cospan]\label{def:commsq-fixed-cospan-source}
  Indexing over cospans
  \[\fc \defeq \begin{tikzcd}[cramped,sep=1.2em]A \rar["f\,"] & C & \lar["\hspace{1ex}g",swap] B\end{tikzcd}\]
  and source objects \(X \oftype \Ob\C\),
  we define the type
  \[ \CommSq_{\fc}(X) \defeq
  \Sigmatype*[(m_A \oftype \Hom[\C]XA)\,(m_B \oftype \Hom[\C]XB)][(\gamma \oftype f \catcomp m_A = g \catcomp m_B)]
  \]
  of \emph{commuting squares on \(\fc\) with source \(X\)}.
\end{definition}

Avigad, Kapulkin and Lumsdaine~\cite{akl:15:hlimits} also call \(\CommSq_{\fc}(X)\) (instantiated in universes) the type of \emph{cones} over \(\fc\) with vertex \(X\).

Characterizing the equality of \(\CommSq_{\fc}(X)\) is routine.

\begin{proposition}[Equality of \(\CommSq_{\fc}(X)\)]\label{prop:commsq-fixed-cospan-source-equality}
  Let
  \(\fc \equiv \begin{tikzcd}[cramped,sep=1.2em]A \rar["f\,"] & C & \lar["\hspace{1ex}g",swap] B\end{tikzcd}\)
  be a cospan, \(X \oftype \Ob\C\) an object, and
  \[\fS \equiv (m_A, m_B, \gamma) \eqntext{and} \primed\fS \equiv (\primed{m_A}, \primed{m_B}, \primed\gamma)\]
  be elements of \(\CommSq_{\fc}(X)\) in a wild category \(\C\).
  Then the equality type \(\fS = \primed\fS\) is equivalent to
  \begin{align*}
    \Sigmatype*[
      (e_A \oftype m_A = \primed{m_A})\,(e_B \oftype m_B = \primed{m_B})
      ]*[
      \gamma = (f \whisker e_A) \pathcomp \primed\gamma \pathcomp \pathinv{(g \whisker e_B)}
      ].
  \end{align*}
\end{proposition}
\begin{proof}
  By a routine application of the fundamental theorem of identity types~\cite[Theorem 11.2.2]{rijke:22:intro-hott}, the algebra of \(\Sigmatype\)-types, and \Cref{prop:whiskering-properties}.
\end{proof}

We have the following operations on commuting squares.

\begin{definition}[Transpose]\label{def:transpose}
  For any cospan
  \(\begin{tikzcd}[cramped,sep=1.2em]A \rar["f\,"] & C & \lar["\hspace{1ex}g",swap] B\end{tikzcd}\)
  and \(X \oftype \Ob\C\), the \emph{transpose} map
  \[ \transpose{\blank} \oftype \CommSq_{(f,g)}(X) \to \CommSq_{(g,f)}(X) \]
  is given by
  \[ \transpose{(m_A, m_B, \gamma)} \defeq (m_B, m_A, \pathinv\gamma). \]
\end{definition}

\begin{proposition}[Transpose is an equivalence]\label{prop:transpose-equivalence}
  Transpose is involutive: \((\fS^\mathrm{T})^\mathrm{T} = \fS\) for all \(\fS\), and so in particular \(\transpose\blank\) is an equivalence.
\end{proposition}

\begin{definition}[Horizontal and vertical pasting]\label{def:pasting}
  From a diagram
  \[\begin{tikzcd}
    \primed{A} \dar["i" swap] \rar["\primed{f}"]
    & \primed{B} \dar["j" swap] \rar["\primed{g}"]
    & \primed{C} \dar["k"]
    \\
    A \rar["f" swap] \ar[ur, "\fq" swap, path]
    & B \rar["g" swap] \ar[ur, "\fp" swap, path]
    & C
  \end{tikzcd}\]
  of commuting squares \(\fQ \defeq (i, \primed{f}, \fq)\) and \(\fP \defeq (j, \primed{g}, \fp)\) we get the \emph{horizontal pasting}
  \[
    \squarehpaste\fQ\fP \eqndefeq
    \begin{tikzcd}
        \primed{A} \dar["i" swap] \rar["\primed{g} \catcomp \primed{f}"]
        & \primed{C} \dar["k"]
        \\
        A \rar["g \catcomp f" swap] \ar[ur, path, "\fq\mid\fp" swap]
        & C
    \end{tikzcd}
  \]
  where
  \(
    \fq\mid\fp \eqndefeq
    \associator \pathcomp (g \whisker \fq) \pathcomp \invassociator \pathcomp (\fp \whisker \primed{f}) \pathcomp \associator
  \).
  Similarly, from a diagram
  \[\begin{tikzcd}
    \primed{A} \dar["\primed{f}" swap] \rar["i"] & A \dar["f"] \\
    \primed{B} \dar["\primed{g}" swap] \rar["j"] \ar[ur, "\fq" swap, path]
    & B \dar["g"] \\
    \primed{C} \rar["k" swap] \ar[ur, "\fp" swap, path]
    & C
  \end{tikzcd}\]
  of commuting squares \(\fQ \defeq (\primed{f}, i, \fq)\) and \(\fP \defeq (\primed{g}, j, \fp)\) we get the \emph{vertical pasting}
  \[ \squarevpaste\fQ\fP \eqndefeq
    \begin{tikzcd}
      \primed{A} \dar["\primed{g} \catcomp \primed{f}" swap] \rar["i"]
      & A \dar["g \catcomp f"] \\
      \primed{C} \rar["k" swap] \ar[ur, "\squarevpaste{\fq}{\fp}" swap, path]
      & C
    \end{tikzcd}
  \]
  where
  \( \squarevpaste\fq\fp \eqndefeq
    \invassociator \pathcomp (\fp \whisker \primed{f}) \pathcomp \associator
    \pathcomp (g \whisker \fq) \pathcomp \invassociator\).
\end{definition}

\Cref{def:pasting} is in some ways redundant: a straightforward calculation shows that
  \[ \squarevpaste\fQ\fP = (\squarehpaste{\transpose\fQ}{\transpose\fP})^\mathrm{T}, \]
  and so one could simply vertically paste squares by horizontally pasting their transposes.
  However, this equality is only propositional, and it turns out to be more convenient later to use the canonical form of the vertical pasting as we have defined it.

\begin{definition}[Vertical pasting map]\label{def:vertical-pasting-map}
  For any \(A \oftype \Ob\C\), \(f \oftype \Hom[\C]AB\) and \(X \oftype \Ob\C\), the vertical pasting with
  \[ \fP \eqndefeq
    \begin{tikzcd}
      \primed{B} \dar["\primed{g}" swap] \rar["j"] & B \dar["g"] \\
      \primed{C} \rar["k" swap] \ar[ur, "\fp" swap, path]
      & C
    \end{tikzcd}
  \]
  yields a map
  \(\CommSq_{(j,f)}(X) \to \CommSq_{(k, g \catcomp f)}(X)\).
  That is, we have the family
  \[ \squarevpasteP \oftype \Pitype*[(A \oftype \Ob\C)\,(f \oftype \Hom[\C]AB)\,(X \oftype \Ob\C)][\CommSq_{(j,f)}(X) \to \CommSq_{(k, g \catcomp f)}(X)]. \]
\end{definition}

Morphisms in \(\C\) act contravariantly on commuting squares by precomposition at their source.

\begin{definition}[Precomposing squares with morphisms]\label{def:commsq-precomp}
  For any cospan \(\fc \equiv \begin{tikzcd}[cramped,sep=1.2em]A \rar["f\,"] & C & \lar["\hspace{1ex}g",swap] B\end{tikzcd}\)
  and \(X, Y \oftype \Ob\C\), there is a \emph{precomposition} map
  \[ \blank \sqcomp \blank \oftype \CommSq_{\fc}(Y) \to \Hom[\C]XY \to \CommSq_{\fc}(X) \]
  defined by
  \((m_A, m_B, \gamma) \sqcomp m \eqndefeq (m_A \catcomp m,\ m_B \catcomp m,\ \pathinv\associator \pathcomp (\gamma \whisker m) \pathcomp \associator)\).
\end{definition}

\begin{lemma}[Right action of morphisms on commuting squares]\label{lem:commsq-right-action-morphisms}
  If \(\fS \oftype \CommSq_{\fc}(X)\) is a commuting square in a 2-coherent wild category then
  \[ \fS \sqcomp \id_X = \fS, \]
  and for all \(f \oftype \Hom[\C]XY\) and \(g \oftype \Hom[\C]YZ\),
  \[ \fS \sqcomp (g \catcomp f) = \fS \sqcomp g \sqcomp f. \]
\end{lemma}
\begin{proof}
  By calculation, properties of whiskering (\Cref{prop:whiskering-properties}), and coherence---namely, the right identity triangle coherence (\Cref{prop:more-coherences-2-coh-wild-cats}) for the first claim, and the pentagon coherence for the second.
\end{proof}

\begin{corollary}\label{cor:commsq-transport}
  By induction on \(e\) and \Cref{lem:commsq-right-action-morphisms}, 
  \[ \fS \transpover[\CommSq_{\fc}(\blank)]{e} =_{\CommSq_{\fc}(\primed{X})} \fS \sqcomp \idd(\pathinv{e}) \]
  for every \(\fS \oftype \CommSq_{\fc}(X)\) and \(e \oftype X = \primed{X}\).
\end{corollary}

Now we can characterize the identity of commuting squares on a cospan, allowing their source objects to vary.

\begin{definition}[Commuting squares on a cospan]\label{def:commsq-fixed-cospan}
  The type of commuting squares on a cospan \(\fc\) is the total space
  \[ \CommSq(\fc) \defeq \Sigmatype*[(X \oftype \Ob\C)][\CommSq_{\fc}(X)]. \]
\end{definition}

\begin{corollary}[Equality of \(\CommSq(\fc)\)]\label{lem:commsq-fixed-cospan-equality}
  Suppose that \((X, \fS)\) and \((\primed{X}, \primed{\fS})\) are two commuting squares on a cospan \(\fc \defeq (f, g)\).
  The equality
  \[ (X, \fS) =_{\CommSq(\fc)} (\primed{X}, \primed{\fS}) \]
  is equivalent to
  \[ \Sigmatype*[(e \oftype X = \primed{X})]*[\fS = \primed{\fS} \sqcomp \idd(e)]. \]
\end{corollary}
\begin{proof}
  By the equality of \(\Sigmatype\)-types and \Cref{cor:commsq-transport}.
\end{proof}

Up to this point we have considered commuting squares with respect to a fixed cospan
\((A, B, C, f, g) \oftype \Cospan(\C)\).
We will also need to compare squares on cospans with propositionally equal but definitionally distinct legs \((f,g)\) and \((\primed{f}, \primed{g})\).
To this end, we index the type of commuting squares over their vertices, and consider the following type family.

\begin{definition}[Commuting squares, with fixed vertices]
  For \(X, A, B, C \oftype \Ob\C\), define
  \[ \CommSqV(X, A, B, C) \defeq \Sigmatype*[(f \oftype \Hom[\C]AC)\,(g \oftype \Hom[\C]BC)]*[\CommSq_{(f,g)}(X)]. \]
  That is, \(\CommSqV(X, A, B, C)\) is the total space of \(\CommSq_{(A,B,C,f,g)}(X)\) fibered over \(f \oftype \Hom[\C]AC\) and \(g \oftype \Hom[\C]BC\).
\end{definition}

We characterize the identity type of \(\CommSqV(X, A, B, C)\) using the following version of Rijke's structure identity principle~\cite[Theorem 11.6.2]{rijke:22:intro-hott}, which may be understood as a dependent form of the fundamental theorem of identity types.

\begin{theorem}[Structure identity principle]\label{thm:sip-rijke}
  Suppose that \(A\) is a type pointed at \(a \oftype A\), and \(B \oftype A \to \UniverseType\) is a type family over \(A\), pointed at \(b \oftype B(a)\).
  Then for any type family
  \[ R \oftype \Pitype*[(x \oftype A)][a = x \to B(x) \to \UniverseType] \]
  pointed at \(r \oftype R(a, \refl_a, b)\), the canonical family of maps indexed over \(x\) and \(y\)
  \[ \Pitype*[(x \oftype A)\,(y \oftype B(x))][(a, b) = (x, y) \to \Sigmatype[(p \oftype a = x)]*[R(x, p, y)]] \]
  is a family of equivalences if and only if the total space \(\Sigmatype[(B(a))][R(a, \refl_a)]\) is contractible.
\end{theorem}

\begin{lemma}[Equality of \(\CommSqV(X, A, B, C)\)]\label{lem:commsq-fixed-vertices-equality}
  Let
  \[ (f,g,\fS),\ (\primed{f},\primed{g},\primed{\fS}) \oftype \CommSqV(X,A,B,C) \]
  be commuting squares with vertices \(X, A, B, C \oftype \Ob\C\), where \(\fS \equiv (m_A, m_B, \gamma)\) and \(\primed{\fS} \equiv (\primed{m_A}, \primed{m_B}, \primed\gamma)\).
  The equality \((f,g,\fS) = (\primed{f},\primed{g},\primed{\fS})\) is equivalent to
  \begin{gather*}
    \Sigmatype*[(e_f \oftype f = \primed{f})\,(e_g \oftype g = \primed{g})\,(e_A \oftype m_A = \primed{m_A})\,(e_B \oftype m_B = \primed{m_B})]\hspace{1pt},\\
    \gamma = (e_f \whisker m_A) \pathcomp (\primed{f} \whisker e_A) \pathcomp \primed\gamma
      \pathcomp \pathinv{(\primed{g} \whisker e_B)} \pathcomp \pathinv{(e_g \whisker m_B)}.
  \end{gather*}
  This last component is a proof that
  \[\begin{tikzcd}[paths, column sep=0ex]
    & f \catcomp m_A \ar[dl, "e_f \whisker m_A"{swap, yshift=-2pt}] \ar[dr, "\gamma" yshift=-1pt] &
    \\[-10pt]
    \primed{f} \catcomp m_A \dar["\primed{f} \whisker e_A" swap]
    && g \catcomp m_B \dar["e_g \whisker m_B"]
    \\
    \primed{f} \catcomp \primed{m_A} \ar[dr, "\primed\gamma"{swap, yshift=1pt}]
    && \primed{g} \catcomp m_B \ar[dl, "\primed{g} \whisker e_B" yshift=2pt]
    \\[-10pt]
    & \primed{g} \catcomp \primed{m_B} &
  \end{tikzcd}\]
  commutes.
\end{lemma}
\begin{proof}
  Use the structure identity principle (\Cref{thm:sip-rijke}) applied to the pointed type
  \((\Hom[\C]AC \times \Hom[\C]BC,\ (f, g))\)
  and the type family
  \[ \CommSq_{(A,B,C,\,\blank\,,\,\blank)}(X) \oftype \Hom[\C]AC \times \Hom[\C]BC \to \UniverseType \]
  pointed at \(\fS \oftype \CommSq_{(f,g)}(X)\).
  We define (in curried form)
  \[ R \oftype \Pitype*[(\primed{f} \oftype \Hom[\C]AC)\,(\primed{g} \oftype \Hom[\C]BC)][
      (f = \primed{f}) \to (g = \primed{g}) \to \CommSq_{(\primed{f},\primed{g})}(X) \to \UniverseType
    ] \]
  such that \(R(\primed{f}, \primed{g}, e_f, e_g)(k_A, k_B, \delta)\) is the type
  \begin{gather*}
    \Sigmatype*[(e_A \oftype m_A = k_A)\,(e_B \oftype m_B = k_B)]\hspace{1pt},\\
      \gamma = (e_f \whisker m_A) \pathcomp (\primed{f} \whisker e_A) \pathcomp \delta
        \pathcomp \pathinv{(\primed{g} \whisker e_B)} \pathcomp \pathinv{(e_g \whisker m_B)},
  \end{gather*}
  pointed at \((\refl_{m_A}, \refl_{m_B}, \refl_\gamma) \oftype R(f,g,\refl_f,\refl_g,\fS)\).
  Then it's enough to show that \(\Sigmatype[(\CommSq_{(f,g)}(X))][R(f,g,\refl_f,\refl_g)]\) is contractible, which we have already done in the proof of \Cref{prop:commsq-fixed-cospan-source-equality}.
\end{proof}

The following somewhat ad hoc-looking technical result is used repeatedly later.

\begin{proposition}\label{prop:equality-comm-sq-id-leg}
  For any commuting square
  \[\begin{tikzcd}
    X \rar["m_B"] \dar["m_A"'] & B \dar["g"] \\
    A \rar["\id"'] \ar[ur,path,"\gamma"'] & A
  \end{tikzcd}\]
  in a 2-coherent wild category, the following equalities in \(\CommSq_{(\id,g)}(X)\) hold:
  \[
    \begin{tikzcd}
      X \rar["m_B"] \dar["m_A"'] & B \dar["g"] \\
      A \rar["\id"'] \ar[ur,path,"\gamma"'] & A
    \end{tikzcd}
    \eqnequal{}
    \begin{tikzcd}
      X \rar["m_B"] \dar["g \catcomp m_B"'] & B \dar["g"] \\
      A \rar["\id"'] \ar[ur,path,"\lunitor"'] & A
    \end{tikzcd}
    \eqnequal{}
    \begin{tikzcd}
      X \rar["\id \catcomp m_B"{xshift=-2pt}] \dar["g \catcomp m_B"'] & B \dar["g"] \\
      A \rar["\id"'] \ar[ur,path,"e"'] & A
    \end{tikzcd}\,,
  \]
  where \(e \defeq \invassociator \pathcomp (\lunitor \pathcomp \pathinv\runitor) \whisker m_B \pathcomp \associator\).
\end{proposition}
\begin{proof}
  Using \Cref{prop:commsq-fixed-cospan-source-equality}, to show that \((m_A, m_B, \gamma) = (g \catcomp m_B, m_B, \lunitor)\) observe that
  \(\pathinv\lunitor \pathcomp \gamma \oftype m_A = g \catcomp m_B\)
  and that
  \[ \gamma = (\id \whisker (\pathinv\lunitor \pathcomp \gamma)) \pathcomp \lunitor \]
  by properties of whiskering and the coherence in \Cref{prop:a-coh-for-lunitor}.
  To show that \((g \catcomp m_B, m_B, \lunitor) = (g \catcomp m_B, \id \catcomp m_B, e)\), take \(\pathinv\lunitor \oftype m_B = \id \catcomp m_B\) and observe that
  \[ \lunitor = e \pathcomp (g \whisker \lunitor) \]
  by properties of whiskering and the triangle coherators.
\end{proof}

\subsection{Pullbacks}

\begin{definition}[Pullbacks]\label{def:pullbacks}
  Let
  \[ \fP \eqndefeq
    \begin{tikzcd}
      P \rar["\pi_B"] \dar["\pi_A" swap] & B \dar["g"] \\
      A \rar["f" swap] \ar[ur, "\fp" swap, path]
      & C
    \end{tikzcd}\]
  be a commuting square on
  \(\fc \defeq \begin{tikzcd}[cramped,sep=1.2em]A \rar["f\,"] & C & \lar["\hspace{1ex}g"'] B\end{tikzcd}\)
  with source \(P\), in a wild category \(\C\).
  By specializing the precomposition map (\Cref{def:commsq-precomp}) at \(\fP\), we obtain the family of maps
  \[ \Psqcomp{\blank} \oftype \Pitype*[(X \oftype \Ob\C)][\Hom[\C]XP \to \CommSq_{\fc}(X)]. \]
  We say that \(\fP\) is a \emph{pullback of \(\fc\)} if \((\Psqcomp{\blank})\) is a family of equivalences, and a \emph{weak pullback of \(\fc\)} if \((\Psqcomp{\blank})\) is a family of retractions, or split surjections.
\end{definition}

\begin{corollary}[Universal property of (weak) pullbacks]\label{cor:pullback-universal-prop}
  By the characterization of equality of \(\CommSq_{\fc}(X)\) (\Cref{prop:commsq-fixed-cospan-source-equality}), for each \(X \oftype \Ob\C\) and commuting square \(\fS \defeq (m_A, m_B, \gamma)\) on \(\fc\) with source \(X\), the fiber of \((\Psqcomp{X})\) at \(\fS\) is equivalent to
  \begin{gather*}
    \Sigmatype*
      [(m \oftype \Hom[\C]XP)\,
      (e_A \oftype \pi_A \catcomp m = m_A)\,
      (e_B \oftype \pi_B \catcomp m = m_B)]\hspace{1pt},\\
    \invassociator \pathcomp (\fp \whisker m)
      \pathcomp \associator = (f \whisker e_A) \pathcomp \gamma \pathcomp \pathinv{(g \whisker e_B)}.
  \end{gather*}
  Thus \(\fP\) is a pullback (respectively, a weak pullback) when this type is contractible (respectively, pointed) for every \(X\) and \(\fS\).
\end{corollary}

Being a pullback is evidently a property: for any cospan \(\fc\) and \(P \oftype \Ob\C\), the predicate
\[ \ispb_{\fc,\,P}(\fP) \defeq \Pitype*[(X \oftype \Ob\C)][\isequiv \, (\fP \sqcomp_X \blank)] \]
on \(\CommSq_{\fc}(P)\) is propositional.

\begin{proposition}[Pullbacks are closed under transpose]\label{prop:pullback-transpose}
  \(\transpose\fP\) is a (weak) pullback if \(\fP\) is.
\end{proposition}
\begin{proof}
  By a straightforward calculation,
  \[ (\transpose{\fP} \sqcomp_{X} \blank) = (\transpose\blank) \comp (\Psqcomp{X}) \]
  for all \(X \oftype \Ob\C\).
  Since \(\transpose\blank\) is an equivalence (\Cref{prop:transpose-equivalence}), \((\transpose{\fP} \sqcomp_{X} \blank)\) is an equivalence (respectively, a retraction) when \((\Psqcomp{X})\) is.
\end{proof}

\begin{proposition}[Identity pullbacks]\label{prop:id-pullback}
  If \(\C\) is a 2-coherent wild category, then for all \(A, B \oftype \Ob\C\) and \(f \oftype \Hom[\C]AB\) the commuting square
  \[ \idpb{f} \defeq
    \begin{tikzcd}
      A \rar["\id"] \dar["f"'] & A \dar["f"] \\
      B \rar["\id"'] \ar[ur,"\lunitor \pathcomp \pathinv\runitor"'{xshift=-6pt,yshift=-2pt},path] & B
    \end{tikzcd}
  \]
  is a pullback.
\end{proposition}
\begin{proof}
  For any \(X \oftype \Ob\C\), a straightforward calculation shows that the map
  \[ \idpb{f} \sqcomp_X \blank \oftype \Hom[\C]XA \to \CommSq_{(\id,f)}(X) \]
  has retraction \(\snd\).
  By \Cref{prop:equality-comm-sq-id-leg}, this retraction is also a section, i.e.
  \[ \idpb{f} \sqcomp_X k \equiv (f \catcomp k, \id \catcomp k, \invassociator \pathcomp ((\lunitor \pathcomp \pathinv\runitor) \whisker k) \pathcomp \associator) = (h,k,\gamma) \]
  for any \((h,k,\gamma) \oftype \CommSq_{(\id,f)}(X)\).
\end{proof}

\begin{remark}\label{rem:comm-sq-id-leg-idpb-sqcomp}
  In the notation of \Cref{prop:id-pullback}, \Cref{prop:equality-comm-sq-id-leg} says that any commuting square \((m_A, m_B, \gamma)\) on a cospan
  \(\begin{tikzcd}[cramped,sep=1.2em]A \rar["\id\hspace{.75ex}"] & A & \lar["\hspace{1ex}g",swap] B\end{tikzcd}\)
  is equal to \(\idpb{g} \sqcomp m_B\).
\end{remark}

The next lemma is inspired by the proof of \cite[Proposition 4.1.11]{akl:15:hlimits} and used in the proof of the pullback pasting lemma (\Cref{lem:vertical-pullback-pasting}).

\begin{lemma}[Pasting maps of (weak) pullbacks]\label{lem:pasting-map-pullbacks}
  A commuting square \(\fP\) in a 2-coherent wild category \(\C\) is a pullback (respectively, a weak pullback) if and only if the vertical pasting map \(\squarevpasteP\) (\Cref{def:vertical-pasting-map}) is a family of equivalences (respectively, retractions).
\end{lemma}
\begin{proof}
  Let \(\fP \defeq (\primed{g}, j, \fp)\) be a commuting square on \((k,g)\) as in \Cref{def:vertical-pasting-map}.
  For any \(A \oftype \Ob\C\), \(f \oftype \Hom[\C]AB\) and \(X \oftype \Ob\C\), the fiber of \(\squarevpasteP_{A,f,X}\) at
  \[ \fX \eqndefeq
    \begin{tikzcd}
      X \dar["m_{\primed{C}}" swap] \rar["m_A"]
      & A \dar["g \catcomp f"]
      \\
      \primed{C} \rar["k" swap] \ar[ur, "\xi" swap, path]
      & C
    \end{tikzcd}
  \]
  is equivalent to the \(\Sigmatype\)-type
  \begin{gather*}
    \Sigmatype*[
      (m \oftype \Hom[\C]{X}{\primed{B}})\,
      (i \oftype \Hom[\C]XA)\,
      (\gamma \oftype j \catcomp m = f \catcomp i)]\hspace{1pt},\\
    (e_{\primed{C}} \oftype \primed{g} \catcomp m = m_{\primed{C}}) \times
    (e_A \oftype i = m_A) \\
    \times\ 
    (\invassociator \pathcomp (\fp \whisker m)
      \pathcomp \associator \pathcomp (g \whisker \gamma)
      \pathcomp \invassociator
    = (k \whisker e_{\primed{C}}) \pathcomp \xi \pathcomp \pathinv{(g \catcomp f \whisker e_A)}),
  \end{gather*}
  by the equality characterization of \Cref{prop:commsq-fixed-cospan-source-equality}.
  Contracting the singleton formed by the components \(i\) and \(e_A\), this is equivalent to
  \begin{gather*}
    \Sigmatype*[
      (m \oftype \Hom[\C]X{\primed{B}})
      (e_{\primed{C}} \oftype \primed{g} \catcomp m = m_{\primed{C}})
      (e_B \oftype j \catcomp m = f \catcomp m_A)]\hspace{1pt},\\
    \invassociator \pathcomp (\fp \whisker m) \pathcomp \associator
    = (k \whisker e_{\primed{C}}) \pathcomp (\xi \pathcomp \associator) \pathcomp \pathinv{(g \whisker e_B)}.
  \end{gather*}
  But this type is also the fiber of the precomposition map \((\Psqcomp{X})\) at the commuting square
  \[\begin{tikzcd}
    X \dar["m_{\primed{C}}" swap] \rar["f \catcomp m_A"{xshift=-2pt}]
    & B \dar["g"]
    \\
    \primed{C} \rar["k" swap] \ar[ur, "\xi \pathcomp \associator" swap, path]
    & C
  \end{tikzcd}\]
  obtained by ``reparenthesizing'' the diagram \(\fX\).
  Thus if \(\fP\) is a pullback (respectively, a weak pullback) then by its universal property (\Cref{cor:pullback-universal-prop}) the fiber \(\fiber{\big(\squarevpasteP_{A,f,X}\big)}(\fX)\) is contractible (respectively, pointed).

  Conversely, for any \(X \oftype \Ob\C\), we claim that the map
  \begin{gather*}
    \varphi \oftype \Hom[\C]X{\primed{B}} \to \CommSq_{(j, \id_B)}(X) \\
    \varphi(m) \defeq (m, j \catcomp m, \pathinv{\lunitor_{j \catcomp m}})
  \end{gather*}
  is an equivalence, and that the diagram of types and functions
  \[\begin{tikzcd}
    & \CommSq_{(j, \id_B)}(X) \dar["\squarevpasteP_{B, \id_B, X}"]
    \\
    \Hom[\C]X{\primed{B}}
      \ar[ur, "\varphi", "\sim" swap]
      \ar[dr, "\Psqcomp{X}" swap]
    & \CommSq_{(k, g \catcomp \id_B)}(X) \dar["\psi", "\sim" swap]
    \\
    & \CommSq_{(k,g)}(X)
  \end{tikzcd}\]
  commutes, where \(\psi\) is the equivalence
  \( (m_{\primed{C}}, m_B, \gamma) \mapsto (m_{\primed{C}}, m_B, \gamma \pathcomp (\runitor \whisker m_B)) \).
  That is, \((\Psqcomp{X})\) is the pre- and post-composition of \(\squarevpasteP_{B,\id_B,X}\) by equivalences.
  Thus, if \(\squarevpasteP\) is a family of equivalences then so is \((\Psqcomp\blank)\), and if \(\squarevpasteP\) is a family of retractions then so is \((\Psqcomp\blank)\).

  Now, the map \(\varphi\) is clearly a section of
  \(\fst \oftype \CommSq_{(j, \id)}(X) \to \Hom[\C]X{\primed{B}}\).
  By \Cref{prop:equality-comm-sq-id-leg} and transposition we see that it's also a retraction of \(\fst\), i.e.\ that
  \[ \varphi(m_{\primed{B}}) \equiv (m_{\primed{B}}, j \catcomp m_{\primed{B}}, \pathinv\lunitor) = (m_{\primed{B}}, m_B, \gamma) \]
  for all \(m_{\primed{B}} \oftype \Hom[\C]X{\primed{B}}\), \(m_B \oftype \Hom[\C]XB\) and \(\gamma \oftype j \catcomp m_{\primed{B}} = \id \catcomp m_B\).

  Finally, given \(m \oftype \Hom[\C]X{\primed{B}}\) we calculate that
  \[ \fS \defeq \textstyle \big(\psi \comp \squarevpasteP_{B,\id_B,X} \comp \varphi\big)(m) \]
  and
  \[ \primed{\fS} \defeq\fP \sqcomp_X m \]
  are commuting squares of type \(\CommSq_{(k,g)(X)}\) with the same morphism components \(\primed{g} \catcomp m\) and \(j \catcomp m\).
  The commutativity witness of \(\fS\) is
  \[ \pathinv\associator \pathcomp (\fp \whisker m) \pathcomp \associator
      \pathcomp (g \whisker \pathinv\lunitor) \pathcomp \pathinv\associator \pathcomp (\runitor \whisker (j \catcomp m)), \]
  while that of \(\primed{\fS}\) is
  \[ \pathinv\associator \pathcomp (\fp \whisker m) \pathcomp \associator, \]
  and these are equal since \((g \whisker \pathinv\lunitor) \pathcomp \pathinv\associator \pathcomp (\runitor \whisker (j \catcomp m)) = \refl\) by the triangle coherator.
\end{proof}

\begin{lemma}[Vertical pullback pasting]\label{lem:vertical-pullback-pasting}
  Suppose we have a diagram
  \[\begin{tikzcd}
    \primed{A} \dar["\primed{f}" swap] \rar["i"] & A \dar["f"] \\
    \primed{B} \dar["\primed{g}" swap] \rar["j"] \ar[ur, "\fq" swap, path]
    & B \dar["g"] \\
    \primed{C} \rar["k" swap] \ar[ur, "\fp" swap, path]
    & C
  \end{tikzcd}\]
  in a 2-coherent wild category \(\C\).
  Then if \(\fP \defeq (\primed{g}, j, \fp)\) is a pullback of \((k,g)\), the commuting square \(\fQ \defeq (\primed{f}, i, \fq)\) is a pullback of \((j,f)\) if and only if the vertical pasting \(\squarevpaste\fQ\fP\) is a pullback of \((k, g \catcomp f)\).
\end{lemma}
\begin{proof}
  We claim that for any \(X \oftype \Ob\C\), the triangle
  \[\begin{tikzcd}[row sep=small]
    & \CommSq_{(j,f)}(X)
      \ar[dd, "\squarevpasteP_{A,f,X}", start anchor={[xshift=-2ex]}, end anchor={[xshift=-2ex]}]
    \\
    \Hom[\C]X{\primed{A}}
      \ar[ur, "\fQ \sqcomp_X \blank", end anchor={[xshift=-1ex]}]
      \ar[dr, "\squarevpaste\fQ\fP \sqcomp_X \blank" swap, end anchor={[xshift=-2ex]}]
      &
    \\
    & \CommSq_{(k, g \catcomp f)}(X)
  \end{tikzcd}\]
  commutes.
  Then since \(\fP\) is a pullback, the map \(\squarevpasteP_{A,f,X}\) is an equivalence (\Cref{lem:pasting-map-pullbacks}), and it follows that \((\fQ \sqcomp \blank)\) is a family of equivalences if and only if \((\squarevpaste\fQ\fP \sqcomp \blank)\) is.

  What remains is to construct a homotopy
  \(\big(\squarevpaste\fQ\fP \sqcomp_X \blank\big) = \big(\squarevpasteP_{A,f,X}\big) \comp (\fQ \sqcomp_X \blank)\)
  for any \(X\), i.e.\ a witness that, for any \(m \oftype \Hom[\C]X{\primed{A}}\), the commuting squares
  \[ \squarevpaste\fQ\fP \sqcomp_X m \eqnequiv
    \left((\primed{g} \catcomp \primed{f}) \catcomp m,\ i \catcomp m,\ 
      \invassociator \pathcomp \big(\textstyle\squarevpaste\fq\fp \whisker m\big) \pathcomp \associator
    \right)
  \]
  and
  \[ \squarevpaste{\fQ \sqcomp_X m}{\fP} \eqnequiv
    \Big(\primed{g} \catcomp \primed{f} \catcomp m,\ i \catcomp m,\ 
      \squarevpaste{\invassociator \pathcomp (\fq \whisker m) \pathcomp \associator}{\fp}
    \Big)
  \]
  are equal.
  By \Cref{prop:commsq-fixed-cospan-source-equality} together with the canonical equalities \(\associator \oftype (\primed{g} \catcomp \primed{f}) \catcomp m = \primed{g} \catcomp \primed{f} \catcomp m \) and \(\refl \oftype i \catcomp m = i \catcomp m\), it's enough to show that
  \[ \invassociator \pathcomp \big(\textstyle\squarevpaste\fq\fp \whisker m\big) \pathcomp \associator
    = (k \whisker \associator) \pathcomp \squarevpaste{\invassociator \pathcomp (\fq \whisker m) \pathcomp \associator}{\fp}. \]

  \input{diag-pullback-pasting}

  With a little path algebra (noting \Cref{prop:whiskering-properties}) this amounts to showing commutativity of the outer boundary of \Cref{diag:pullback-pasting}.
  By inserting associators \(\associator\) as shown in the interior of the diagram, we decompose the outer shape into a pasting of three commuting pentagons (by the pentagonators) and two commuting squares (by \Cref{prop:whiskering-properties}).
  Thus the entire diagram commutes.
\end{proof}

\begin{corollary}[Horizontal pullback pasting]
  Since the transpose of a pullback is a pullback (\Cref{prop:pullback-transpose}), by taking transposes as appropriate we deduce the more familiar horizontal version of the pullback pasting lemma.
\end{corollary}

\begin{lemma}[Pullback prism]\label{lem:pullback-prism}
  Suppose we have a diagram
  \[\begin{tikzcd}
    & P \ar[rr, "\pi_B"] \ar[dd, "\pi_A" near end]
    && B \ar[dd, "g"]
    \\
    \primed{P} \ar[rr, "\pi_{\primed{B}}", crossing over, near end]
      \ar[dr, "\primed{\pi_A}" swap]
    && \primed{B} \ar[dr, "\primed{g}" swap] \ar[ur, "h"]
    \\
    & A \ar[rr, "f" swap] && C
  \end{tikzcd}\]
  in a 2-coherent wild category \(\C\), such that \(c \oftype g \catcomp h = \primed{g}\) is a commuting triangle,
  \(\fp \oftype f \catcomp \pi_A = g \catcomp \pi_B\) and \(\primed\fp \oftype f \catcomp \primed{\pi_A} = \primed{g} \catcomp \pi_{\primed{B}}\),
  and where the squares
  \(\fP \defeq (\pi_A, \pi_B, \fp) \oftype \CommSq_{(f,g)}(P)\)
  and \(\primed\fP \defeq (\primed{\pi_A}, \pi_{\primed{B}}, \primed\fp) \oftype \CommSq_{(f,\primed{g})}(\primed{P})\)
  are both pullbacks.
  Then there is a contractible type of data consisting of:
  \begin{itemize}
    \item a morphism \(m \oftype \Hom[\C]{\primed{P}}P\),
    \item equalities
      \(e \oftype \pi_A \catcomp m = \primed{\pi_A}\)
      and \(\fq \oftype \pi_B \catcomp m = h \catcomp \pi_{\primed{B}}\) completing the boundary of the prism, and
    \item an equality 3-cell \(\eta\) filling the volume of the completed prism.
  \end{itemize}
  Even more, the top face \((m, \pi_{\primed{B}}, \fq)\) of the completed prism is a pullback of \((\pi_B, h)\).
\end{lemma}
\begin{proof}
  From the universal property of \(\fP\) we get \(m\), \(e\), \(\fq\) and the equality
  \[ \eta \oftype \invassociator \pathcomp (\fp \whisker m) \pathcomp \associator
      = (f \whisker e) \pathcomp \primed{\fp} \pathcomp (\pathinv{c} \whisker \pi_{\primed{B}}) \pathcomp \associator \pathcomp \pathinv{(g \whisker \fq)} \]
  as the center of contraction of the fiber of \((\fP \sqcomp_{\primed{P}} \blank)\) at the commuting square
  \[ \fS \defeq \big(\primed{\pi_A},\ h \catcomp \pi_{\primed{B}},\ \primed{\fp} \pathcomp (\pathinv{c} \whisker \pi_{\primed{B}}) \pathcomp \associator\big) \]
  on \((f,g)\).
  Let \(\fQ \defeq (m, \pi_{\primed{B}}, \fq)\); then by \(\eta\) and \Cref{lem:commsq-fixed-vertices-equality} it follows that \((f, g \catcomp h, \squarevpaste{\fQ}{\fP})\) and \((f, \primed{g}, \primed{\fP})\) are equal commuting squares on \(\primed{P}\), \(A\), \(\primed{B}\), \(C\).
  Since \(\primed{\fP}\) is a pullback, by transport so is \(\squarevpaste{\fQ}{\fP}\), and by pullback pasting (\Cref{lem:vertical-pullback-pasting}) so too is \(\fQ\).
\end{proof}

\subsection{The truncation level of pullbacks}


\begin{definition}[Pullbacks on a cospan]\label{def:pullbacks-cospan}
  As observed earlier, the predicate
  \[ \ispb_{\fc,\,P}(\fP) \defeq \Pitype*[(X \oftype \Ob\C)][\isequiv \, (\fP \sqcomp_X \blank)] \]
  on \(\CommSq_{\fc}(P)\) is propositional.
  We obtain the subtype of pullbacks on a cospan \(\fc\), with fixed source
  \[ \Pullback_{\fc}(P) \defeq \Sigmatype*[(\fP \oftype \CommSq_{\fc}(P))][\ispb_{\fc,\,P}(\fP)], \]
  and with arbitrary source,
  \[ \Pullback(\fc) \defeq \Sigmatype*[((P, \fP) \oftype \CommSq(\fc))][\ispb_{\fc,\,P}(\fP)]. \]
\end{definition}

\begin{proposition}[\(\Pullback(\fc)\) is a set in set-level categories]\label{prop:pullback-is-set-set-level-cats}
  If \(\C\) is set-level then \(\Pullback(\fc)\) is a \(\Sigmatype\)-type of sets and propositions for any cospan \(\fc\), and thus also a set.
\end{proposition}

\begin{proposition}[\(\Pullback(\fc)\) is a proposition in univalent 2-coherent wild categories]\label{prop:pullback-is-prop-univalent-wild-cats}%
  If \(\C\) is a univalent 2-coherent wild category, then \(\Pullback(\fc)\) is a proposition for any cospan \(\fc\).
\end{proposition}
\begin{proof}
  In summary, by univalence and the universal property of pullbacks.
  Suppose that \((P, \fP)\) and \((\primed{P}, \primed{\fP})\) are elements of \(\Pullback(\fc)\).
  Then \({(\fP \sqcomp \blank)}\) and \((\primed{\fP} \sqcomp \blank)\) are equivalences, and from the centers of contraction of
  \(\fiber{(\fP \sqcomp \blank)}(\primed{\fP})\)
  and
  \(\fiber{(\primed{\fP} \sqcomp \blank)}(\fP)\)
  we get \(m \oftype \Hom[\C]P{\primed{P}}\) and \(\primed{m} \oftype \Hom[\C]{\primed{P}}P\) such that
  \[ e \oftype \primed{\fP} \sqcomp m = \fP \eqntextspace\text{and}\eqntextspace \primed{e} \oftype \fP \sqcomp \primed{m} = \primed{\fP}. \]

  Furthermore, by \Cref{lem:commsq-right-action-morphisms}
  \[ \fP \sqcomp (\primed{m} \catcomp m) \eqnequal{} (\fP \sqcomp \primed{m}) \sqcomp m \eqnequal{} \primed{\fP} \sqcomp m \eqnequal{} \fP, \]
  and so \(\primed{m} \catcomp m = \id_P\) by contractibility of \(\fiber{(\fP \sqcomp \blank)}(\fP)\) and \Cref{lem:commsq-right-action-morphisms} again.

  By a similar argument \(m \catcomp \primed{m} = \id_{\primed{P}}\), and so \(m \oftype \Hom[\C]P{\primed{P}}\) is a \(\C\)-equivalence.
  From univalence of \(\C\) we now get an equality
  \[ \wildequivtoid\C(m) \oftype P = \primed{P}, \]
  with
  \[ \fP \eqnequal{} \primed{\fP} \sqcomp m \eqnequal{} \primed{\fP} \sqcomp \idd(\wildequivtoid\C(m)). \]
  By \Cref{lem:commsq-fixed-cospan-equality}, this proves \((P, \fP) = (\primed{P}, \primed{\fP})\).
\end{proof}

%% file: diag-pullback-pasting.tex
\begin{diag}
  \caption{Construction of
    \(\invassociator \pathcomp \big(\textstyle\squarevpaste\fq\fp \whisker m\big) \pathcomp \associator
    = (k \whisker \associator) \pathcomp \squarevpaste{\invassociator \pathcomp (\fq \whisker m) \pathcomp \associator}{\fp}\).
  }\label{diag:pullback-pasting}
  \renewcommand{\comp}{\catcomp}
  \newcommand{\pr}{\primed}
  \[\begin{tikzcd}[
	ampersand replacement=\&,
    paths,
    column sep=-2em,
	nodes={row 7/.style={row sep=1ex}}
    ]
    \& [-3.25em] \&\& [-1.25em] \&\& [0.5em] \& |[xshift=1.75em]| {k \comp (\pr{g} \comp \pr{f}) \comp m} \& [2.25em] \&\& [2.25em] \& |[xshift=-1.75em]| {k \comp \pr{g} \comp \pr{f} \comp m} \& [0.5em] \&\& [-1em] \&\& [-2em] \\
	\&\&\&\& {(k \comp \pr{g} \comp \pr{f}) \comp m} \&\&\&\&\&\&\&\& {(k \comp \pr{g}) \comp \pr{f} \comp m} \\
	\&\& {((k \comp \pr{g}) \comp \pr{f}) \comp m} \&\&\&\&\&\&\&\&\&\&\&\& {(g \comp j) \comp \pr{f} \comp m} \\
	{((g \comp j) \comp \pr{f}) \comp m} \&\&\&\&\&\&\&\&\&\&\&\&\&\&\&\& {g \comp j \comp \pr{f} \comp m} \\
	\&\& {(g \comp j \comp \pr{f}) \comp m} \&\&\&\&\&\&\&\&\&\&\&\& {g \comp (j \comp \pr{f}) \comp m} \\
	\&\&\&\& {(g \comp f \comp i) \comp m} \&\&\&\&\&\&\&\& {g \comp (f \comp i) \comp m} \\
	\&\&\&\&\&\& {((g \comp f) \comp i) \comp m} \&\&\&\& {g \comp f \comp i \comp m} \\
	\&\&\&\&\&\&\&\& {(g \comp f) \comp i \comp m}
	\arrow["{k \whisker \associator}", from=1-7, to=1-11]
	\arrow["{\pathinv\associator}"'{yshift=-1.5ex}, from=1-7, to=2-5]
	\arrow["{\pathinv\associator}"{yshift=-1.5ex}, from=1-11, to=2-13]
	\arrow["{\pathinv\associator \whisker m}"'{yshift=-1.5ex,xshift=-2pt}, from=2-5, to=3-3]
	\arrow["{\fp \whisker (\pr{f} \comp m)}"{yshift=-1.75ex,xshift=2pt}, from=2-13, to=3-15]
	\arrow["\associator"{xshift=2em,yshift=0.75ex}, from=3-3, to=2-13]
	\arrow["{(\fp \whisker \pr{f}) \whisker m}"'{yshift=-2ex}, from=3-3, to=4-1]
	\arrow["\associator"{yshift=-1.75ex}, from=3-15, to=4-17]
	\arrow["\associator"{xshift=1.25em,yshift=0.25ex}, from=4-1, to=3-15]
	\arrow["{\associator \whisker m}"'{yshift=1.5ex}, from=4-1, to=5-3]
	\arrow["{g \whisker \pathinv\associator}"{yshift=2.5ex}, from=4-17, to=5-15]
	\arrow["\associator"{xshift=1ex}, from=5-3, to=5-15]
	\arrow["{(g \whisker \fq) \whisker m}"'{yshift=2ex,xshift=-2pt}, from=5-3, to=6-5]
	\arrow["{g \whisker (\fq \whisker m)}"{yshift=2ex,xshift=2pt}, from=5-15, to=6-13]
	\arrow["\associator"{xshift=0.5ex}, from=6-5, to=6-13]
	\arrow["{\pathinv\associator \whisker m}"'{yshift=2ex,xshift=-2pt}, from=6-5, to=7-7]
	\arrow["{g \whisker \associator}"{yshift=1.25ex,xshift=1pt}, from=6-13, to=7-11]
	\arrow["\associator"'{xshift=-2pt}, from=7-7, to=8-9]
	\arrow["{\pathinv\associator}"{yshift=0.5ex,xshift=2pt}, from=7-11, to=8-9]
  \end{tikzcd}\]
\end{diag}

%% file: wild-cwfs.tex
\section[Wild categories with families]{Wild Categories with Families}\label{sec:wild-cwfs}

We can now define precoherent higher internal models of homotopical \MLTT{}.
We begin by simply taking Dybjer's generalized algebraic definition of a category with families~\cite{dybjer:96:internal-tt}, and allowing contexts to form wild categories.
This notion has previously been briefly considered by Kraus~\cite[Definition 5]{kraus:21:internal}.

\renewcommand{\subcomp}{\catcomp}
\renewcommand{\Con}{\Ob\thecwf}
\renewcommand{\Sub}[2]{\Hom[\thecwf]{#1}{#2}}

\subsection{Typed term structures}

\begin{definition}[Typed term structures on wild categories]\label{def:typed-term-structure}
  Let \(\UniverseType\) be a universe and \(\C\) a wild category.
  A \emph{typed term structure} on \(\C\) (valued in \(\UniverseType\)) consists of the following data:
  \begin{itemize}
    \item\label{decl:ty-psh}
      A wild \(\UniverseCat\)-valued presheaf of \emphb{\(\C\)-types} over \(\C\), presented as a generalized algebraic theory by the components%
      \footnote{\label{footnote:cwf-ctx-quant}Implicitly quantifying over objects \(\Gamma, \Delta, \Epsilon \oftype \Ob\C\) as needed.}
      \begin{alignat*}{3}
        \Ty && \ & \oftype & \ \ & \Ob\C \to \UniverseType \\
        \blank \substT \blank &&& \oftype && \Ty*\Delta \to \Sub\Gamma\Delta \to \Ty*\Gamma
      \end{alignat*}
      and equations\cref{footnote:cwf-ctx-quant} expressing functoriality
      \begin{alignat*}{4}
        \substTid && \ & \oftype & \ \ & A \substT {\id[\Gamma]} = A & \ &
          \eqntextall A \oftype \Ty*\Gamma \\
        \substTcomp &&& \oftype && A \substT {\tau \subcomp \sigma} = A \substT \tau \substT \sigma &&
          \eqntextall A \oftype \Ty*\Epsilon,\ \sigma \oftype \Sub\Gamma\Delta,\ \tau \oftype \Sub\Delta\Epsilon.
      \end{alignat*}

    \item\label{decl:tm-psh}
      A wild \(\UniverseCat\)-valued presheaf of \emphb{\(\C\)-terms} over the (wild) category of elements of the \(\C\)-type presheaf, presented\cref{footnote:cwf-ctx-quant} by
      \begin{alignat*}{4}
        \Tm && \ & \oftype & \ \ & (\Gamma \oftype \Ob\C) \to \Ty*\Gamma \to \UniverseType \\
        \blank \substt \blank &&& \oftype &&
          \Tm[\Delta]*A \to (\sigma \oftype \Sub\Gamma\Delta) \to \Tm[\Gamma]*(A \substT \sigma) & \ &
          \eqntextall A \oftype \Ty*\Delta
      \end{alignat*}
      and
      \begin{alignat*}{4}
        \substtid && \ & \oftype & \ \ &
          a \substt {\id[\Gamma]} = a \transpover[\Tm[\Gamma]]{\pathinv\substTid} & \ &
          \eqntextall A \oftype \Ty*\Gamma,\ a \oftype \Tm[\Gamma]*A \\
        \substtcomp &&& \oftype &&
            a \substt {\tau \subcomp \sigma} = a \substt \tau \substt \sigma \transpover[\Tm[\Gamma]]{\pathinv\substTcomp} &&
            \eqntextall
            \begin{aligned}[t]
              A & \oftype \Ty*\Epsilon,\ a \oftype \Tm[\Epsilon]*A \\
              \sigma & \oftype \Sub\Gamma\Delta,\ \tau \oftype \Sub\Delta\Epsilon.
            \end{aligned}
      \end{alignat*}
  \end{itemize}
  The actions \(\blank\substT\blank\) and \(\blank\substt\blank\) of the type and term presheaves on morphisms are called \emph{substitution in types} and \emph{substitution in terms}, respectively.
\end{definition}

We will often denote a typed term structure on a wild category simply by the object parts of its component presheaves \((\Ty, \Tm)\).
We also frequently elide the first argument of \(\Tm\) and write, for example, \(\Tm*A\) instead of \(\Tm[\Gamma]*A\).

Note the following equivalences and equalities in typed term structures \((\Ty, \Tm)\) on wild categories \(\C\).

\begin{proposition}\label{prop:substtid-is-equiv}
  For every \(\Gamma \oftype \Ob\C\) and \(A \oftype \Ty*\Gamma\), the equation \(\substtid\) (\Cref{def:typed-term-structure}) implies that the function
  \begin{alignat*}{2}
    \blank \substt \id & \ \oftype \ \ & \Tm*A & \to \Tm*(A \substT \id) \\
    && a & \mapsto a \substt \id
  \end{alignat*}
  is equal to transport in \(\Tm[\Gamma]\) along \(\pathinv\substTid\), and is hence an equivalence.
\end{proposition}

\begin{definition}\label{def:eqsubsubstT-eqsubsubstt}
  Assume objects \(\Gamma, \Delta \oftype \Ob\C\), a \(\C\)-type \(A \oftype \Ty*\Delta\), and an equality \(e \oftype \sigma = \tau\) of morphisms \(\sigma, \tau \oftype \Sub\Gamma\Delta\).
  We write
  \[ \eqsubsubstT e \defeq \ap[(A \substT \blank)][e] \]
  for the induced equality \(A \substT \sigma = A \substT \tau\).
  By induction on \(e\), we also have an equality
  \[ \eqsubsubstt e \oftype {a \substt \sigma}\transpover[\Tm]{\eqsubsubstT e} = a \substt \tau \]
  for any \(\C\)-term \(a \oftype \Tm*A\).
\end{definition}

By \cite[Lemma 2.2.2]{hott-book}, \(\eqsubsubstT\blank\) respects trivial, composite and inverse equalities.
Furthermore,

\begin{proposition}\label{prop:transport-in-Tm-of-substT}
  For any \(\Gamma, \Delta \oftype \Ob\C\), \(A \oftype \Ty*\Delta\), \(a \oftype \Tm*(A \substT \sigma)\) and morphisms \(\sigma, \tau \oftype \Sub\Gamma\Delta\) such that \(e \oftype \sigma = \tau\),
  \[ a \transpover[\Tm*(A \substT \blank)] e = a \transpover[\Tm]{\eqsubsubstT e} \]
  by \cite[Lemma 2.3.10]{hott-book}.
\end{proposition}

\begin{definition}\label{def:eqtypesubstT-eqtermsubstt}
  Suppose that \(A, \primed{A} \oftype \Ty*\Delta\) are \(\C\)-types such that \(e \oftype A = \primed{A}\).
  For any \(\sigma \oftype \Sub\Gamma\Delta\), we write
  \[ \eqtypesubstT{e}{\sigma} \defeq \ap[(\blank \substT \sigma)][e] \]
  for the induced equality \(A \substT \sigma = \primed{A} \substT \sigma\).

  Similarly, if \(a, \primed{a} \oftype \Tm*A\) with \(e \oftype a = \primed{a}\), we write
  \[ \eqtermsubstt{e}{\sigma} \defeq \ap[(\blank \substt \sigma)][e] \]
  for the induced equality \(a \substt \sigma = \primed{a} \substt \sigma\).
\end{definition}

\begin{proposition}[Substitution in transported terms]\label{prop:substitution-in-transported-terms}
  If \(e \oftype A =_{\Ty*\Delta} \primed{A}\),
  then for any \(a \oftype \Tm*A\) and morphism \(\sigma \oftype \Sub\Gamma\Delta\),
  \[ (a \transpover[\Tm] e) \substt \sigma = {a \substt \sigma} \transpover[\Tm]{\eqtypesubstT{e}{\sigma}} \]
  by induction on \(e\).
\end{proposition}

\begin{proposition}[\(\substTcomp\) is a natural isomorphism]\label{lem:substTcomp-nat-iso}
  Suppose that \(\sigma, \primed\sigma \oftype \Sub\Gamma\Delta\) are morphisms such that \(e \oftype \sigma = \primed\sigma\).
  By induction on \(e\), we have that
  the square
  \[\begin{tikzcd}[paths]
    A \substT {\tau \subcomp \sigma}
      \rar["\substTcomp"]
      \dar["\eqsubsubstT{\tau \whisker e}" swap]
    & A \substT \tau \substT \sigma
      \dar["\eqsubsubstT e"]
    \\
    A \substT {\tau \subcomp \primed\sigma}
      \rar["\substTcomp" swap]
    & A \substT \tau \substT {\primed\sigma}
  \end{tikzcd}\]
  canonically commutes for all \(A \oftype \Ty*\Epsilon\) and \(\tau \oftype \Sub\Delta\Epsilon\), and also that
  \[\begin{tikzcd}[paths]
    A \substT {\sigma \subcomp \varrho}
      \rar["\substTcomp"]
      \dar["\eqsubsubstT{e \whisker \varrho}" swap]
    & A \substT \sigma \substT \varrho
      \dar["\eqtypesubstT{\eqsubsubstT e}{\varrho}"]
    \\
    A \substT {\primed{\sigma} \subcomp \varrho}
      \rar["\substTcomp" swap]
    & A \substT {\primed\sigma} \substT \varrho
  \end{tikzcd}\]
  canonically commutes for all \(A \oftype \Ty*\Delta\) and \(\varrho \oftype \Sub\Beta\Gamma\).
\end{proposition}

The following two definitions are analogous to the conditions for a pseudofunctor between weak (2,\,1)-categories (i.e.\ bicategories where all 2-cells are invertible).\footnote{See e.g.\ \cite[\S4.1]{johnson-yau:21:2-dim-cats}.}
In the case that \(\C\) is 2-coherent, they improve the wild presheaf \(\Ty\) of a typed term structure on \(\C\) to what might be called a \emph{wild weak (2,\,1)-presheaf}.

\begin{definition}[Type triangulators]\label{def:type-triangulators}
  A typed term structure on a wild category \(\C\) is said to have \emph{type triangulators} if for all morphisms \(\sigma \oftype \Sub\Gamma\Delta\) and \(\thecwf\)-types \(A \oftype \Ty*\Delta\) the following triangles commute:
  \[
    \begin{tikzcd}[paths, column sep=0pt]
      A \substT {\id \subcomp \sigma}
        \ar[rr, "\substTcomp"]
        \ar[dr, "\eqsubsubstT\lunitor"']
      && A \substT \id \substT \sigma
        \ar[dl, "\eqtypesubstT\substTid\sigma"]
      \\ & A \substT \sigma
    \end{tikzcd} \eqntext{and}
    \begin{tikzcd}[paths, column sep=0pt]
      A \substT {\sigma \subcomp \id}
        \ar[rr, "\substTcomp"]
        \ar[dr, "\eqsubsubstT\runitor"']
      && A \substT \sigma \substT \id
        \ar[dl, "\substTid"]
      \\ & A \substT \sigma
    \end{tikzcd}
  \]
\end{definition}

\begin{definition}[Type pentagonators]\label{def:type-pentagonators}
  A typed term structure on a wild category \(\C\) has \emph{type pentagonators} if for all morphisms
  \[ \Gamma \xrightarrow{\varrho} \Delta \xrightarrow{\sigma} \Epsilon \xrightarrow{\tau} \Zeta \]
  and \(\thecwf\)-types \(A \oftype \Ty*\Zeta\), the following pentagon commutes:
  \[\begin{tikzcd}[
    paths, 
    column sep=-4.5ex,
    nodes={
      row 1/.style={row sep=4ex},
      row 2/.style={row sep=5ex}}
    ]
    & A \substT{\tau \subcomp \sigma \subcomp \varrho}
      \ar[dl, "\eqsubsubstT{\pathinv\associator}"{swap, yshift=-1.75ex, xshift=-0.75ex},
        start anchor={[xshift=-1ex]}, end anchor={[xshift=1ex]}]
      \ar[dr, "\substTcomp"{yshift=-1.25ex, xshift=0.5ex}, ""{name=U, below},
        start anchor={[xshift=1ex]}, end anchor={[xshift=-1ex]},] &
    \\
    A \substT{(\tau \subcomp \sigma) \subcomp \varrho}
      \dar["\substTcomp"{swap, yshift=1ex, xshift=-0.5ex},
        start anchor={[xshift=3ex]}, end anchor={[xshift=4.25ex]}]
    &&
    A \substT \tau \substT{\sigma \subcomp \varrho}
      \dar["\substTcomp"{yshift=1.25ex, xshift=0.5ex},
        start anchor={[xshift=-2.75ex]}, end anchor={[xshift=-4ex]}]
    \\
    |[xshift=2.1ex]| A \substT{\tau \subcomp \sigma} \substT \varrho
      \ar[rr, "\eqtypesubstT{\substTcomp}{\varrho}" swap]
    &&
    |[xshift=-2.1ex]| A \substT \tau \substT \sigma \substT \varrho
  \end{tikzcd}\]
\end{definition}

\subsection{Context extension structures and wild cwfs}

\begin{definition}[Context extension structures]\label{def:context-extension-structure}
  Assume a typed term structure \((\Ty, \Tm)\) on a wild category \(\C\).
  A \emph{context extension structure} on \((\C, \Ty, \Tm)\) is given%
  \footnote{Again, implicitly generalizing over \(\Gamma, \Delta, \Epsilon \oftype \Ob\C\) as needed.}
  by the following components
  \begin{alignat*}{4}
    \blank \hspace{1pt} \ctxext \blank && \ & \oftype & \ \ & (\Gamma \oftype \Ob\C) \to \Ty*\Gamma \to \Ob\C \\
    \cwfproj &&& \oftype && (A \oftype \Ty*\Gamma) \to \Sub{\Gamma \ctxext A}\Gamma \\
    \cwfvar &&& \oftype && (A \oftype \Ty*\Gamma) \to \Tm[\Gamma \ctxext A]*(A \substT {\cwfproj_A}) \\
    \blank \subext \blank &&& \oftype &&
      (\sigma \oftype\Sub\Gamma\Delta) \to \Tm[\Gamma]*(A \substT \sigma) \to \Sub\Gamma{\Delta \ctxext A} & \ &
      \eqntextall A \oftype \Ty*\Delta
  \end{alignat*}
  and equations (note \Cref{def:eqsubsubstT-eqsubsubstt})
  \begin{alignat*}{4}
    \cwfprojbeta && \ & \oftype & \ \ & \cwfproj_A \subcomp (\sigma \subext a) = \sigma & \ &
      \eqntextspace\text{and} \\
    \cwfvarbeta &&& \oftype && \cwfvar_A \substt {\sigma \subext a} = a \transpover[\Tm]{\pathinv{\eqsubsubstT\cwfprojbeta} \pathcomp \substTcomp} &&
      \eqntextall
      \begin{aligned}[t]
        \sigma & \oftype \Sub\Gamma\Delta, \\
        A & \oftype \Ty*\Delta,\ a \oftype \Tm[\Gamma]*(A \substT \sigma)
      \end{aligned} \\
    \subexteta &&& \oftype && (\cwfproj_A \subext \cwfvar_A) = \id[\Gamma \ctxext A] &&
      \eqntextall A \oftype \Ty*\Gamma \\
    \subextcomp &&& \oftype &&
      (\tau \subext a) \subcomp \sigma
      = ({\tau \subcomp \sigma} \subext
          {{a \substt \sigma} \transpover[\Tm]{\pathinv\substTcomp}}) &&
      \eqntextall
      \begin{aligned}[t]
        \sigma & \oftype \Sub\Gamma\Delta,\ \tau \oftype \Sub\Delta\Epsilon, \\
        A & \oftype \Ty*\Epsilon,\ a \oftype \Tm[\Delta]*(A \substT \tau).
      \end{aligned}
  \end{alignat*}

  We call \(\cwfproj\) the \emph{display map}, and \(\cwfvar\) the \emph{generic term} of the context extension structure.
\end{definition}

We sometimes elide the argument \(A \oftype \Ty*\Gamma\) to the display map \(\cwfproj\) and the generic term \(\cwfvar\) of a context extension structure.
When we need to be concise, we denote the display map \(\Gamma\ctxext A \xrightarrow{\cwfproj_A} \Gamma\) by \(\Gamma\ctxext A \twoheadrightarrow \Gamma\).

\begin{definition}[Cwf structures on wild categories]
  If \(\C\) is a wild category, a \emph{cwf structure} on \(\C\) consists of:
  \begin{itemize}
    \item a terminal object \(\emptycon \oftype \Ob\C\),
    \item a typed term structure \((\Ty, \Tm)\) on \(\C\), and
    \item a context extension structure on \((\C, \Ty, \Tm)\).
  \end{itemize}
  These model the structural rules of a Martin-L\"{o}f type theory over \(\C\).
\end{definition}

\begin{definition}[Wild categories with families]
  A \emph{wild category with families (wild cwf)} is a wild category \(\C\) together with a cwf structure on \(\C\).
  In this case, we call \(\C\) the \emph{category of contexts} of the wild cwf, its objects \emph{contexts}, and its morphisms \emph{substitutions}.
\end{definition}

We usually denote a wild cwf by its category of contexts.


Of course, every 1-cwf is a wild cwf.

\begin{example}[Universe cwfs]\label{eg:universe-cwfs}
  If a universe wild category \(\UniverseCat\) has \(\Sigmatype\)-types that satisfy the \(\laweta\)-rule, then it supports a canonical wild cwf structure given as follows.
  \begin{itemize}
    \item
      The terminal context \(\emptycon\) is the unit type \(\UnitType \oftype \UniverseType\).

    \item
      The typed term structure is as follows.
      \(\UniverseCwf\)-types in context \(\Gamma\) are \(\Gamma\)-indexed type families
      \begin{gather*}
        \Ty \oftype \UniverseType \to \UniverseSuc\UniverseType \\
        \Ty*\Delta \defeq \Delta \to \UniverseType,
      \end{gather*}
      while \(\UniverseCwf\)-terms \(a \oftype \Tm[\Delta]*A\) are sections of \(A \oftype \Ty*\Delta\)
      \[ \Tm[\Delta]*A \defeq \Pitype[\Delta][A]. \]

      Substitution of \(\sigma \oftype \Hom[\UniverseCat]\Gamma\Delta\) in \(\UniverseCwf\)-types \(A \oftype \Ty*\Delta\) and \(\UniverseCwf\)-terms \(a\) is given by precomposition
      \begin{gather*}
        A \substT \sigma \defeq A \comp \sigma, \\
        a \substt \sigma \defeq a \comp \sigma.
      \end{gather*}
      This action is definitionally functorial---that is, \(\substTid\), \(\substTcomp\), \(\substtid\) and \(\substtcomp\) are all families of trivial identities.

    \item
      The context extension structure is given by dependent pairing.
      The extended context \(\Delta \ctxext A\) is \(\Sigmatype[\Delta][A]\), and \(\cwfproj\) and \(\cwfvar\) are the functions \(\fst\) and \(\snd\) respectively.
      For \(\sigma \oftype \Hom[\UniverseCat]\Gamma\Delta\) and \(t \oftype \Tm[\Gamma]*(A \comp \sigma)\), the extended substitution \((\sigma \subext t) \oftype \Hom[\UniverseCat]\Gamma{\Sigmatype[\Delta][A]}\) is given by
      \[ (\sigma \subext t)(\gamma) \defeq (\sigma(\gamma), t(\gamma)). \]
      Again, the equations for context extension structures hold definitionally.
      In particular, the \(\laweta\)-rule for \(\Sigmatype\)-types is used for \(\subexteta\).
  \end{itemize}

  We refer to the resulting wild cwf as the \emph{universe cwf}.
\end{example}

Variations of this canonical universe cwf structure appear throughout the literature as the ``standard model''.

\begin{example}[Subuniverse cwfs and the 1-cwf of sets]\label{eg:subuniverse-cwfs}
  The construction of the typed term and context extension structures of \Cref{eg:universe-cwfs} works equally well for any \emph{sub}universe wild category (\Cref{eg:wild-categories}) that has a terminal object and is closed under \(\Sigmatype\)-types with \(\laweta\).
  In particular, the ``1-cwf'' of sets \(\SetCwf_\UniverseType\) is a subuniverse cwf of \(\UniverseCwf\).
\end{example}

\begin{definition}[Univalent wild cwfs]\label{def:univalent-wild-cwf}
  A wild cwf \(\thecwf\) is called \emph{univalent} if its category of contexts is univalent.
\end{definition}

\begin{example}[Univalent wild cwfs]
  Any subuniverse of a univalent universe \(\UniverseType\) yields a univalent cwf.
  In particular, \(\SetCwf_\UniverseType\) and \(\UniverseCwf\) are univalent cwfs if \(\UniverseType\) is a univalent universe.
\end{example}

\subsection{Structural properties of wild cwfs}

From now on we assume that \(\thecwf\) is a wild cwf.

\begin{lemma}[Substitutions into extended contexts are pairs]\label{lem:substitutions-into-extended-context}
  Let \(\Gamma, \Delta \oftype \Con\) be contexts, and \(A \oftype \Ty*\Delta\) a \(\thecwf\)-type.
  There is an equivalence
  \[\begin{tikzcd}[column sep=3em]
    \Sub\Gamma{\Delta\ctxext A}
      \rar[phantom,"\simeq"]
      \rar["{\sigma\ \mapsto\ (\cwfproj_A\,\subcomp\,\sigma\,,\ {\cwfvar_A\substt \sigma}\transpover[\Tm]{\pathinv\substTcomp})}"{yshift=0.5ex},bend left=20,start anchor=north east,end anchor=north west]
    & \Sigmatype*[(\sigma \oftype \Sub\Gamma\Delta)]*[\Tm*(A \substT \sigma)],
      \lar["{(\sigma \subext a)\ \mapsfrom\ (\sigma,\,a)}"{yshift=-1.25ex},bend left=20,start anchor=south west,end anchor=south east]
  \end{tikzcd}\]
  where the reverse function sends a pair \((\sigma, a)\) to the extended substitution \((\sigma \subext a)\) given by the context extension structure (\Cref{def:context-extension-structure}).
\end{lemma}
\begin{proof}
  For one composition, it's enough to show that for all \(\sigma \oftype \Sub\Gamma{\Delta\ctxext A}\),
  \[ (\cwfproj_A \subcomp \sigma \subext {{\cwfvar_A \substt \sigma} \transpover {\pathinv\substTcomp}})
    \eqnequal{\pathinv\subextcomp}
    (\cwfproj_A \subext \cwfvar_A) \subcomp \sigma
    \eqnequal{\subexteta \whisker \sigma}
    \id \subcomp \sigma
    \eqnequal{\lambda}
    \sigma.
  \]

  For the other, we show the equality of pairs
  \[ (\cwfproj_A \subcomp (\sigma \subext a), {\cwfvar_A \substt {\sigma \subext a}} \transpover {\pathinv\substTcomp}) = (\sigma, a). \]
  Equality of the first components is given by \(\cwfprojbeta\) (\Cref{def:context-extension-structure}), and for the second components we have that
  \[
  \begin{alignedat}[b]{2}
    & {\cwfvar_A \substt {\sigma \subext a}} \transpover[\Tm]{\pathinv\substTcomp} \transpover[\Tm*(A\substT\blank)]{\cwfprojbeta} \\
    \eqnequal{} & {\cwfvar_A \substt {\sigma \subext a}} \transpover[\Tm]{\pathinv\substTcomp \pathcomp \eqsubsubstT\cwfprojbeta}
      & \text{(by \Cref{prop:transport-in-Tm-of-substT})} \\
    \eqnequal{} & a \hspace{7ex}
      & \text{(by \(\cwfprojbeta\) and properties of transport)}
  \end{alignedat}. \qedhere
  \]
\end{proof}

\begin{corollary}[Elimination principle for \(\Sub\Gamma{\Delta\ctxext{A}}\)]
  By \Cref{lem:substitutions-into-extended-context}, to construct a section of a family of types \(P\) over \(\Sub\Gamma{\Delta\ctxext{A}}\), it's enough to give an element of \(P((\sigma \subext a))\) for every \(\sigma \oftype \Sub\Gamma\Delta\) and \(a \oftype \Tm*(A \substT \sigma)\).
\end{corollary}

\begin{corollary}[Equality of substitutions into extended contexts]\label{cor:equality-substitutions-into-extended-contexts}
  If \(\sigma\) and \(\tau\) are substitutions from \(\Gamma\) to \(\Delta \ctxext A\), then by \Cref{lem:substitutions-into-extended-context}, the fact that equivalences induce equivalent identity types~\cite[Theorem 2.11.1]{hott-book}, and \Cref{prop:transport-in-Tm-of-substT}, an equality \(\sigma = \tau\) is equivalent to a pair of equalities
  \[e \oftype \cwfproj \subcomp \sigma = \cwfproj \subcomp \tau \eqntext{and} \cwfvar \substt \sigma \transpover {\pathinv\substTcomp \pathcomp \eqsubsubstT{e} \pathcomp \substTcomp} = \cwfvar \substt \tau.\]

  An alternative but equivalent
  formulation is the following---for substitutions of the form \((\sigma, a), (\tau, b) \oftype \Sub\Gamma{\Delta \ctxext A}\),
  \[ \big((\sigma, a) =_{\Sub\Gamma{\Delta \ctxext A}} (\tau, b)\big) \ \equivalent \ \big(\Sigmatype*[(e \oftype \sigma = \tau)]*[a \transpover[\Tm]{\eqsubsubstT e} = b]\big). \]
  To see this, write \(\varphi\) for the forward equivalence of \Cref{lem:substitutions-into-extended-context}.
  In the proof of \Cref{lem:substitutions-into-extended-context} we showed that the equalities of pairs \(\varphi((\sigma, a)) = (\sigma, a)\) and \(\varphi((\tau, b)) = (\tau, b)\) hold.
  Then \(\varphi((\sigma, a)) = \varphi((\tau, b))\) is equivalent to the equality type of pairs \((\sigma, a) = (\tau, b)\), which by \Cref{prop:transport-in-Tm-of-substT} is equivalent to the \(\Sigma\)-type as claimed.
\end{corollary}

\begin{lemma}[Terms are sections of display maps]\label{lem:terms-sections}
  For all contexts \(\Gamma \oftype \Con\) and \(\thecwf\)-types \(A \oftype \Ty*\Gamma\), there is an equivalence
  \[ \Tm*A \simeq \Sect(\cwfproj_A) \]
  whose forward map sends the \(\thecwf\)-term \(a\) to the section \((\id \subext a \substt \id)\) of \(\cwfproj_A\), witnessed by \(\cwfprojbeta\).
\end{lemma}
\begin{proof}
  We have that%
  \newcommand{\adjust}{\hspace{-32ex}}%
  \begin{alignat*}{3}
    & \ \ && \Sect(\cwfproj_A) && \\
    \equiv &&& \Sigmatype*[(\sigma \oftype \Sub\Gamma{\Gamma\ctxext A})]*[\cwfproj_A \subcomp \sigma = \id] \\
    \simeq &&& \Sigmatype*[(u \oftype \Sigmatype[(\sigma \oftype \Sub\Gamma\Gamma)]*[\Tm*(A \substT \sigma)])]*[\cwfproj_A \subcomp (\fst*u \subext \snd*u) = \id]
      & \adjust\text{(by \Cref{lem:substitutions-into-extended-context})} \\
    \simeq &&& \Sigmatype*[(u \oftype \Sigmatype[(\sigma \oftype \Sub\Gamma\Gamma)]*[\Tm*(A \substT \sigma)])]*[\fst*u = \id]
      & \adjust\text{(by \(\cwfprojbeta\)~(\Cref{def:context-extension-structure}))} \\
    \simeq &&& \Sigmatype*[(u \oftype \Sigmatype[(\sigma \oftype \Sub\Gamma\Gamma)]*[\sigma = \id])]*[\Tm*(A \substT{\fst*u})]
      & \adjust\text{(assoc. of \(\Sigmatype\) and comm. of \(\times\))} \\
    \simeq &&& \Tm*(A \substT \id)
      & \adjust\text{(contractibility of singletons)} \\
    \simeq &&& \Tm*A
      & \adjust\text{(by the inverse of the equivalence \(\blank \substt \id\)~(\Cref{prop:substtid-is-equiv})).}
  \end{alignat*}
  Tracing the composition of this chain of equivalences, we compute that its inverse is equal to the map
  \begin{align*}
    \Tm*A & \to \Sect(\cwfproj_A) \\
    a & \mapsto ((\id \subext a \substt \id) , \cwfprojbeta). \qedhere
  \end{align*}
\end{proof}

Analogues of \Cref{lem:substitutions-into-extended-context,lem:terms-sections} were already observed for set-level 1-cwfs by Dybjer~\cite{dybjer:96:internal-tt} and by Hofmann~\cite{hofmann:97:syntax-semantics} in a traditional 1-categorical setting.
In that setting, these properties essentially follow from the fact that context extension structures on \((\C, \Ty, \Tm)\) are choices of representing objects for particular presheaves on slices of \(\C\).


It may seem slightly surprising that the fully coherent homotopical versions of the same properties hold for arbitrary, even noncoherent, wild cwfs.
On the other hand, given that a context extension structure for 1-cwfs essentially spells out the universal property of representability of a certain presheaf (and, relatedly, that the equivalent natural models~\cite{awodey:18:natural-models} have a relatively simple axiomatization in terms of representability of pullbacks of presheaves), it is perhaps to be expected that certain consequences would carry over immediately to higher generalizations even without imposing additional coherence conditions.

%% file: 2-coh-ctx-ext.tex
\section[2-coherence for context extension]{2-Coherence for Context Extension}\label{sec:2-coh-ctx-ext}

We now construct and analyze in more detail a characterization of the equality of substitutions into extended contexts.

First, consider the case where \(\sigma\) and \(\tau\) are substitutions from \(\Gamma\) to \emph{arbitrary} contexts \(\Delta\), with \(A \oftype \Ty*\Delta\), \(a \oftype \Tm*(A \substT \sigma)\) and \(b \oftype \Tm*(A \substT \tau)\).
From the equivalence
\[ \big(\Sigmatype*[(\sigma \oftype \Sub\Gamma\Delta)]*[\Tm*(A \substT \sigma)]\big) \equivto[(\blank \subext \blank)] \Sub\Gamma{\Delta\ctxext{A}} \]
of \Cref{lem:substitutions-into-extended-context}, we obtain
\[ \big((\sigma, a) =_{\Sigmatype[(\Sub\Gamma{\Delta \ctxext A})][(\Tm*(A \substT \blank))]} (\tau, b)\big)
\xrightarrow[{\ap[(\blank \subext \blank)]}]{\sim}
\big((\sigma \subext a) =_{\Sub\Gamma{\Delta \ctxext A}} (\tau \subext b)\big).
\]
Precomposing this with
\begin{align*}
\big(\Sigmatype*[(e \oftype \sigma = \tau)]*[a \transpover[\Tm]{\eqsubsubstT{e}} = b]\big)
  & \equivto[\id \,\times\, \varphi] \Sigmatype*[(e \oftype \sigma = \tau)]*[a \transpover[\Tm*(A \substT \blank)]{e} = b] \\
  & \equivto[\Sigmatypeeq] \big((\sigma, a) =_{\Sigmatype[(\Sub\Gamma{\Delta \ctxext A})][(\Tm*(A \substT \blank))]} (\tau, b)\big),
\end{align*}
where
\[ \varphi \oftype \Pitype*[(e \oftype \sigma = \tau)][\big(a \transpover[\Tm]{\eqsubsubstT e} = b \equivto a \transpover[\Tm*(A \substT \blank)]{e} = b\big)] \]
is the family of equivalences induced by \Cref{prop:transport-in-Tm-of-substT} and \(\Sigmatypeeq\) is the standard characterization of the equality of \(\Sigmatype\)-types,
we get an equivalence
\begin{gather*}
\subeqaux \oftype \big(\Sigmatype*[(e \oftype \sigma = \tau)]*[a \transpover[\Tm]{\eqsubsubstT{e}} = b]\big) \equivto (\sigma, a) =_{\Sub\Gamma{\Delta \ctxext A}} (\tau, b) \\
\subeqaux \, (e, \primed{e}) \defeq \ap[(\blank \subext \blank)][(\Sigmatypeeq \, (e, \varphi_{e}\primed{e}))].
\end{gather*}

\begin{lemma}[\(\cwfprojbeta\) is a natural isomorphism]\label{lem:cwfprojbeta-nat-iso}
  By definition, \(\subeqaux\,(\refl, \refl) \equiv \refl\).
  Hence for all substitutions \(\sigma, \tau \oftype \Sub\Gamma{\Delta \ctxext A}\),
  \(\thecwf\)-terms \(a \oftype \Tm*(A \substT \sigma)\) and \(b \oftype \Tm*(A \substT \tau)\),
  and equalities \(e \oftype \sigma = \tau\) and \(\primed{e} \oftype a \transpover{\eqsubsubstT e} = b\),
  we have that the square
  \[\begin{tikzcd}[paths]
      \cwfproj_A \subcomp (\sigma \subext a)
      \rar["\cwfprojbeta"]
      \dar["{\cwfproj_A \, \whisker \, \subeqaux\,(e, \primed{e})}" swap]
      & \sigma \dar["e"]
      \\
      \cwfproj_A \subcomp (\tau \subext b)
      \rar["\cwfprojbeta" swap]
      & \tau
  \end{tikzcd}\]
  canonically commutes by induction on \(e\) and \(\primed{e}\).
  Equivalently,
  \[ \cwfproj_A \whisker \subeqaux\,(e, \primed{e}) = \cwfprojbeta \pathcomp e \pathcomp \pathinv\cwfprojbeta. \]
\end{lemma}

\begin{definition}[\(\laweta\)-equality of substitutions]\label{notation:subetaequality}
  For all \(A \oftype \Ty*\Delta\) and \(\sigma \oftype \Sub\Gamma{\Delta\ctxext{A}}\), denote by
  \begin{gather*}
    \subetaequality{\sigma} \oftype (\cwfproj_A \subcomp \sigma \subext {{\cwfvar_A \substt \sigma} \transpover {\pathinv\substTcomp}}) = \sigma \\
    \subetaequality{\sigma} \defeq \pathinv\subextcomp \pathcomp (\subexteta \whisker \sigma) \pathcomp \lunitor
  \end{gather*}
  the equality in the first part of the proof of \Cref{lem:substitutions-into-extended-context}.
  This is an \(\laweta\)-rule for substitutions into extended contexts.
\end{definition}

\begin{definition}[Equality of substitutions into extended contexts, revisited]\label{def:subeq}
  Suppose that \(\sigma, \tau \oftype \Sub\Gamma{\Delta \ctxext A}\) are substitutions into an extended context.
  We define an equivalence
  \[ \subeq_{\sigma, \tau} \oftype
    \big(\Sigmatype*[(e \oftype \cwfproj \subcomp \sigma = \cwfproj \subcomp \tau)]*[
    \cwfvar \substt \sigma \transpover {\pathinv\substTcomp \pathcomp \eqsubsubstT{e} \pathcomp \substTcomp}
        = \cwfvar \substt \tau]\big)
    \equivto \sigma = \tau
  \]
  as follows: if
  \[ e \oftype \cwfproj_A \subcomp \sigma = \cwfproj_A \subcomp \tau \]
  and
  \[ \primed{e} \oftype
      \cwfvar \substt \sigma \transpover {\pathinv\substTcomp \pathcomp \eqsubsubstT{e} \pathcomp \substTcomp}
      = \cwfvar \substt \tau,
  \]
  then take \(\subeq_{\sigma, \tau}\,(e, \primed{e})\) to be the composite
  \[ \sigma
    \eqnequal{\pathinv{\subetaequality\sigma}}
    ({\cwfproj \subcomp \sigma} \subext {\cwfvar \substt \sigma}\transpover{\pathinv\substTcomp})
    \eqnequal{\subeqaux\,(e, \pprimed{e})}
    ({\cwfproj \subcomp \tau} \subext {\cwfvar \substt \tau}\transpover{\pathinv\substTcomp})
    \eqnequal{\subetaequality\tau}
    \tau,
  \]
  where \(\pprimed{e} \oftype {\cwfvar \substt \sigma \transpover{\pathinv\substTcomp} \transpover{\eqsubsubstT e}} = {\cwfvar \substt \tau \transpover{\pathinv\substTcomp}}\)
  is canonically constructed from \(\primed{e}\).
  This definition yields an equivalence, being essentially the composition of \(\subeqaux\) with the equivalence given by path-composing with \(\pathinv{\subetaequality\sigma}\) and \(\subetaequality\tau\).
\end{definition}

Now, it is natural to ask if a \(\lawbeta\)-rule holds for the first argument of \(\subeq_{\sigma, \tau}\), i.e.\ if, for all \(\sigma\) and \(\tau\), the composition
\[
  \big(\Sigmatype*[(\cwfproj \subcomp \sigma = \cwfproj \subcomp \tau)]*[
    \cwfvar \substt \sigma \transpover {\pathinv\substTcomp \pathcomp \eqsubsubstT{\blank} \pathcomp \substTcomp}
      = \cwfvar \substt \tau]\big)
  \xrightarrow{\subeq_{\sigma, \tau}} \sigma = \tau
  \xrightarrow{\cwfproj \,\whisker\, \blank} \cwfproj \subcomp \sigma = \cwfproj \subcomp \tau
\]
is equal to the first projection.
By \Cref{lem:cwfprojbeta-nat-iso}, we calculate that
\[
  \cwfproj \whisker \subeq_{\sigma, \tau}\,(e, \primed{e})
  = (\cwfproj \whisker \pathinv{\subetaequality\sigma})
    \pathcomp \cwfprojbeta \pathcomp e \pathcomp \pathinv\cwfprojbeta
    \pathcomp (\cwfproj \whisker \subetaequality\tau)
\]
for all \(e\) and \(\primed{e}\).
The desire to have this expression be equal to \(e\) motivates the next definition, whence \Cref{prop:subeq-beta} immediately follows.

\begin{definition}[Coherators for \(\subetaequality{}\)]\label{def:subetaequality-coherators}
  We say that a wild cwf \(\thecwf\) \emph{has coherators for \(\subetaequality{}\)} if for all \(\Gamma, \Delta \oftype \Con\), \(A \oftype \Ty*\Delta\) and \(\sigma \oftype \Sub\Gamma{\Delta \ctxext A}\), we have that
  \[ \cwfproj_A \whisker \subetaequality{\sigma} = \cwfprojbeta \]
  as 2-cells of type \(\cwfproj_A \subcomp (\cwfproj_A \subcomp \sigma \subext \cwfvar_A \substt \sigma \transpover{\pathinv\substTcomp}) = \cwfproj_A \subcomp \sigma\).
\end{definition}

\begin{proposition}[\(\lawbeta\)-reduction for \(\subeq_{\sigma, \tau}\)]\label{prop:subeq-beta}%
  Suppose \(\sigma, \tau \oftype \Sub\Gamma{\Delta \ctxext A}\) are equal substitutions, witnessed by
  \(e \oftype \cwfproj_A \subcomp \sigma = \cwfproj_A \subcomp \tau\)
  and
  \( \primed{e} \oftype
    \cwfvar \substt \sigma \transpover {\pathinv\substTcomp \pathcomp \eqsubsubstT{e} \pathcomp \substTcomp}
    = \cwfvar \substt \tau \).
  If \(\thecwf\) has coherators for \(\subetaequality{}\), then
  \[ \cwfproj_A \whisker \subeq_{\sigma, \tau}\,(e, \primed{e}) = e. \]
\end{proposition}

In fact, a wild cwf \(\thecwf\) has coherators for \(\subetaequality{}\) if its category of contexts has triangle coherators (\Cref{def:triangle-coherators-unitors}), and it further satisfies the following coherence condition:

\begin{definition}[Coherators for context extension]\label{node:coherence-context-extension}
  A wild cwf \(\thecwf\) \emph{has coherators for context extension} if, for all \(\Gamma, \Delta \oftype \Ob\C\), \(A \oftype \Ty*\Delta\) and \(\sigma \oftype \Sub\Gamma{\Delta \ctxext A}\), the following diagrams of equalities commute:
  { \newcommand{\p}{\cwfproj_A}
    \newcommand{\q}{\cwfvar_A}
  \[\begin{tikzcd}[paths, column sep=0.5ex]
    & \p \subcomp (\p \subext \q)
      \ar[dl, "\p \whisker \subexteta"{swap, xshift=-2pt, yshift=-6pt}]
      \ar[dr, "\cwfprojbeta"{xshift=2pt, yshift=-6pt}]
    & \\
    \p \subcomp \id
      \ar[rr, "\runitor" swap]
    && \p
  \end{tikzcd}\]
  and
  \[\begin{tikzcd}[paths, column sep=0.5ex]
    \p \subcomp (\p \subext \q) \subcomp \sigma
      \rar["\pathinv\associator"]
      \dar["\p \whisker \subextcomp" swap]
    & (\p \subcomp (\p \subext \q)) \subcomp \sigma
      \dar["\cwfprojbeta \whisker \sigma"]
    \\
    \p \subcomp (\p \subcomp \sigma \subext \q \substt \sigma \transpover{\pathinv\substTcomp})
      \rar["\cwfprojbeta" swap]
    & \p \subcomp \sigma
  \end{tikzcd}\]
  }
\end{definition}

\begin{lemma}\label{lem:coh-for-ctxt-ext-implies-coh-for-subetaequality}
  If a wild cwf \(\thecwf\) has triangle coherators as well as coherators for context extension, then it has coherators for \(\subetaequality{}\).
  Thus, the conclusion of \Cref{prop:subeq-beta} also holds if \(\thecwf\) has coherators for context extension.
\end{lemma}
\begin{proof}%
  Having coherators for \(\subetaequality{}\) is equivalent to having the outer boundary of \Cref{diag:coh-for-ctxt-ext-implies-coh-for-subetaequality} commute for all \(\Gamma, \Delta \oftype \Ob\C\), \(A \oftype \Ty*\Delta\) and \(\sigma \oftype \Sub\Gamma{\Delta \ctxext A}\).
  We show that this holds by pasting together the regions shown in the interior of the diagram, where the topmost interior square is filled by associativity of whiskering (\Cref{prop:whiskering-properties}), the rightmost triangle by the triangle coherator, and the regions marked \(\circlearrowright\) using the coherators for context extension.
  \begin{diag}
    \caption{Coherators for \(\subetaequality{}\) from coherators for context extension.%
    }\label{diag:coh-for-ctxt-ext-implies-coh-for-subetaequality}%
    \newcommand{\p}{\cwfproj_A}%
    \newcommand{\q}{\cwfvar_A}%
    \newcommand{\oo}{\subcomp}%
    \[\begin{tikzcd}[paths, column sep=0pt]
      \p \oo (\p \subext \q) \oo \sigma
        \ar[rrr, "\cwfproj \whisker (\subexteta \whisker \sigma)"]
        \ar[dr, "\pathinv\associator"{yshift=-6pt, xshift=2pt}]
        \ar[dd, "\cwfproj \whisker \subextcomp" swap]
      & [-4em] & [3em] & \p \oo \id \oo \sigma
        \ar[dd, "\cwfproj \whisker \lunitor"]
      \\
      & (\p \oo (\p \subext \q)) \oo \sigma
        \rar["(\cwfproj \whisker \subexteta) \whisker \sigma"]
        \ar[drr, out=-75, in=175, "\cwfprojbeta \whisker \sigma"{yshift=11pt, xshift=-3em}, ""{name=U, below}]
        \markcomm[\circlearrowright]{dl}
      & (\p \oo \id) \oo \sigma
        \ar[ur, "\associator"{yshift=-6pt, xshift=-2pt}]
        \ar[dr, bend right=5, "\runitor \whisker \sigma"{swap, yshift=8pt}]
        \markcomm[\circlearrowright]{U}
      &
      \\[3em]
      \p \oo (\p \oo \sigma \subext \q \substt \sigma \transpover{\pathinv\substTcomp})
        \ar[rrr, "\cwfprojbeta" swap]
      &&& \p \oo \sigma
    \end{tikzcd}\]
  \end{diag}
\end{proof}

Presumably, the coherence conditions for context extension structures would arise out of the universal properties of sufficiently coherent representable presheaves \`{a} la a formulation via wild natural models.
In any case, we can now define:

\begin{definition}[Structurally 2-coherent wild cwfs]\label{def:2-coherent-cwfs}
  We say that a wild cwf \(\thecwf\) is \emph{(structurally) 2-coherent} if \(\thecwf\) has
  \begin{itemize}
    \item a 2-coherent wild category of contexts (\Cref{def:2-coherent-wild-cat}),
    \item type triangulators (\Cref{def:type-triangulators}) and type pentagonators (\Cref{def:type-pentagonators}), and
    \item coherators for context extension (\Cref{node:coherence-context-extension}).
  \end{itemize}
\end{definition}

\begin{examples}[2-coherent internal cwfs]
  Any set-level 1-cwf is immediately 2-coherent, and it is straightforward to check that the universe cwfs have type triangulators, type pentagonators and coherators for context extension.
\end{examples}

\begin{conjecture}[The container model]
  We also expect that the higher \emph{container model} of Altenkirch and Kaposi~\cite{altenkirch-kaposi:21:container-model} is 2-coherent, to be shown by forthcoming work of Damato and Altenkirch~\cite{damato-altenkirch:24:coherences-container-model}.
\end{conjecture}

%% file: ctx-comprehension.tex
\section[Context comprehension in 2-coherent wild cwfs]{Context Comprehension in 2-Coherent Wild Cwfs}\label{sec:ctx-comprehension}

Central to the 1-categorical semantics of dependent type theory is the recognition that types form a Grothendieck fibration over contexts, and that context extension maps this fibrational structure back into the category of contexts \(\C\), by forming a cartesian-morphism-preserving map into the arrow category \(\arrowcat\C\).
This is summed up in the notion of a \emph{comprehension category}~\cite{jacobs:93:comprehension-cats}, and it is widely known that 1-cwfs are equivalent to full split comprehension categories~\cite{blanco:91:categorical-approaches,aln:25:semantic-frameworks}.
In this section, we prove a higher version of this property for 2-coherent wild cwfs.

\subsection{The universal property of context extension}

\begin{lemma}[2-coherent substitution in types is weak pullback]\label{lem:substT-wk-pb}
  Suppose that \(\thecwf\) is a 2-coherent wild cwf, \(\sigma \oftype \Sub\Gamma\Delta\) is a substitution,
  and \(A \oftype \Ty*\Delta\) is a \(\thecwf\)-type.
  Then there is a substitution
  \[ \subliftT{\sigma}{A} \defeq (\sigma \subcomp \cwfproj_{A \substT \sigma} \subext \cwfvar_{A \substT \sigma} \transpover[\Tm]{\pathinv\substTcomp}) \]
  from \(\Gamma \ctxext A \substT \sigma\) to \(\Delta \ctxext A\), such that the square
  \[
    \substTpb{\sigma}{A} \eqndefeq
    \begin{gathered}
      \begin{tikzcd}
        \Gamma \ctxext A \substT \sigma
          \rar["\subliftT{\sigma}{A}"] \dar["\cwfproj" swap]
        & \Delta \ctxext A
          \dar["\cwfproj"]
        \\
        \Gamma
          \rar["\sigma" swap] \ar[ur, path, "\pathinv\cwfprojbeta" swap]
        & \Delta
      \end{tikzcd}
    \end{gathered}
  \]
  is a weak pullback in \(\C\).
  That is, for any \(\Beta \oftype \Con\) and commuting square \(\fS \defeq (\tau, \varrho, \gamma)\) with source \(\Beta\) as in
  \[
    \begin{tikzcd}
      \Beta \rar["\varrho"] \dar["\tau" swap]
      & \Delta \ctxext A \dar["\cwfproj"]
      \\
      \Gamma \rar["\sigma" swap] \ar[ur, path, "\gamma" swap]
      & \Delta
    \end{tikzcd},
  \]
  the fiber \(\fiber{(\substTpb{\sigma}{A} \sqcomp \blank)}(\fS)\) is pointed, i.e.\ there is a mediating substitution
  \[ \gapmap{\sigma}{A}{}_{, \fS} \oftype \Sub\Beta{\Gamma \ctxext {A \substT \sigma}} \]
  such that
  \[ \theta_{\sigma, A, \fS} \oftype \substTpb{\sigma}{A} \sqcomp \gapmap{\sigma}{A}{}_{, \fS} = \fS. \]
\end{lemma}
\begin{proof}
  \newcommand{\sigmaA}{\subliftT{\sigma}A}
  \newcommand{\pA}{\cwfproj_A}
  \newcommand{\pAsigma}{\cwfproj_{A \substT \sigma}}
  \newcommand{\qA}{\cwfvar_A}
  \newcommand{\qAsigma}{\cwfvar_{A \substT \sigma}}
  We claim that a mediating substitution is given by
  \[ \gapmap{\sigma}{A}{}_{, \fS} \defeq (\tau \subext \cwfvar_A \substt \varrho \transpover[\Tm]{\pathinv\substTcomp \pathcomp \pathinv{\eqsubsubstT\gamma} \pathcomp \substTcomp}), \]
  where \(\cwfvar_A \substt \varrho \oftype \Tm*{A \substT {\cwfproj_A} \substT \varrho}\) is transported in the family \(\Tm[\Beta]\) along
  \[ A \substT \cwfproj \substT \varrho
    \xlongequal{\pathinv\substTcomp} A \substT {\cwfproj \subcomp \varrho}
    \xlongequal{\pathinv{\eqsubsubstT\gamma}} A \substT {\sigma \subcomp \tau}
    \xlongequal{\substTcomp} A \substT \sigma \substT \tau.
  \]

  For brevity, denote \(\gapmap{\sigma}{A}{}_{, \fS}\) by \(\mu\).
  By \Cref{prop:commsq-fixed-cospan-source-equality}, constructing
  \(\theta \oftype \substTpb{\sigma}{A} \sqcomp \thegapmap = \fS\)
  is equivalent to constructing witnesses
  \[ \delta \oftype \cwfproj_{A \substT \sigma} \subcomp \thegapmap = \tau \]
  and
  \[ \epsilon \oftype \subliftT{\sigma}{A} \subcomp \thegapmap = \varrho \]
  such that
  \[ \pathinv\associator \pathcomp (\pathinv\cwfprojbeta \whisker \thegapmap) \pathcomp \associator = (\sigma \whisker \delta) \pathcomp \gamma \pathcomp \pathinv{(\cwfproj_A \whisker \epsilon)}. \]

  Let \(\delta \defeq \cwfprojbeta\).
  Using the equivalence \(\subeq_{\blank\,,\ \blank}\) (\Cref{def:subeq}), we define
  \[ \epsilon \defeq \subeq_{(\sigmaA \subcomp \thegapmap),\, \varrho}(\epsilon_0, \epsilon_1), \]
  where
  \[ \epsilon_0 \oftype \pA \subcomp \sigmaA \subcomp \thegapmap = \pA \subcomp \varrho \]
  and
  \[ \epsilon_1 \oftype \qA \substt{\sigmaA \subcomp \thegapmap} \transpover[\Tm]{\pathinv\substTcomp \pathcomp \eqsubsubstT{\epsilon_0} \pathcomp \substTcomp} = \qA \substt \varrho \]
  are the 2-cells%
  \footnote{Since \(\thecwf\)-terms correspond to display maps in \(\thecwf\) (\Cref{lem:terms-sections}), we are justified in also calling \(\epsilon_1\) a 2-cell.}
  constructed as follows.

  First, let \(\epsilon_0\) be the concatenation of equalities
  \[ \pA \subcomp \sigmaA \subcomp \thegapmap \xlongequal{\pathinv\associator \pathcomp (\cwfprojbeta \whisker \thegapmap) \pathcomp \associator} \sigma \subcomp \pAsigma \subcomp \thegapmap
    \xlongequal{\sigma \whisker \cwfprojbeta} \sigma \subcomp \tau
    \xlongequal\gamma \pA \subcomp \varrho. \]

  Now calculate that
  \newcommand{\adjust}{\hspace{2em}}
  \begin{alignat*}{3}
    & \ \ && \qA \substt {\sigmaA \subcomp \thegapmap} && \\
    = &&& \qA \substt \sigmaA \substt \thegapmap \transpover{\pathinv\substTcomp}
      & \adjust\text{(by \(\substtcomp\))} \\
    \equiv &&& \qA \substt {\sigma \subcomp \pAsigma \subext \qAsigma \transpover {\pathinv\substTcomp}} \substt \thegapmap \transpover {\pathinv\substTcomp} \\
    = &&& (\qAsigma \transpover {\pathinv\substTcomp \pathcomp {\pathinv{\eqsubsubstT\cwfprojbeta}} \pathcomp \substTcomp}) \substt \thegapmap \transpover {\pathinv\substTcomp}
      & \adjust\text{(by \(\cwfvarbeta\))} \\
    = &&& \qAsigma \substt\thegapmap \transpover {\eqtypesubstT{(\pathinv\substTcomp \pathcomp \pathinv{\eqsubsubstT{\cwfprojbeta}} \pathcomp \substTcomp)}{\thegapmap} \,\pathcomp\, \pathinv\substTcomp}
      & \adjust\text{(by \Cref{prop:substitution-in-transported-terms})} \\
    = &&& \qA \substt \varrho \transpover e & \adjust\text{(by \(\cwfvarbeta\))}
  \end{alignat*}
  where the transports are all in \(\Tm[\Beta]\), and \(e\) is the composition
  \[ e \defeq \pathinv\substTcomp \pathcomp \pathinv{\eqsubsubstT\gamma} \pathcomp \substTcomp
    \pathcomp \pathinv{\eqsubsubstT\cwfprojbeta} \pathcomp \substTcomp
    \pathcomp \eqtypesubstT{(\pathinv\substTcomp \pathcomp \pathinv{\eqsubsubstT\cwfprojbeta} \pathcomp \substTcomp)}{\thegapmap} \pathcomp \pathinv\substTcomp.
  \]
  So to construct \(\epsilon_1\), we may as well show that
  \[ \qA \substt \varrho \transpover {e \pathcomp \pathinv\substTcomp \pathcomp \eqsubsubstT{\epsilon_0} \pathcomp \substTcomp} = \qA \substt \varrho. \]
  We do this by showing that the equality
  \[ \widetilde{e} \defeq e \pathcomp \pathinv\substTcomp \pathcomp \eqsubsubstT{\epsilon_0} \pathcomp \substTcomp \]
  is in fact equal to the trivial identity.
  Some path algebra shows that \(\widetilde{e}\) is equal to the outer boundary of \Cref{diag:type-subst-wk-pb}.
  \input{diag-type-subst-wk-pb}
  This boundary commutes, since we can fill the interior of the diagram with the following commuting regions:
  \begin{itemize}
    \item (1), which commutes straightforwardly,
    \item (2) and (4), which commute by \Cref{lem:substTcomp-nat-iso}, and
    \item (3) and (5), which are filled by the type pentagonators.
  \end{itemize}
  This shows that \(\widetilde{e} = \refl\), which completes the proof \(\epsilon_1\) that
  \[ \qA \substt{\sigmaA \subcomp \thegapmap} \transpover[\Tm]{\pathinv\substTcomp \pathcomp \eqsubsubstT{\epsilon_0} \pathcomp \substTcomp} = \qA \substt \varrho \transpover{\widetilde{e}} = \qA \substt \varrho, \]
  and thus also the proof \(\epsilon \defeq \subeq_{(\sigmaA \subcomp \thegapmap),\, \varrho}(\epsilon_0, \epsilon_1)\) that
  \[ \sigmaA \subcomp \thegapmap = \varrho. \]

  Finally, what remains is to show that
  \[ \pathinv\associator \pathcomp (\pathinv\cwfprojbeta \whisker \thegapmap) \pathcomp \associator = (\sigma \whisker \delta) \pathcomp \gamma \pathcomp \pathinv{(\cwfproj_A \whisker \epsilon)}. \]
  But by \Cref{lem:coh-for-ctxt-ext-implies-coh-for-subetaequality} we have that \(\pathinv{(\pA \whisker \epsilon)} = \pathinv\epsilon_0\) on the right hand side, and the equality then follows by calculation.
\end{proof}

\begin{theorem}[2-coherent substitution in types is pullback]\label{thm:substT-pb}
  The weak pullbacks \(\substTpb{\sigma}{A}\) of \Cref{lem:substT-wk-pb} are pullbacks.
\end{theorem}
\begin{proof}
  \newcommand{\gp}{\gapmap\sigma{A}}
  By \Cref{lem:substT-wk-pb} we have that, for any \(\Beta \oftype \Con\), the map
  \begin{gather*}
    \gp \oftype \CommSq_{(\sigma, \cwfproj_A)}(\Beta) \to \Sub\Beta{\Gamma \ctxext {A \substT\sigma}} \\
    \gp((\tau, \varrho, \gamma)) \defeq (\tau \subext \cwfvar_A\substt{\varrho} \transpover[\Tm]{\pathinv\substTcomp \pathcomp \pathinv{\eqsubsubstT\gamma} \pathcomp \substTcomp})
  \end{gather*}
  is a section of the precomposition map \((\substTpb\sigma{A} \sqcomp_\Beta \blank)\).
  We show that it's a retraction of the same, and therefore that \((\substTpb\sigma{A} \sqcomp_{\blank} \blank)\) is a family of equivalences.

  That is, for \(m \oftype \Sub\Beta{\Gamma\ctxext{A\substT\sigma}}\), we want the equality of substitutions
  \[ \gp(\substTpb\sigma{A} \sqcomp m) = m. \]
  By \Cref{cor:equality-substitutions-into-extended-contexts} and a calculation similar to the one in the proof of \Cref{lem:substT-wk-pb} we have that
  \[ \gp(\substTpb\sigma{A} \sqcomp m)
    = (\cwfproj \subcomp m \subext \cwfvar_{A\substT\sigma} \substt m \transpover e), \]
  where
  \[ e \defeq
    \eqtypesubstT{(\conjsubstTcomp{\pathinv{\eqsubsubstT\cwfprojbeta}})}{m}
    \pathcomp \pathinv\substTcomp
    \pathcomp \conjsubstTcomp{\eqsubsubstT{\invassociator \pathcomp (\cwfprojbeta \whisker m) \pathcomp \associator}}.
  \]
  On the other hand,
  \( m = (\cwfproj \subcomp m \subext \cwfvar_{A\substT\sigma}\substt{m} \transpover{\pathinv\substTcomp}) \)
  by \Cref{lem:substitutions-into-extended-context}, and thus by \Cref{cor:equality-substitutions-into-extended-contexts} again it's enough to show that \(e = \pathinv\substTcomp\).

  By path algebra this amounts to showing the commutativity of a diagram of equalities that looks like the one formed by regions (3), (4) and (5) of \Cref{diag:type-subst-wk-pb}, but where we replace \(\thegapmap\) with \(m\).
  Commutativity of this diagram then follows as in the proof of \Cref{lem:substT-wk-pb}, i.e. by \Cref{lem:substTcomp-nat-iso} and type pentagonators.
\end{proof}

\subsection{Split 2-coherent wild cwfs}

\begin{definition}[Cleavings of wild cwfs]
  A \emph{cleaving} of a wild cwf \(\thecwf\) is an assignment
  \[ \thecleaving \oftype \Pitype*[(\Gamma, \Delta \oftype \Con)\,(\sigma \oftype \Sub\Gamma\Delta)\,(A \oftype \Ty*\Delta)][\Pullback_{(\sigma, \cwfproj_A)}(\Gamma \ctxext {A \substT \sigma})] \]
  of pullbacks
  \[ \cleaving\Gamma\Delta\sigma{A} \eqndefeq
    \begin{tikzcd}
      \Gamma \ctxext A \substT \sigma
        \dar
        \rar["\cleavinglift\sigma{A}"]
      & \Delta \ctxext A \dar[twoheadrightarrow]
      \\
      \Gamma
        \ar[ur, path, "\cleavingcomm\sigma{A}"']
        \rar["\sigma"']
      & \Delta
    \end{tikzcd}
  \]
  to each cospan in \(\thecwf\) of the form
  \(\begin{tikzcd}[cramped,sep=1.5em]\Gamma \rar["\sigma"] & \Delta & \lar["\cwfproj_A"{xshift=2pt},swap] \Delta \ctxext A\end{tikzcd}\).
  We call the component \(\cleavinglift\sigma{A}\) of the pullback \(\cleaving\Gamma\Delta\sigma{A}\) the \emph{chosen lift of \(\sigma\) at \(A\)}.
\end{definition}

The upshot of \Cref{thm:substT-pb} is then that every 2-coherent wild cwf has a cleaving
\[ \cleaving\Gamma\Delta\sigma{A} \defeq \substTpb\sigma{A}, \]
which we call the \emph{type substitution cleaving}.
In particular, \(\cleavinglift{\sigma}A \defeq \subliftT{\sigma}A\) is the chosen lift of a substitution \(\sigma \oftype \Sub\Gamma\Delta\) at \(A \oftype \Ty*\Delta\).
Observe, additionally, that the vertical legs of the pullbacks chosen by the type substitution cleaving are all display maps.
This allows us to make the following definition.

\begin{definition}[Split 2-coherent cwfs]\label{def:split-2-coh-wild-cwf} 
  A 2-coherent wild cwf \(\C\) is called \emph{(coherently) split} if the following equality types are contractible:
  \begin{enumerate}
    \item
      For every \(\Gamma \oftype \Con\) and \(A \oftype \Ty*\Gamma\), the type of equalities
      \[ (\Gamma \ctxext {A \substT \id},\ \substTpb\id{A}) = (\Gamma \ctxext A,\ \idpb{\cwfproj_A}) \]
      of pullbacks on \(\begin{tikzcd}[cramped,sep=1.75em]\Gamma \rar["\id"{xshift=-2pt}] & \Gamma & \lar["\cwfproj_A"{xshift=2pt},swap] \Gamma \ctxext A\end{tikzcd}\).
    \item
      For all substitutions
      \(\begin{tikzcd}[cramped,sep=1.75em]
        \Beta \rar["\tau"] & \Gamma \rar["\sigma"] & \Delta
      \end{tikzcd}\)
      and \(\thecwf\)-types \(A \oftype \Ty*\Delta\), the type of equalities
      \[ \big(\Beta \ctxext{A \substT{\sigma \subcomp \tau}},\ \substTpb{\sigma \subcomp \tau}{A}\big)
        = \big(\Beta \ctxext {A \substT \sigma \substT \tau},\ \squarehpaste{\substTpb\tau{A\substT\sigma}}{\substTpb{\sigma}{A}}\big) \]
      of the pullbacks
      \[ \substTpb{\sigma \subcomp \tau}{A}
        \equiv
        \begin{tikzcd}[column sep=2em]
          \Beta \ctxext A \substT {\sigma \subcomp \tau}
            \dar[twoheadrightarrow]
            \rar["\subliftT{\sigma \subcomp \tau}{A}"{xshift=-1pt,yshift=3pt}]
          & \Delta \ctxext A \dar[twoheadrightarrow]
          \\
          \Beta
            \ar[ur, path, "\substTpbtwocell"'{yshift=-1pt}]
            \rar["\sigma \subcomp \tau"']
          & \Delta
        \end{tikzcd}
      \]
      and
      \[ \squarehpaste{\substTpb\tau{A\substT\sigma}}{\substTpb{\sigma}{A}}
        \equiv
        \begin{gathered}\begin{tikzcd}
          \Beta \ctxext {A \substT \sigma \substT \tau}
            \dar[twoheadrightarrow]
            \rar["\subliftT\tau{A\substT\sigma}"{yshift=2pt}]
          & [-1em] \Gamma \ctxext A \substT \sigma
            \dar[twoheadrightarrow]
            \rar["\subliftT\sigma{A}"{yshift=2pt}]
          & [0.3em] \Delta \ctxext A \dar[twoheadrightarrow]
          \\
          \Beta
            \ar[ur, path, "\substTpbtwocell"', end anchor={[xshift=1ex,yshift=0.75ex]}]
            \rar["\tau"']
          & \Gamma
            \ar[ur, path, "\substTpbtwocell"']
            \rar["\sigma"']
          & \Delta
        \end{tikzcd}\end{gathered}
      \]
      on
      \(\begin{tikzcd}[cramped,sep=1.75em]\Beta \rar["\sigma \subcomp \tau"{xshift=-2pt}] & \Delta & \lar["\cwfproj_A"{xshift=2pt},swap] \Delta \ctxext A\end{tikzcd}\).
  \end{enumerate}
\end{definition}

A split 2-coherent wild cwf may be considered a coherent higher version of a splitting of a full comprehension category.
As is to be expected, set-level internal cwfs are split in the sense of \Cref{def:split-2-coh-wild-cwf}.
We now show that \emph{univalent} 2-coherent wild cwfs are also split.


\begin{proposition}\label{prop:idd-ap-ctxext}
  Let \(e \oftype A = \primed{A}\) be an equality of \(\thecwf\)-types \(A, \primed{A} \oftype \Ty*\Gamma\) in a wild cwf \(\thecwf\).
  Then
  \[ \idd(\ap[(\Gamma \ctxext \blank)][e]) =
    (\cwfproj_A \subext \cwfvar_A \transpover[\Tm]{\eqtypesubstT{e}{\cwfproj_A}}). \]
\end{proposition}
\begin{proof}
  By induction on \(e\) it's enough to show that \(\idd(\refl_{\Gamma \ctxext A}) = (\cwfproj_A \subext \cwfvar_A)\), which holds by \(\subexteta\) (\Cref{def:context-extension-structure}).
\end{proof}

\begin{lemma}\label{lem:substT-lift-id}
  In any 2-coherent wild cwf \(\thecwf\),
  \[ 
    \subliftT\id{A} = \idd(\ap[(\Gamma \ctxext \blank)][\substTid]) \]
  for any \(\Gamma \oftype \Con\) and \(A \oftype \Ty*\Gamma\).
\end{lemma}
\begin{proof}
  By \Cref{prop:idd-ap-ctxext} it's enough to show that \(\subliftT\id{A} = (\cwfproj_{A\substT\id} \subext \cwfvar_{A\substT\id} \transpover{\eqtypesubstT\substTid{\cwfproj_{A\substT\id}}})\), which holds by \Cref{cor:equality-substitutions-into-extended-contexts} and the left type triangulator.
\end{proof}

In fact, the equality of \Cref{lem:substT-lift-id} improves to an equality of pullbacks.
Recall from \Cref{prop:id-pullback} that the identity commuting square \(\idpb{f}\) is a pullback for any morphism \(f\) in a 2-coherent wild category \(\C\).

\begin{corollary}\label{cor:substT-cleaving-id}
  For any context \(\Gamma \oftype \Ob\C\) and \(\C\)-type \(A \oftype \Ty*\Gamma\) in a 2-coherent wild cwf \(\C\), the pullbacks
  \[ \substTpb{\id}{A} \equiv
    \begin{tikzcd}
      \Gamma \ctxext {A \substT \id}
        \rar["\subliftT\id{A}"]
        \dar[twoheadrightarrow]
      & \Gamma \ctxext A
        \dar[twoheadrightarrow]
      \\
      \Gamma
        \rar["\id"']
        \ar[ur,path,"\pathinv\cwfprojbeta"']
      & \Gamma
    \end{tikzcd}
    \eqntext{and}
    \idpb{\cwfproj_A} \equiv
    \begin{tikzcd}
      \Gamma \ctxext A
        \rar["\id"]
        \dar[twoheadrightarrow]
      & \Gamma \ctxext A
        \dar[twoheadrightarrow]
      \\
      \Gamma
        \rar["\id"']
        \ar[ur,path,"\lunitor\pathcomp\pathinv\runitor"']
      & \Gamma
    \end{tikzcd}
  \]
  are equal elements of \(\Pullback((\id, \cwfproj_A))\).
\end{corollary}
\begin{proof}
  We have that
  \[ \ap[(\Gamma \ctxext \blank)][\substTid] \oftype \Gamma \ctxext {A \substT \id} = \Gamma \ctxext A, \]
  and by \Cref{lem:commsq-fixed-cospan-equality} it's enough to show that
  \( \substTpb{\id}{A} = \idpb{\cwfproj_A} \sqcomp \idd(\ap[(\Gamma \ctxext \blank)][\substTid]) \).
  But by \Cref{lem:substT-lift-id} we may as well show that
  \[ \substTpb{\id}{A} = \idpb{\cwfproj_A} \sqcomp \subliftT\id{A}, \]
  which holds by \Cref{rem:comm-sq-id-leg-idpb-sqcomp}.
\end{proof}

\begin{lemma}\label{lem:substT-cleaving-equality}
  Suppose \(\C\) is a 2-coherent wild cwf.
  For all substitutions
  \(\begin{tikzcd}[cramped,sep=1.75em]
    \Beta \rar["\tau"] & \Gamma \rar["\sigma"] & \Delta
  \end{tikzcd}\)
  and \(\C\)-types \(A \oftype \Ty*\Delta\),
  \[
    (\Beta \ctxext A \substT{\sigma \subcomp \tau},\ \substTpb{\sigma \subcomp \tau}{A})
  = (\Beta \ctxext A \substT \sigma \substT \tau,\ \squarehpaste{\substTpb\tau{A \substT \sigma}}{\substTpb\sigma{A}}).
  \]
\end{lemma}
\begin{proof}
  \newcommand{\invsubstTcomp}{\pathinv\substTcomp}
  \newcommand{\PstA}{\substTpb{\sigma \subcomp \tau}{A}}
  \newcommand{\PtAs}{\substTpb{\tau}{A \substT \sigma}}
  \newcommand{\PsA}{\substTpb{\sigma}{A}}
  \newcommand{\pqsub}{(\cwfproj \subext \cwfvar \transpover{\eqtypesubstT{\substTcomp}{\cwfproj}})}
  \newcommand{\slift}{\subliftT\sigma{A}}
  \newcommand{\tlift}{\subliftT\tau{A \substT \sigma}}
  \newcommand{\slifttlift}{(\slift \subcomp \tlift)}
  \newcommand{\stlift}{\subliftT{(\sigma\subcomp\tau)}{A}}
  By \Cref{lem:commsq-fixed-cospan-equality} it's enough to give
  \[ e \oftype \Beta \ctxext A \substT{\sigma \subcomp \tau} = \Beta \ctxext A \substT \sigma \substT \tau \]
  such that
  \[ \PstA = (\squarehpaste\PtAs\PsA) \sqcomp \idd(e). \]
  Take
  \[ e \defeq \ap[(\Beta \ctxext \blank)][\substTcomp], \]
  then by \Cref{prop:idd-ap-ctxext} we may as well show that
  \[ \PstA = (\squarehpaste\PtAs\PsA) \sqcomp \pqsub, \]
  or, equivalently, give three equalities
  \begin{gather*}
    \delta \oftype
      \cwfproj_{A \substT {\sigma \subcomp \tau}} = \cwfproj_{A \substT \sigma \substT \tau} \subcomp \pqsub, \\
    \epsilon \oftype
      \subliftT{(\sigma \subcomp \tau)}{A} = \slifttlift \subcomp \pqsub
  \end{gather*}
  and
  \[ \eta \oftype
    \cleavingcomm{\sigma \subcomp \tau}{A}
    = \big((\sigma \subcomp \tau) \whisker \delta\big)
      \pathcomp \invassociator
      \pathcomp \big((\squarehpaste{\cleavingcomm{\tau}{A \substT \sigma}}{\cleavingcomm{\sigma}{A}}) \whisker \pqsub\big)
      \pathcomp \associator
      \pathcomp (\cwfproj_A \whisker \pathinv\epsilon),
  \]
  where \(\cleavingcomm{\sigma \subcomp \tau}{A}\), \(\cleavingcomm{\tau}{A \substT \sigma}\) and \(\cleavingcomm{\sigma}{A}\) are the 2-cells \(\substTpbtwocell\) of the corresponding pullbacks given by the type substitution cleaving.

  Take \(\delta \defeq \pathinv\cwfprojbeta\).
  We will define \(\epsilon \defeq \subeq(\epsilon_0, \epsilon_1)\), where \(\subeq\) is the equivalence defined at \Cref{def:subeq}, and where we seek equalities
  \[ \epsilon_0 \oftype \cwfproj \subcomp \stlift = \cwfproj \subcomp \slifttlift \subcomp \pqsub \]
  and
  \[
    \epsilon_1 \oftype \cwfvar \substt{\stlift} \transpover{\pathinv\substTcomp \pathcomp \eqsubsubstT{\epsilon_0} \pathcomp \substTcomp}
    = \cwfvar \substt {\slifttlift \subcomp \pqsub}.
  \]

  Now, from \Cref{prop:subeq-beta} we have that
  \[ \cwfproj_A \whisker \pathinv\epsilon = \pathinv{(\cwfproj_A \whisker \epsilon)} = \pathinv\epsilon_0, \]
  and by rearranging the type of \(\eta\) we may take
  \[ \epsilon_0 \defeq
    \pathinv{\cleavingcomm{\sigma \subcomp \tau}{A}}
    \pathcomp \big((\sigma \subcomp \tau) \whisker \delta\big)
    \pathcomp \invassociator
    \pathcomp \big((\squarehpaste{\cleavingcomm{\tau}{A \substT \sigma}}{\cleavingcomm{\sigma}{A}}) \whisker \pqsub\big)
    \pathcomp \associator.
  \]

  What remains, then, is to construct \(\epsilon_1\).
  By applying \(\cwfvarbeta\) and \(\substtcomp\) to reduce the generic terms on the left and right, we calculate that its type is equivalent to
  \[ \cwfvar_{A \substt {\sigma \subcomp \tau}} \transpover[\Tm*\Beta]{e_1}
    = \cwfvar_{A \substt {\sigma \subcomp \tau}} \transpover[\Tm*\Beta]{e_2}, \]
  where the left and right hand sides are transported, respectively, over equalities
  \[ e_1 \defeq \invsubstTcomp \pathcomp \eqsubsubstT{\pathinv\cwfprojbeta} \pathcomp \eqsubsubstT{\epsilon_0} \pathcomp \substTcomp \]
  and
  \[ e_2 \defeq
    \eqtypesubstT\substTcomp\cwfproj \pathcomp \eqsubsubstT{\pathinv\cwfprojbeta} \pathcomp \substTcomp
    \pathcomp \eqtypesubstT{\widetilde{e}\,}{\cwfproj \subext \cwfvar \transpover{\eqtypesubstT{\substTcomp}{\cwfproj}}} \pathcomp \invsubstTcomp, \]
  and where
  \[ \widetilde{e} \defeq
    \conjsubstTcomp{\eqsubsubstT{\pathinv\cwfprojbeta}}
    \pathcomp \eqtypesubstT{\big(\conjsubstTcomp{\eqsubsubstT{\pathinv\cwfprojbeta}}\big)}{\tlift}
    \pathcomp \invsubstTcomp.
  \]
  It's now enough to show that \(e_1 = e_2\).
  This amounts to giving a filling of \Cref{diag:substT-cleaving-equality},
  which we divide into three regions filled with coherence cells as shown in \Cref{diags:substT-cleaving-equality-1-2} and \Cref{diag:substT-cleaving-equality-3}.
  \input{diag-substT-cleaving-equality}
  \input{diags-substT-cleaving-equality-1-2}
  \input{diag-substT-cleaving-equality-3}
\end{proof}

\clearpage

Putting the previous results together, we have:

\begin{theorem}[Split comprehension for set-level and univalent cwfs]\label{cor:set-level-univalent-cwf-split}
  Any set-level or univalent 2-coherent wild cwf \(\thecwf\) is coherently split.
\end{theorem}
\begin{proof}
  By \Cref{cor:substT-cleaving-id} and \Cref{lem:substT-cleaving-equality}, the equality types
  \[ (\Gamma \ctxext {A \substT \id},\ \substTpb\id{A}) = (\Gamma \ctxext A,\ \idpb{\cwfproj_A}) \]
  and
  \[
    (\Beta \ctxext A \substT {\sigma \subcomp \tau},\ \substTpb{\sigma \subcomp \tau}{A})
  = (\Beta \ctxext A \substT \sigma \substT \tau,\ \squarehpaste{\substTpb{\tau}{A \substT\sigma}}{\substTpb{\sigma}{A}})
  \]
  are inhabited for all appropriately typed \(A\), \(\sigma\) and \(\tau\).
  By \Cref{prop:pullback-is-set-set-level-cats,prop:pullback-is-prop-univalent-wild-cats} they are propositions when \(\thecwf\) is set-level, or 2-coherently univalent.
\end{proof}


\begin{corollary}
  Thus, the syntax cwf QIIT of Altenkirch and Kaposi~\cite{altenkirch-kaposi:16:tt-in-tt} as well as any univalent universe cwf \(\UniverseCwf\) is split 2-coherent.
\end{corollary}

%% file: diag-type-subst-wk-pb.tex
\begin{diag}
  \caption{\(e \pathcomp \pathinv\substTcomp \pathcomp \eqsubsubstT{\epsilon_0} \pathcomp \substTcomp = \refl\).}\label{diag:type-subst-wk-pb}
  \newcommand{\subs}{\substT}
  \renewcommand{\comp}{\subcomp}
  \[\mathclap{\begin{tikzcd}[column sep=-1.75em,
    nodes={
      ampersand replacement=\&,
      font=\footnotesize,
      row 1/.style={row sep=1ex},
      row 4/.style={row sep=2em},
      row 5/.style={row sep=2em},
      row 8/.style={row sep=1ex}}
    ]
    \& [-1.75em] \&\& [0.75em] \&\& [3em] \& {A \subs \pA \subs \varrho} \& [3em] \&\& [0.75em] \&\& [-1.75em] \\
    \&\&\&\& {A \subs {\pA \comp \varrho}} \&\&\&\& {A \subs {\pA \comp \varrho}} \\
    \&\& {A \subs {\sigma \comp \tau}} \&\&\&\&\&\&\&\& {A \subs {\sigma \comp \tau}} \\
    {A \subs \sigma \subs \tau} \&\&\&\&\&\&\&\&\&\&\&\& {A \subs {\sigma \comp \pAsigma \comp \thegapmap}} \\
    |[xshift=-1ex]| {A \subs \sigma \subs \pAsigma \comp \thegapmap} \&\&\&\&\&\&\&\&\&\&\&\& |[xshift=1ex]| {A \subs {(\sigma \comp \pAsigma) \comp \thegapmap}} \\
    {A \subs \sigma \subs \pAsigma \subs \thegapmap} \&\&\&\&\&\&\&\&\&\&\&\& {A \subs {(\pA \comp \sigmaA) \comp \thegapmap}} \\
    \&\& {A \subs {\sigma \comp \pAsigma} \subs \thegapmap} \&\&\&\&\&\&\&\& {A \subs {\pA \comp \sigmaA \comp \thegapmap}} \\
    \&\&\&\& {A \subs {\pA \comp \sigmaA} \subs \thegapmap} \&\&\&\& {A \subs \pA \subs {\sigmaA \comp \thegapmap}} \\
    \&\&\&\&\&\& {A \subs \pA \subs \sigmaA \subs \thegapmap}
    \arrow["{\pathinv\substTcomp}"'{xshift=0.75ex,yshift=0.5ex}, from=1-7, to=2-5, path]
    \arrow["{\pathinv{\eqsubsubstT\gamma}}"'{xshift=-0.5ex,yshift=-0.25ex}, from=2-5, to=3-3, path]
    \arrow["\substTcomp"'{xshift=-0.25ex,yshift=0.5ex}, from=2-9, to=1-7, path]
    \arrow[equal, scaling nfold=3, from=3-3, to=3-11, ""{name=one,above}]
      \arrow[commutesstyle, "(1)" description, from=1-7, to=one]
    \arrow["\substTcomp"'{xshift=-.5ex,yshift=-0.5ex}, from=3-3, to=4-1, path]
    \arrow["{\eqsubsubstT\gamma}"'{xshift=0.75ex,yshift=-.5ex}, from=3-11, to=2-9, path]
    \arrow["{\pathinv{\eqsubsubstT\cwfprojbeta}}"'{xshift=-.5ex,yshift=-1ex}, from=4-1, to=5-1, path]
    \arrow["{\eqsubsubstT{\sigma \whisker \cwfprojbeta}}"'{xshift=1ex,yshift=-1.5ex}, from=4-13, to=3-11, path]
    \arrow["\substTcomp" description, from=4-13, to=5-1, path, ""{name=two,below}]
      \arrow[commutesstyle, "(2)"{description}, from=one, to=two]
    \arrow["\substTcomp"'{xshift=-1ex,yshift=1.75ex}, from=5-1, to=6-1, path]
    \arrow["{\eqsubsubstT\associator}"'{xshift=1ex,yshift=-1.75ex}, from=5-13, to=4-13, path]
    \arrow["\substTcomp" description, from=5-13, to=7-3, path, ""{name=three,above}]
      \arrow[commutesstyle, "(3)"{description,xshift=-1pt}, from=two, to=three]
    \arrow["{\eqtypesubstT{\pathinv\substTcomp}{\thegapmap}}"'{xshift=-1.5ex,yshift=1.5ex}, from=6-1, to=7-3, path]
    \arrow["{\eqsubsubstT{\cwfprojbeta \whisker \thegapmap}}"'{xshift=0.75ex,yshift=1.25ex}, from=6-13, to=5-13, path]
    \arrow["\substTcomp" description, from=6-13, to=8-5, path, ""{name=four,below}]
      \arrow[commutesstyle, "(4)"{description,xshift=-1ex}, from=three, to=four]
      \arrow[commutesstyle, "(5)"{description,xshift=-1em}, from=four, to=8-9]
    \arrow["{\eqtypesubstT{\pathinv{\eqsubsubstT\cwfprojbeta}}{\thegapmap}}"'{xshift=-2ex,yshift=1ex}, from=7-3, to=8-5, path]
    \arrow["{\pathinv{\eqsubsubstT\associator}}"'{xshift=1ex,yshift=1.25ex}, from=7-11, to=6-13, path]
    \arrow["{\eqtypesubstT{\substTcomp}{\thegapmap}}"'{xshift=-0.75ex,yshift=-0.25ex}, from=8-5, to=9-7, path]
    \arrow["{\pathinv\substTcomp}"'{xshift=1ex,yshift=0.75ex}, from=8-9, to=7-11, path]
    \arrow["{\pathinv\substTcomp}"'{xshift=0.75ex,yshift=-.25ex}, from=9-7, to=8-9, path]
  \end{tikzcd}}\]
\end{diag}

%% file: diag-substT-cleaving-equality.tex
\begin{diag}
  \caption{
	The pasting proof of \(e_1 = e_2\) splits into three regions, which are filled with the cells shown in \Cref{diags:substT-cleaving-equality-1-2} and \Cref{diag:substT-cleaving-equality-3}.
	We abbreviate the substitution \((\cwfproj \subext \cwfvar \transpover {\eqtypesubstT\substTcomp\cwfproj})\) by \(\varrho\)
	and the substitution \(\eqtypesubstT{\eqtypesubstT{(\invsubstTcomp \pathcomp \eqsubsubstT{\pathinv\cwfprojbeta} \pathcomp \substTcomp)}\tlift}\varrho\) by \(\xi\).
  }\label{diag:substT-cleaving-equality}
  \renewcommand{\comp}{\subcomp}
  \[\mathclap{\begin{tikzcd}[
    ampersand replacement=\&,
    paths,
    column sep=-4em,
    nodes={font=\footnotesize} ]
    \& [-1em] \& [2em] {A\substT{(\sigma \comp \tau) \comp \cwfproj}} \& [2em] \& [-1em] \\[-2em]
    \& |[xshift=-3em]| {A\substT{\sigma\comp\tau}\substT{\cwfproj}} \&\& |[xshift=3em]| {A \substT{\cwfproj \comp \stlift}} \\[-0.5ex]
	|[xshift=2.5em]| {A\substT\sigma\substT\tau\substT{\cwfproj}} \&\&\&\& |[xshift=-2.5em]| {A\substT{(\sigma \comp \tau) \comp \cwfproj}} \\
	|[xshift=1em]| {A\substT\sigma\substT\tau \substT{\cwfproj \comp \varrho}} \&\&\&\& |[xshift=-1em]| {A \substT {(\sigma \comp \tau) \comp \cwfproj \comp \varrho}} \\
	|[xshift=.5ex]| {A\substT\sigma\substT\tau\substT{\cwfproj}[\varrho]} \&\&\&\& |[xshift=-.5ex]| {A \substT {((\sigma \comp \tau) \comp \cwfproj) \comp \varrho}} \\
	{A\substT\sigma \substT {\tau \comp \cwfproj}[\varrho]} \&\&\&\& {A \substT {(\sigma \comp \tau \comp \cwfproj) \comp \varrho}} \\
	{A\substT\sigma \substT {\cwfproj \comp \tlift}[\varrho]} \&\&\&\& {A \substT {(\sigma \comp \cwfproj \comp \tlift) \comp \varrho}} \\
	|[xshift=.5ex]| {A\substT\sigma\substT\cwfproj\substT\tlift[\varrho]} \&\&\&\& |[xshift=-.5ex]| {A \substT {((\sigma \comp \cwfproj) \comp \tlift) \comp \varrho}} \\
	|[xshift=1em]| {A \substT{\cwfproj} \substT\slift \substT\tlift [\varrho]} \&\&\&\& |[xshift=-1em]| {A \substT {((\cwfproj \comp \slift) \comp \tlift) \comp \varrho}} \\
	|[xshift=3em]| {A \substT{\cwfproj} \substT{\slift\comp\tlift} [\varrho]} \&\&\&\& |[xshift=-3em]| {A \substT {(\cwfproj \comp \slift \comp \tlift) \comp \varrho}} \\[-1ex]
	\& |[xshift=-1.5em]| {A \substT{\cwfproj} \substT{(\slift\comp\tlift) \comp \varrho}} \&\& |[xshift=1.5em]| {A \substT{\cwfproj \comp (\slift \comp \tlift) \comp \varrho}}
	\arrow["{\invsubstTcomp}"{xshift=2ex,yshift=0ex}, from=2-2, to=1-3]
	\arrow["{\eqtypesubstT\substTcomp\cwfproj}"'{xshift=-.25ex,yshift=-1ex}, from=2-2, to=3-1]
	\arrow["{\eqsubsubstT{\pathinv\cwfprojbeta}}"{xshift=-.5ex,yshift=-.25ex}, from=1-3, to=2-4]
	\arrow["{\eqsubsubstT\cwfprojbeta}"{xshift=1ex,yshift=-2ex}, from=2-4, to=3-5]
	\arrow["{\eqsubsubstT{\pathinv\cwfprojbeta}}"'{yshift=-2ex}, from=3-1, to=4-1]
	\arrow["{\eqsubsubstT{(\sigma \comp \tau) \whisker \pathinv\cwfprojbeta}}"{xshift=.5ex,yshift=-2ex}, from=3-5, to=4-5]
	\arrow["{\substTcomp}"'{yshift=-2ex}, from=4-1, to=5-1]
	\arrow["{\eqsubsubstT{\invassociator}}"{yshift=-2ex}, from=4-5, to=5-5]
	\arrow["{\invsubstTcomp[\varrho]}"'{yshift=-2ex}, from=5-1, to=6-1]
	\arrow["{\eqsubsubstT{\associator\whisker\varrho}}"{yshift=-2ex}, from=5-5, to=6-5]
	\arrow["{\eqtypesubstT{\eqsubsubstT{\pathinv\cwfprojbeta}}\varrho}"', from=6-1, to=7-1]
	\arrow[""'{name=U}, ""{name=U'}, from=6-5, to=6-1]
	\arrow["{\eqsubsubstT{(\sigma \whisker \pathinv\cwfprojbeta) \whisker \varrho}}", from=6-5, to=7-5]
	\arrow["{\eqtypesubstT\substTcomp\varrho}"'{yshift=2ex}, from=7-1, to=8-1]
	\arrow["{\eqsubsubstT{\invassociator \whisker \varrho}}"{yshift=2.5ex}, from=7-5, to=8-5]
	\arrow["\xi"'{yshift=2ex}, from=8-1, to=9-1]
	\arrow[""'{name=V}, ""{name=V'}, from=8-5, to=8-1]
	\arrow["{\eqsubsubstT{(\pathinv\cwfprojbeta \whisker \tlift) \whisker \varrho}}"{yshift=2.5ex}, from=8-5, to=9-5]
	\arrow["{\eqtypesubstT\invsubstTcomp\varrho}"'{xshift=-.25ex,yshift=2.25ex}, from=9-1, to=10-1]
	\arrow["{\eqsubsubstT{\associator \whisker \varrho}}"{xshift=.5ex,yshift=2ex}, from=9-5, to=10-5]
	\arrow["{\invsubstTcomp}"'{xshift=-1.25ex,yshift=1.75ex}, from=10-1, to=11-2]
	\arrow["{\eqsubsubstT{\associator}}"{xshift=.75ex,yshift=1.5ex}, from=10-5, to=11-4]
	\arrow["{\substTcomp}", from=11-4, to=11-2, ""'{name=W}]
	\arrow[commutesstyle, "\text{(I)}"{description,yshift=-1ex}, from=1-3, to=U]
	\arrow[commutesstyle, "\text{(II)}"{description}, from=U', to=V]
	\arrow[commutesstyle, "\text{(III)}"{description}, from=V', to=W]
  \end{tikzcd}}\]
\end{diag}

%% file: diags-substT-cleaving-equality-1-2.tex
\begin{diag}
  \renewcommand{\comp}{\subcomp}
  \newcommand{\pp}{\cwfproj}
  \caption[]{
    Filling regions (I) and (II) of \Cref{diag:substT-cleaving-equality} with coherence cells.
    Regions marked (1) are filled using type pentagonators (\Cref{def:type-pentagonators}); those marked (2), by naturality of \(\substTcomp\) (\Cref{lem:substTcomp-nat-iso}); (3), by associativity of whiskering (\Cref{prop:whiskering-properties}); and (4), by the pentagon associator of the category of contexts (\Cref{def:type-pentagonators}).
  }\label{diags:substT-cleaving-equality-1-2}

  \input{diag-substT-cleaving-equality-1}

  \input{diag-substT-cleaving-equality-2}
\end{diag}

%% file: diag-substT-cleaving-equality-1.tex
\begin{subfigure}{\textwidth}
  \[\mathclap{\begin{tikzcd}[
	ampersand replacement=\&,
    column sep=-4em,
    nodes={font=\footnotesize} ]
	\& [-1ex] \& [-0.5ex] \& [1.5em] {A \substT{(\sigma \comp \tau) \comp \pp}} \& [1.5em] \& [-1ex] \& [-1.5ex] \& \\[-2em]
	\&\& |[xshift=-3.5em]| {A \substT{\sigma \comp \tau} \substT{\pp}} \&\& |[xshift=3.5em]| {A \substT{\pp \comp \stlift}} \\[-.5ex]
	\& |[xshift=-5em]| {A \substT{\sigma} \substT{\tau} \substT{\pp}} \&\&\&\& |[xshift=5em]| {A \substT{ (\sigma \comp \tau) \comp \pp }} \\[-1.5em]
	\&\& {A \substT\sigma \substT{\tau \comp \pp}} \&\& {A \substT{\sigma \comp \tau \comp \pp}} \\[-1.5em]
	|[xshift=2em]| {A \substT\sigma \substT\tau \substT{\pp \comp \varrho}} \&\&\&\&\&\& |[xshift=-2em]| {A \substT{ (\sigma \comp \tau) \comp \pp \comp \varrho }} \\[-1em]
	\&\& {A \substT\sigma \substT{\tau \comp \pp \comp \varrho}} \&\& {A \substT{\sigma \comp \tau \comp \pp \comp \varrho}} \\[-2.25em]
	|[xshift=1.25em]| {A \substT\sigma \substT\tau \substT\pp \substT\varrho} \&\&\&\&\&\& |[xshift=-1.25em]| {A \substT{((\sigma \comp \tau) \comp \pp) \comp \varrho}} \\
	|[xshift=1em]| {A \substT\sigma \substT{\tau \comp \pp} \substT\varrho} \&\& {A \substT\sigma \substT{(\tau \comp \pp) \comp \varrho}} \&\& {A \substT{\sigma \comp (\tau \comp \pp) \comp \varrho}} \&\& |[xshift=-1em]| {A \substT{(\sigma \comp \tau \comp \pp) \comp \varrho}}
	\arrow["{\eqsubsubstT{ \pathinv\cwfprojbeta }}"{yshift=.5ex}, from=1-4, to=2-5, path]
	\arrow[equal, scaling nfold=3, bend right=5, from=1-4, to=3-6]
	\arrow["{{\invsubstTcomp}}"{xshift=1ex,yshift=.5ex}, from=2-3, to=1-4, path]
	\arrow["{{\eqtypesubstT\substTcomp\pp}}"'{xshift=-1ex,yshift=-.5ex}, from=2-3, to=3-2, path]
	\arrow["{\eqsubsubstT{ \cwfprojbeta }}"{xshift=2ex,yshift=-1.25ex}, from=2-5, to=3-6, path]
	\arrow["{\eqsubsubstT{ \pathinv\cwfprojbeta }}"'{xshift=-1ex,yshift=-1ex}, from=3-2, to=5-1, path]
	\arrow["{\eqsubsubstT{ (\sigma \comp \tau) \whisker \pathinv\cwfprojbeta }}"{xshift=1ex,yshift=-1.5ex}, from=3-6, to=5-7, path]
	\arrow["\substTcomp"'{yshift=.5ex}, from=4-3, to=3-2, path, ""{name=V}]
	\arrow["\eqsubsubstT\associator"'{yshift=.25ex}, from=3-6, to=4-5, path,""{name=W}]
	\arrow["\substTcomp"'{yshift=.75ex,name=U}, from=4-5, to=4-3, path, ""{name=U'}]
	\arrow["{{\substTcomp}}"'{xshift=-1ex,yshift=-1ex}, from=5-1, to=7-1, path]
	\arrow["{\eqsubsubstT{ \invassociator }}"{xshift=1ex,yshift=-1.25ex}, from=5-7, to=7-7, path]
	\arrow["\eqsubsubstT{\tau \whisker \pathinv\cwfprojbeta}"{description, yshift=4pt}, from=4-3, to=6-3, pathnearend]
	\arrow["\eqsubsubstT{\sigma \whisker (\tau \whisker \pathinv\cwfprojbeta)}"{description, yshift=4pt}, from=4-5, to=6-5, pathnearend]
	\arrow["\substTcomp"{yshift=-.5ex,name=V''}, from=6-3, to=5-1, path, ""'{name=V'}]
	\arrow["\eqsubsubstT\associator"{yshift=-.5ex,name=W''}, from=5-7, to=6-5, path, ""'{name=W'}]
	\arrow["\substTcomp"{yshift=-.75ex,name=U'''}, from=6-5, to=6-3, path, ""'{name=U''}]
	\arrow["{{\eqtypesubstT\invsubstTcomp\varrho}}"'{xshift=-.5ex,yshift=-1.5ex}, from=7-1, to=8-1, path]
	\arrow["{\eqsubsubstT{ \associator \whisker \varrho }}"{xshift=1ex,yshift=-1.5ex}, from=7-7, to=8-7, path]
	\arrow["\eqsubsubstT\associator"{description, yshift=-3pt}, from=8-3, to=6-3, pathnearend]
	\arrow["{{\substTcomp}}"{yshift=-.75ex}, from=8-3, to=8-1, path,""'{name=V'''}]
	\arrow["\eqsubsubstT{\sigma \whisker \associator}"{description, yshift=-3pt}, from=8-5, to=6-5, pathnearend]
	\arrow["{{\substTcomp}}"{yshift=-.75ex}, from=8-5, to=8-3, path, ""'{name=U''''}]
	\arrow["{\eqsubsubstT{ \associator }}"{yshift=-.75ex}, from=8-7, to=8-5, path,""'{name=W'''}]
	\arrow[commutesstyle, "(1)"description, from=1-4, to=U]
	\arrow[commutesstyle, "(2)"description, from=U', to=U'']
	\arrow[commutesstyle, "(2)"description, from=U''', to=U'''']
	\arrow[commutesstyle, "(2)"description, from=V, to=V']
	\arrow[commutesstyle, "(1)"description, from=V'', to=V''']
	\arrow[commutesstyle, "(3)"description, from=W, to=W']
	\arrow[commutesstyle, "(4)"description, from=W'', to=W''']
  \end{tikzcd}}\]
  \caption*{Region (I)}
\end{subfigure}

%% file: diag-substT-cleaving-equality-2.tex
\begin{subfigure}{\textwidth}
  \[\mathclap{\begin{tikzcd}[
    ampersand replacement=\&,
    column sep=-3.75em,
    nodes={font=\footnotesize} ]
    {A \substT\sigma \substT{\tau \comp \pp} \substT\varrho} \& [4.75em] {A \substT\sigma \substT{(\tau \comp \pp) \comp \varrho}} \&\& {A \substT{\sigma \comp (\tau \comp \pp) \comp \varrho}} \& [4.75em] {A \substT{(\sigma \comp \tau \comp \pp) \comp \varrho}} \\[1em]
    \& {A \substT\sigma \substT{(\pp \comp \tlift) \comp \varrho}} \&\& {A \substT{\sigma \comp (\pp \comp \tlift) \comp \varrho}} \\[-1.5em]
    {A \substT\sigma \substT{\pp \comp \tlift} \substT\varrho} \&\&\&\& {A \substT{(\sigma \comp \pp \comp \tlift) \comp \varrho}} \\
    \&\& {A \substT{\sigma \comp \pp \comp \tlift} \substT\varrho} \\
    {A \substT\sigma \substT\pp \substT\tlift \substT\varrho} \& {A \substT{\sigma \comp \pp} \substT\tlift \substT\varrho} \&\& {A \substT{(\sigma \comp \pp) \comp \tlift} \substT\varrho} \& {A \substT{((\sigma \comp \pp) \comp \tlift) \comp \varrho}}
    \arrow["{\eqtypesubstT{\pathinv\cwfprojbeta}\varrho}"', from=1-1, to=3-1,path]
    \arrow["\substTcomp"'{yshift=.5ex}, from=1-2, to=1-1,path,""{name=U}]
    \arrow["{\eqsubsubstT{\pathinv\cwfprojbeta \whisker \varrho}}"{description,yshift=3pt}, from=1-2, to=2-2,pathnearend]
    \arrow["\substTcomp"'{yshift=.5ex}, from=1-4, to=1-2,path,""{name=V}]
    \arrow["{\eqsubsubstT{\sigma \whisker (\pathinv\cwfprojbeta \whisker \varrho)}}"{description,yshift=3pt}, from=1-4, to=2-4,pathnearend]
    \arrow["{\eqsubsubstT\associator}"'{yshift=.5ex}, from=1-5, to=1-4,path,""{name=W}]
    \arrow["{\eqsubsubstT{(\sigma \whisker \pathinv\cwfprojbeta) \whisker \varrho}}", from=1-5, to=3-5,path]
    \arrow["\substTcomp"{yshift=-.5ex}, from=2-2, to=3-1,path,""'{name=U'}]
    \arrow["\substTcomp"{yshift=-.5ex,name=V''}, from=2-4, to=2-2,path,""'{name=V'}]
    \arrow["{\eqtypesubstT\substTcomp\varrho}"', from=3-1, to=5-1,path]
    \arrow["{\eqsubsubstT\associator}"{xshift=1ex,yshift=-.5ex}, from=3-5, to=2-4,path,""'{name=W'}]
    \arrow["\substTcomp"{name=X}, from=3-5, to=4-3,path]
    \arrow["{\eqsubsubstT{\invassociator \whisker \varrho}}", from=3-5, to=5-5,path]
    \arrow["{\eqtypesubstT\substTcomp\varrho}"{name=U''}, from=4-3, to=3-1,path]
    \arrow["{\eqtypesubstT{\eqsubsubstT\invassociator}\varrho}"{xshift=1ex,yshift=-1ex}, from=4-3, to=5-4,path]
    \arrow["{\eqtypesubstT{\eqtypesubstT\substTcomp\tlift}\varrho}"{yshift=-1ex}, from=5-2, to=5-1,path]
    \arrow["{\eqtypesubstT\substTcomp\varrho}"{yshift=-1ex}, from=5-4, to=5-2,path]
    \arrow["\substTcomp"{yshift=-1ex}, from=5-5, to=5-4,path,""'{name=X'}]
    \arrow[commutesstyle,from=U,to=U',"(2)"description]
    \arrow[commutesstyle,from=U'',to=5-2,"(1)"description]
    \arrow[commutesstyle,from=V,to=V',"(2)"description]
    \arrow[commutesstyle,from=V'',to=4-3,"(1)"description]
    \arrow[commutesstyle,from=W,to=W',"(3)"description]
    \arrow[commutesstyle,from=X,to=X',"(2)"description]
  \end{tikzcd}}\]
  \caption*{Region (II)}
\end{subfigure}

%% file: diag-substT-cleaving-equality-3.tex
\begin{diag}
  \renewcommand{\comp}{\subcomp}
  \newcommand{\pp}{\cwfproj}
  \caption[]{
    Filling region (III) of \Cref{diag:substT-cleaving-equality} with coherence cells.
	Regions marked (1) are filled using type pentagonators (\Cref{def:type-pentagonators}); regions marked (2), by naturality of \(\substTcomp\) (\Cref{lem:substTcomp-nat-iso}).
  }\label{diag:substT-cleaving-equality-3}
  \[\mathclap{\begin{tikzcd}[
    ampersand replacement=\&,
    column sep=-1em,
    nodes={font=\footnotesize}
	]
  	{A \substT\sigma \substT\pp \substT\tlift \substT\varrho} \& [-3em] {A \substT{\sigma \comp \pp} \substT\tlift \substT\varrho} \& [.5em] {A \substT{(\sigma \comp \pp) \comp \tlift} \substT\varrho} \& [-2.75em] {A \substT{((\sigma \comp \pp) \comp \tlift) \comp \varrho}} \\
	{A \substT{\sigma \comp \pp} \substT\tlift \substT\varrho} \&\& |[xshift=-1em]| {A \substT{(\pp \comp \slift) \comp \tlift} \substT\varrho} \& {A \substT{((\pp \comp \slift) \comp \tlift) \comp \varrho}} \\
	|[xshift=.5em]| {A \substT{\pp \comp \slift} \substT\tlift \substT\varrho} \&\& |[xshift=-1em]| {A \substT{\pp \comp \slift \comp \tlift} \substT\varrho} \& |[xshift=-.5em]| {A \substT{(\pp \comp \slift \comp \tlift) \comp \varrho}} \\
	|[xshift=2em]| {A \substT\pp \substT\slift \substT\tlift \substT\varrho} \&\&\& |[xshift=-2em]| {A \substT{\pp \comp (\slift \comp \tlift) \comp \varrho}} \\
	\& |[xshift=-2em]| {A \substT\pp \substT{\slift \comp \tlift} \substT\varrho} \& |[xshift=2em]| {A \substT\pp \substT{(\slift \comp \tlift) \comp \varrho}}
	\arrow["{\eqtypesubstT{\eqtypesubstT\invsubstTcomp\tlift}\varrho}"', from=1-1, to=2-1, path]
	\arrow["{\eqtypesubstT{\eqtypesubstT{\substTcomp}\tlift}\varrho}"'{yshift=1ex}, from=1-2, to=1-1, path]
	\arrow[equal, scaling nfold=3, from=1-2, to=2-1, start anchor={[xshift=2ex]},""{name=U}]
	\arrow["{\eqtypesubstT\substTcomp\varrho}"'{yshift=1ex}, from=1-3, to=1-2, path]
	\arrow["{\eqtypesubstT{\eqsubsubstT{\pathinv\cwfprojbeta \whisker \tlift}}\varrho}"', from=1-3, to=2-3, path, start anchor={[xshift=-.5em]}]
	\arrow["\substTcomp"'{yshift=.5ex}, from=1-4, to=1-3, path,""{name=W}]
	\arrow["{\eqsubsubstT{(\pathinv\cwfprojbeta \whisker \tlift) \whisker \varrho}}", from=1-4, to=2-4, path]
	\arrow["{\eqtypesubstT{\eqtypesubstT{\pathinv\cwfprojbeta}\tlift}\varrho}"'{xshift=-.5ex, yshift=1.5ex}, from=2-1, to=3-1, path]
	\arrow["{\eqtypesubstT\substTcomp\varrho}"{name=V}, from=2-3, to=3-1, path, start anchor={[xshift=-1.5ex, yshift=1.5ex]}, end anchor={[yshift=-.75ex]},""{name=U'}]
	\arrow["{\eqtypesubstT{\eqsubsubstT\associator}\varrho}"', from=2-3, to=3-3, path]
	\arrow["\substTcomp"'{yshift=.5ex,name=W'}, from=2-4, to=2-3, path,""{name=X}]
	\arrow["{\eqsubsubstT{\associator \whisker \varrho}}"{xshift=.5ex,yshift=1.25ex}, from=2-4, to=3-4, path]
	\arrow["{\eqtypesubstT{\eqtypesubstT\substTcomp\tlift}\varrho}"'{xshift=-.5ex,yshift=1.5ex}, from=3-1, to=4-1, path]
	\arrow["{\eqtypesubstT\substTcomp\varrho}"{name=Y}, from=3-3, to=5-2, path]
	\arrow["\substTcomp"{yshift=-.5ex}, from=3-4, to=3-3, path,""'{name=X'}]
	\arrow["{\eqsubsubstT{\associator}}"{xshift=.5ex,yshift=1.25ex}, from=3-4, to=4-4, path]
	\arrow["{\eqtypesubstT\invsubstTcomp\varrho}"'{xshift=-1ex,yshift=.75ex}, from=4-1, to=5-2, path]
	\arrow["\substTcomp", from=4-4, to=5-3, path]
	\arrow["\invsubstTcomp"'{yshift=-.5ex}, from=5-2, to=5-3, path]
	\arrow[commutesstyle, from=U, to=U', "(2)"{description}]
	\arrow[commutesstyle, from=V, to=5-2, "(1)"']
	\arrow[commutesstyle, from=W, to=W', "(2)"{description}]
	\arrow[commutesstyle, from=X, to=X', "(2)"{description}]
	\arrow[commutesstyle, from=3-3, to=5-3, "(1)"{description}]
  \end{tikzcd}}\]
\end{diag}

%% file: discussion.tex
\section{Discussion}\label{sec:discussion}

We have given a unified account of the cloven fibrational structure of 2-coherent internal models of homotopical dependent type theory, in such a way so as to include set-level models such as the syntax as well as the higher models given by universe types.

Now, the generalized algebraic presentation of 2-coherent wild cwfs straightforwardly yields an internally definable type of \emph{morphisms} of such, which sets up the possibility of internally studying ``transfers'' of constructions
between internal models; in particular, from the syntax to a universe type.
This is a large part of the motivation of the present work---specifically, the theory developed here is intended to provide a precise formal setting in which to investigate
\begin{enumerate*}[label=(\arabic*)]
  \item the construction of classifiers of semisimplicial and other Reedy fibrant inverse diagrams~\cite{kraus-sattler:17:space-diagrams} in (internal models of) homotopical type theory~\cite{chen-kraus:21:internal,chen-kraus:24:inverse-diagrams}, and
  \item the relation of this problem to that of the self-interpretation of \HoTT~\cite{shulman:14:hott-should-eat-itself}.
\end{enumerate*}

Separately from questions of infinite higher coherent constructions, we hope that our theory can still be useful by immediately specializing to yield notions of 1-truncated ``2-cwfs''.
For instance, by modifying the definition of a 2-coherent wild cwf to additionally require that
\begin{enumerate*}[label=(\arabic*)]
  \item the category of contexts \(\C\) is a precategory,
  \item the presheaf of \(\C\)-types is valued in 1-types, and
  \item the presheaf of \(\C\)-terms is set-valued,
\end{enumerate*}
we obtain a simple higher generalization of the notion of a 1-cwf, which conjecturally includes the container higher model of type theory~\cite{altenkirch-kaposi:21:container-model} as an instance.
Then via Rezk completion and our results, any instance of such a higher cwf should be equivalent to a split one.
A less na\"{i}ve approach would be to use univalent bicategory theory~\cite{ahrens+:21:bicat-univalent-foundations}, noting \Cref{prop:2-coh-wild-cat-prebicat}, to develop a full theory of 2-cwfs.

As a final remark, we have developed wild categories with families for their anticipated applicability to specific further internal constructions, but we also expect the study of wild natural models and wild comprehension categories to prove complementarily fruitful.